\documentclass{fac}

\usepackage{amsmath,amssymb}
\usepackage{pvs}
\usepackage{pstricks,pst-node,pstcol}
\usepackage{setspace}
\usepackage{ifthen}
\usepackage{graphicx}

\newcommand{\s}{\mbox{\ensuremath{y \!=\! a + b}}}

\newgray{verylightgray}{0.89}
\newcommand{\ctlbp}{\ensuremath{\text{CTL}_\text{bp}}}

\newcommand{\AX}{\scalebox{1}{\pvsid{AX}}}
\newcommand{\EX}{\scalebox{1}{\pvsid{EX}}}
\newcommand{\AF}{\scalebox{1}{\pvsid{AF}}}
\newcommand{\EF}{\scalebox{1}{\pvsid{EF}}}
\newcommand{\AW}{\scalebox{1}{\pvsid{AW}}}
\newcommand{\EW}{\scalebox{1}{\pvsid{EW}}}
\newcommand{\AG}{\scalebox{1}{\pvsid{AG}}}
\newcommand{\EG}{\scalebox{1}{\pvsid{EG}}}
\newcommand{\AU}{\scalebox{1}{\pvsid{AU}}}
\newcommand{\EU}{\scalebox{1}{\pvsid{EU}}}
\newcommand{\AY}{\scalebox{1}{\pvsid{AY}}}
\newcommand{\EY}{\scalebox{1}{\pvsid{EY}}}
\newcommand{\AP}{\scalebox{1}{\pvsid{AP}}}
\newcommand{\EP}{\scalebox{1}{\pvsid{EP}}}
\newcommand{\AS}{\scalebox{1}{\pvsid{AS}}}
\newcommand{\ES}{\scalebox{1}{\pvsid{ES}}}
\newcommand{\AH}{\scalebox{1}{\pvsid{AH}}}
\newcommand{\EH}{\scalebox{1}{\pvsid{EH}}}
\newcommand{\sthat}{\ensuremath{\; | \;}}
\newcommand{\veczero}{\ensuremath{\vec{0}}}
\newcommand{\vecone}{\ensuremath{\vec{1}}}
\newcommand{\defas}{\ensuremath{\stackrel{\triangle}{=}}}

\newcommand{\bool}{\ensuremath{\mathcal{B}}}
\newcommand{\one}{\ensuremath{{1}}}
\newcommand{\zero}{\ensuremath{{0}}}
\newcommand{\mystates}{\ensuremath{{N}}}

\newcommand{\myprops}{\ensuremath{{P}}}
\newcommand{\mylabeling}{\ensuremath{{L}}}
\newcommand{\mykripke}{\ensuremath{{M}}}
\newcommand{\zerom}{\ensuremath{\mathbf{\zero}}}
\newcommand{\onem}{\ensuremath{\mathbf{\one}}}
\newcommand{\entry}{\ensuremath{\vec{\varepsilon}}}
\newcommand{\exit}{\ensuremath{\vec{\omega}}}

\newcommand{\pair}[2]{\ensuremath{\langle #1,#2 \rangle}}

\newcommand{\statemap}{\ensuremath{{C}}}
\newcommand{\mat}{\ensuremath{{A}}}

\newcommand{\Mm}{\mbox{\ensuremath{\cdot}}}
\newcommand{\Ng}[1]{\ensuremath{\overline{{#1}}}}
\newcommand{\timp}{\ensuremath{\rightarrow}}
\newcommand{\graph}{\ensuremath{G}}
\newcommand{\pnl}{\\[0mm]}
\newcommand{\Tr}[1]{\ensuremath{\widehat{#1}}}
\newcommand{\Pl}{\mbox{\ensuremath{+}}}
\newcommand{\Mn}{\mbox{\ensuremath{-}}}
\newcommand{\St}{\mbox{\ensuremath{*}}}
\newcommand{\Tm}{\mbox{\ensuremath{\times}}}
\newcommand{\Imp}{\mbox{{\ensuremath{\;\supset\;}}}}
\newcommand{\delstates}{\delta}

\newcommand{\newstates}{\ensuremath{\eta'}}

\newcommand{\node}[5]
{\rput(#1,#2){\rnode{n#3}{\psframebox[fillcolor=white,fillstyle=solid]{\phantom{\!#4\!}}}}
\rput(#1,#2){#5}
\rput(#1,#2){
\ifthenelse{#3>9}{\rput(-11.5,0){#3}}{\rput(-9.5,0){#3}}
}}

\newcommand{\myrnode}[5]
{\rput(#1,#2){\rnode{n#3}{\psframebox[fillcolor=white,fillstyle=solid]{\phantom{#4}}}}
\rput(#1,#2){#5}
\rput(#1,#2){
\ifthenelse{#3>9}{\rput(10,0){#3}}{\rput(8,0){#3}}
}}

\newcommand{\tnode}[5]
{\rput(#1,#2){\rnode{m#3}{\psframebox{\phantom{#4}}}}
\rput(#1,#2){#5}
\rput(#1,#2){
\ifthenelse{#3>9}{\rput(-11,0){#3}}{\rput(-7,0){#3}}
}}

\newcommand{\lnode}[5]
{\rput(#1,#2){\rnode{n#3}{\psframebox[fillstyle=solid,fillcolor=verylightgray]{\phantom{\!#4\!}}}}
\rput(#1,#2){#5}
\rput(#1,#2){
\ifthenelse{#3>9}{\rput(-11.5,0){#3}}{\rput(-9.5,0){#3}}
}}

\newcommand{\dnode}[5]
{\rput(#1,#2){\rnode{n#3}{\psframebox[fillcolor=lightgray,fillstyle=solid]{\phantom{{#4}}}}}
\rput(#1,#2){{\small #5}}
\rput(#1,#2){
\ifthenelse{#3>9}{\rput(-8,0){#3}}{\rput(-7,0){#3}}
}}

\newcommand{\tbnode}[5]
{\rput(#1,#2){\rnode{m#3}{\psframebox[fillstyle=solid,fillcolor=lightgray]{\phantom{#4}}}}
\rput(#1,#2){#5}
\rput(#1,#2){
\ifthenelse{#3>9}{\rput(-13,0){#3}}{\rput(-7.5,0){#3}}
}}

\newcommand{\tlnode}[5]
{\rput(#1,#2){\rnode{m#3}{\psframebox[fillstyle=solid,fillcolor=verylightgray]{\phantom{#4}}}}
\rput(#1,#2){#5}
\rput(#1,#2){
\ifthenelse{#3>9}{\rput(-13,0){#3}}{\rput(-7.5,0){#3}}
}}

\newcommand{\newnode}[5]
{\rput(#1,#2){\rnode{n#3}{\psframebox[linestyle=dashed,fillstyle=solid,fillcolor=verylightgray]{\phantom{\!#4\!}}}}
\rput(#1,#2){#5}
\rput(#1,#2){
\ifthenelse{#3>9}{\rput(-10.5,0){#3}}{\rput(-9.5,0){#3}}
}}

\newlength{\minipagewidth}
\newlength{\shortpagewidth}
\newcommand{\Kfsim}{\mbox{\ensuremath{\triangleright}}}
\newcommand{\Kfbsim}{\mbox{\ensuremath{\bowtie}}}
\newcommand{\Kwbisim}{\mbox{\ensuremath{\preceq}}}
\newcommand{\comment}[2]{\ifthenelse{#1=1}{\ensuremath{\text{#2}}}{}}
\newcommand{\mycomment}[3]{\ifthenelse{#1=1}{\ensuremath{\text{#2}}}{\ensuremath{\text{#3}}}}
\renewcommand{\pnl}{\\[0mm]}

\newcommand{\specification}[3]{
\begin{figure}
\centering
\fbox{
\centering
\begin{minipage}{\minipagewidth}
{
{#1}}
\end{minipage}}
\caption{{#2}}
\label{fig:#3}
\end{figure}
}

\newlength{\smallpagewidth}
\newtheorem{definition}{Definition}
\newtheorem{lemma}[definition]{Lemma}
\newtheorem{theorem}[definition]{Theorem}

\title[A Logic for Correlating Temporal Properties across Program Transformations]
{A Logic for Correlating Temporal Properties across Program Transformations}
\author[Kanade, Sanyal, Khedker]
{Aditya Kanade$^1$, Amitabha Sanyal$^2$, and Uday P. Khedker$^2$\\
$^1$Indian Institute of Science, 
$^2$Indian Institute of Technology Bombay}
\correspond{Aditya Kanade. e-mail: kanade@csa.iisc.ernet.in}
\pagerange{\pageref{firstpage}--\pageref{lastpage}}

\begin{document}

\makecorrespond
\maketitle

\begin{abstract}
Program transformations are widely used in synthesis, optimization, and maintenance of software.
Correctness of program transformations depends on preservation of some important properties of the input program.
By regarding programs as Kripke structures,
many interesting properties of programs can be expressed in temporal logics.
In temporal logic, a formula is interpreted on a single program.
However, to prove correctness of transformations, we encounter formulae
which contain some subformulae interpreted on the input program and some on the transformed program.
An example where such a situation arises is verification of optimizing program transformations applied by compilers.

In this paper, we present a logic called Temporal Transformation Logic (TTL) to reason about such formulae.
We consider different types of primitive transformations and present TTL inference rules for them.
Our definitions of program transformations and temporal logic operators are novel in their use of the boolean matrix algebra.
This results in specifications that are succinct and constructive. Further, we use
the boolean matrix algebra in a uniform manner to prove soundness of the TTL inference rules.
\end{abstract}



\section{Introduction}
\label{sec:introduction}

Program transformations are widely used in synthesis, optimization, and maintenance of software.
Correctness of program transformations depends on preservation of
some important properties of the input program.
By regarding programs as Kripke structures,
many interesting properties of programs can be expressed in temporal logics~\cite{temporal_logic_of_programs}.
In temporal logic, a formula is interpreted on a single program.
However, to prove correctness of transformations,
we encounter formulae which contain some subformulae interpreted on the input program and some on the transformed program.
The proof systems for temporal logics~\cite{DBLP:journals/tcs/MannaP91,deductive:CTL} are therefore
not sufficient to prove correlations of temporal properties across program transformations.
In this paper, we present a logic called Temporal Transformation Logic (TTL) to reason about such formulae.

Our study of this logic is motivated by an interesting application in the area of compiler verification.
A compiler optimizer analyzes a program and identifies potential performance improvements.
An optimized version of the program is then obtained by applying several transformations
to the input program.
However, a mistake in the design of an optimizer can proliferate in the form of bugs
in the softwares compiled through it.
It is therefore important to verify whether an optimization routine preserves program semantics. 
Proving semantic equivalence of programs is usually tedious.
The complexity of proofs can be conquered by taking advantage of the fact that
optimizations with similar objectives employ similar program transformations.
This observation lead to identification of transformation primitives and 
their soundness conditions~\cite{DBLP:journals/entcs/KanadeSK07}.

A transformation primitive denotes a small-step program transformation that is used in many optimizing transformations.
These primitives can be used to specify a large class of optimizations by sequential composition.
The soundness condition for a transformation primitive is a condition on the input programs
to the primitive which if satisfied implies that the transformed program is semantically equivalent to the input program.
This approach reduces proving soundness of an optimization
to merely showing that soundness conditions of the underlying primitives are
satisfied on the versions of the input program on which they are applied.
This is much simpler than directly proving semantics preservation for each optimization.

Program analyses that are used for identifying optimization opportunities
and the soundness conditions of the transformation primitives are global dataflow properties.
These properties can be either forward,
backward, or more generally bidirectional~\cite{khedker:bidirectional}.
Further, these can be classified as may or must properties depending
on whether they are defined in terms of some or all paths originating 
(for forward analyses) or terminating (for backward analyses) at a node.
These properties can be
expressed naturally in a temporal logic called computational tree logic with branching past ($\ctlbp$)~\cite{once_and_for_all}. 
$\ctlbp$ formulae are state formulae and thus 
application points for program transformations can be identified by models of $\ctlbp$ formulae.

A program analysis used in an optimization is interpreted on the input program.
However, the input program is transformed step-by-step and hence
the soundness condition of a transformation primitive is interpreted on the version of the input program to which it is applied.
Thus to prove soundness of an optimization,
we need a logic to correlate temporal properties of programs 
across transformations. Program transformations are used widely in various software engineering
activities like program synthesis, refactoring, software renovation, and reverse engineering~\cite{DBLP:journals/tcs/Visser01}.
Temporal logic is known to be a powerful language for specifying properties of program executions. 
We therefore believe that TTL will also be useful in verification problems arising in these domains.

A program transformation may involve changes to
the control flow graph (structural transformation) and to the program statements (content transformation).
In general, if we consider arbitrary big-step transformations resulting in diverse changes to the input program,
we may not be able to correlate any interesting temporal properties across them.
We therefore present inference rules for the (small-step) transformation primitives. 
A larger transformation is expressed as a sequential composition
of such primitives. A correlation between temporal properties of its input and output
programs can be inferred by using the rules for the component transformations
on intermediate versions of the program.

The semantics of temporal operators is usually defined in a relational setting
by considering transitive closure (paths) of the transition relation (edges) of a Kripke structure.
We make novel use of the boolean matrix algebra to define temporal logic operators and also the primitive program transformations.
The definitions are succinct and facilitate algebraic proofs of soundness of the TTL inference rules.
Further, the definitions are constructive and can be evaluated directly.
The choice of boolean matrix algebra thus enables both verification
and translation validation of compiler optimizations~\cite{sefm06-KSK,DBLP:journals/entcs/KanadeSK07,gcc-validation}.

\subsection*{Related work}
Temporal logic has been used for specifying dataflow properties.
In ~\cite{schmidt:abstract,schmidt_steffen:abstract},
Schmidt et al. propose that static analysis of a program can be viewed as model checking of
an abstraction of the program, with the properties of interest specified in a suitable temporal logic.
In the recent literature on verification of compiler optimizations~\cite{lacey:proving,frederiksen:correctness,lerner:correctness}, 
conditional rewrites are proposed as a means for specifying optimizing transformations.
The enabling condition of a rewrite is defined in a temporal logic.
However, unlike our approach~\cite{DBLP:journals/entcs/KanadeSK07}, 
they do not use the temporal logic in verification.
The verification is performed by proving semantic equivalence of programs by induction on length of program executions.

Kripke structures serve as a natural modeling paradigm for system specification
when the properties of interest are temporal in nature.
Simulation relations~\cite{milner-algebraic-simulation,DBLP:conf/fsttcs/Namjoshi97} make it possible
to correlate temporal properties between Kripke structures and are used 
to show refinement or equivalence of system specifications.
Transformations of Kripke structures arise in state space reduction 
techniques~\cite{park:81,DBLP:conf/charme/FislerV99,735331,DBLP:conf/focs/HenzingerHK95} 
where a Kripke structure is transformed to a (bi)similar Kripke structure with less number of states.
Model checking certain formulae over the reduced Kripke structure is 
equivalent to model checking the formulae on the original Kripke structure.
In order to prove soundness of the TTL inference rules, we also construct
simulation relations between programs (seen as Kripke structures).
However, unlike the usual set-based formulation of simulation relations, 
we present boolean matrix algebraic formulation. This enables us to prove correlations
of temporal properties between programs (Kripke structures) in a completely novel boolean matrix algebraic setting.

Boolean matrix algebra or more generally, modal algebra, 
has been used for defining semantics of modal logic operators~\cite{DBLP:journals/jsyml/Thomason72,modal-logic}.
The boolean algebra is used for modeling propositional logic 
and the boolean matrix algebraic operators are used for capturing different modalities.
In~\cite{fitting:BisBool}, Fitting develops a modal algebraic formulation of bisimulation
in modal and multi-modal setting.
Our definitions correspond to frame bisimulations in~\cite{fitting:BisBool}.
However, our approach is constructive, that is, for each type of primitive transformations,
we show that a simulation relation of a certain nature exists between a program and its transformed version 
under the primitive transformation.
To the best of our knowledge, our definitions of program transformations
and proofs of soundness of correlations of temporal properties across transformations
is the first approach which makes use of the boolean matrix algebra in this setting.

Temporal logic model checking is an algorithmic technique of checking whether a model
satisfies a temporal logic formula~(cf. \cite{clarke:model-checking}).
The validation of individual program transformations can be achieved by
model checking a specified property on the input program and (with possible renaming of atomic propositions)
also on its transformed version. For translation validation of compiler optimizers,
model checking can be used effectively~\cite{sefm06-KSK,gcc-validation}.
In this paper, we address the problem of proving preservation of temporal properties for all possible applications of
a transformation primitive. In temporal logic, a formula is interpreted on a single program and
therefore the proof systems for temporal logics~\cite{DBLP:journals/tcs/MannaP91,deductive:CTL} 
are not sufficient in this setting.

\subsection*{Outline}
In Section~\ref{sec:motivation}, we motivate the need for a logic like TTL 
through an application to verification of compiler optimizations.
In Section~\ref{sec:graph-transformations}, we define several primitive graph transformations which
are used in Section~\ref{sec:program-transformations} to define primitive program transformations.
In Section~\ref{sec:kripke-transformations}, we define the notion of transformations 
of Kripke structures and model program transformations as transformations of Kripke structures.
The syntax, semantics, and inference rules of TTL are presented in Section~\ref{sec:ttl}.
The proofs of soundness of the TTL inference rules are derived in Section~\ref{sec:proofs}.
In Section~\ref{sec:conclusions}, we present the conclusions.

\section{Motivation: Verification of compiler optimizations}
\label{sec:motivation}

\subsection*{Specification scheme}
\specification{
$
\begin{array}{lcl}
\pvsid{Avail(prog, e)}: \pvstype{mat} & = & \pvsid{Comp(prog, e)}\\ 
\pvsid{Redund(prog, e)}: \pvstype{mat} & = & \pvsid{Antloc(prog, e)} \;*\; 
\pvsid{AY(prog`cfg, AS(prog`cfg,Transp(prog, e),Avail(prog, e)))}\\
\pvsid{OrgAvail(prog, e)}: \pvstype{mat} & = & \pvsid{Avail(prog,e)} - \pvsid{Redund(prog,e)}
\end{array}
$
\\[2mm]
\pvsid{\;\;CSE\_Transformation(prog1, e)}: \pvstype{Program} = \\
\mbox{\;\quad {IF}} \pvsid{member(e,Expressions(prog1))} \;$\wedge\; \neg$ \pvsid{BASE?(e)} \texttt{THEN}\\
$
\begin{array}{llll}
\mbox{\;\quad \quad \texttt{LET}} & \pvsid{n1} & = & \pvsid{length(prog1`cfg`T)},\\
&		  \pvsid{orgavails} & = & \pvsid{OrgAvail(prog1, e)},\\
&		  \pvsid{m} & = & \pvsid{Count1s(orgavails)},\\
&          \pvsid{redund} & = & \pvsid{Redund(prog1, e)},\\
&          \pvsid{newpoints} & = & \pvsid{append(zerol(n1),onel(m))},\\
&          \pvsid{prog2} & = & \underline{\pvsid{IP(prog1, orgavails, newpoints)}},\\
&          \pvsid{t} & = & \pvsid{NEWVAR(prog2)},\\
&          \pvsid{prog3} & = & \underline{\pvsid{IA(prog2, newpoints, t, e)}},\\
&		  \pvsid{orgavails3} & = & \pvsid{append(orgavails,zerol(m))},\\
&          \pvsid{prog4} & = & \underline{\pvsid{RE(prog3, orgavails3, e, t)}},\\
&		  \pvsid{redund4} & = & \pvsid{append(redund,zerol(m))}\\
\mbox{\;\quad \quad \texttt{IN}} &&& \underline{\pvsid{RE(prog4, redund4, e, t)}}
\end{array}
$\\
\mbox{\;\quad \texttt{ELSE}} \pvsid{prog1} \texttt{ENDIF}
}{A specification of common subexpression elimination optimization}{cse-spec}

Consider the formal specification of common subexpression elimination (CSE) optimization given in Figure~\ref{fig:cse-spec}.
We use PVS (Prototype Verification System)~\cite{PVS:userguide} to develop, validate, and verify 
formal specifications of optimizations.
In this paper, we explain the CSE specification intuitively without giving details of the PVS specification language.
The specification defines both, the program analysis and the optimizing transformation. 
In Figure~\ref{fig:cse-spec}, \pvsid{prog} is a program and \pvsid{e} is an expression in \pvsid{prog}.
A function that defines a program analysis returns a boolean vector 
(whose type is a column matrix \pvstype{mat}) such that the boolean value
corresponding to a program point denotes whether the property holds at the program point or not.
We explain the CSE specification with an example optimization shown in Figure~\ref{fig:gcc-cse-example}.

The functions \pvsid{Transp}, \pvsid{Antloc}, and \pvsid{Comp} define \emph{local dataflow properties} i.e.
the dataflow information that depends only on the statement at a program point.
An expression \pvsid{e} is \pvsid{Transp} at a program point if none of its operands are modified at the program point.
An expression \pvsid{e} is \pvsid{Antloc} at a program point if the expression is computed at the program point.
If an expression is both \pvsid{Transp} and \pvsid{Antloc} at a program point then it is \pvsid{Comp} at the program point.
For brevity, we do not give the definitions of these local dataflow properties.
For the definitions, we refer the reader to~\cite{kanade:thesis}.
For \pvsid{prog1} in Figure~\ref{fig:gcc-cse-example}, the expression $a/b$ is \pvsid{Transp}
at all the program points, and is \pvsid{Antloc} and \pvsid{Comp} at program points 2, 3, and 5. 

The function \pvsid{Redund} defines a \emph{global dataflow property} i.e.
the dataflow information that depends on the statements along the paths ending or starting at a program point.
We specify global dataflow properties using $\ctlbp$. The functions \pvsid{AY} and \pvsid{AS} respectively denote
the \emph{universal predecessor} and the \emph{universal since} $\ctlbp$ operators.
Since an optimization transforms a program step-by-step,
we have different versions of the input program.
Thus, in the PVS specifications, we pass the control flow graph of a program 
(denoted as \pvsid{prog'cfg}) that 
a $\ctlbp$ operator is interpreted on as a parameter to the function.
However, in the latter part of the paper, we follow the usual convention where 
temporal operators do not take a model as a parameter.

A program point satisfies the \pvsid{Redund} property if
(1)~the expression \pvsid{e} is \pvsid{Antloc} at a program point and 
(2)~from all its predecessors,
the expression is \pvsid{Transp} along all incoming paths, since it is \pvsid{Avail}.
The origins of availability (\pvsid{OrgAvail}) are the program points where the expression is \pvsid{Comp} but not \pvsid{Redund}.
For \pvsid{prog1} in Figure~\ref{fig:gcc-cse-example}, the expression $a/b$ is \pvsid{Redund} 
at program point 5 and it is \pvsid{OrgAvail} at program points 2 and 3.

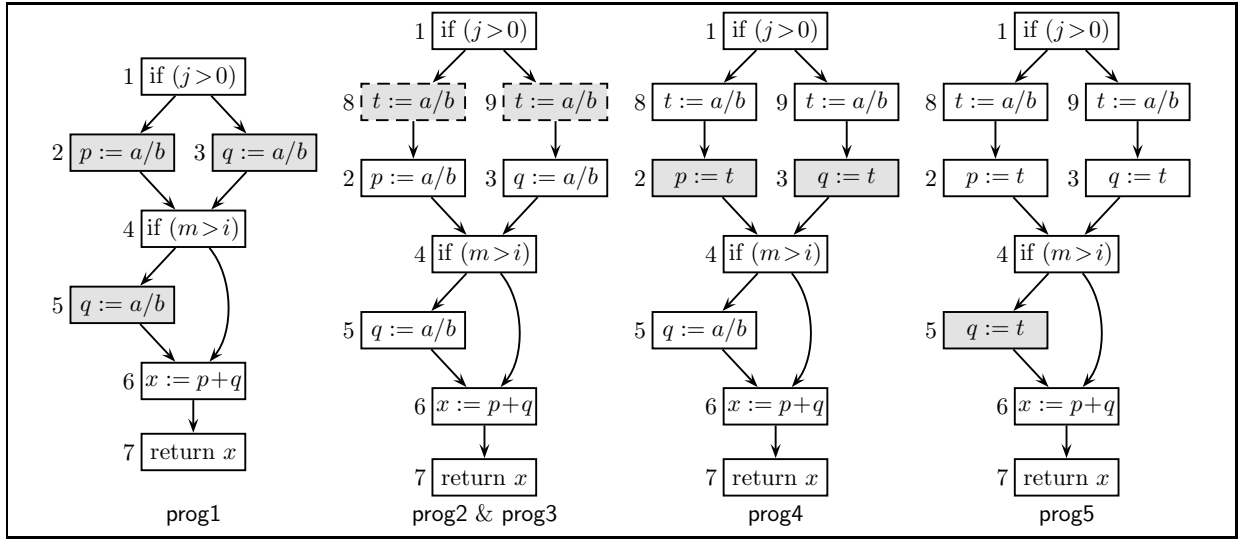
\begin{figure*}
\def\sepr{\\[0mm]}
\centering
\fbox{
\begin{minipage}{\minipagewidth}
\centering
\scalebox{0.9}{
\begin{tabular}{@{}c@{\hspace*{1.1cm}}c@{\hspace*{1.1cm}}c@{\hspace*{1.1cm}}c@{}}
\begin{minipage}{3.2cm}
\begin{pspicture}(0,0)(3.2,7.25)
\psset{xunit=1mm,yunit=0.75mm,arrowscale=1.4}

\def\reftext{$c := j \!>\! 0$}

\node{16}{84}{1}{\reftext}{$\text{if } (j \!>\! 0)$}
\lnode{5.5}{69}{2}{\reftext}{$p := a/b$}
\lnode{26.5}{69}{3}{\reftext}{$q := a/b$}
\node{16}{54}{4}{\reftext}{$\text{if } (m \!>\! i)$}
\lnode{5.5}{39}{5}{\reftext}{$q := a/b$}
\node{16}{24}{6}{\reftext}{$x := p\!+\!q$}
\node{16}{10}{7}{\reftext}{$\text{return } x$}

\ncline{->}{n1}{n2}
\ncline{->}{n1}{n3}
\ncline{->}{n3}{n4}
\ncline{->}{n2}{n4}
\ncline{->}{n4}{n5}
\ncarc[arcangleA=40,arcangleB=40]{->}{n4}{n6}
\ncline{->}{n5}{n6}
\ncline{->}{n6}{n7}
\end{pspicture}
\end{minipage}
&
\begin{minipage}{3.2cm}
\begin{pspicture}(0,0)(3.2,7.25)
\psset{xunit=1mm,yunit=0.75mm,arrowscale=1.4}

\def\reftext{$c := j \!>\! 0$}

\node{16}{93}{1}{\reftext}{$\text{if } (j \!>\! 0)$}
\newnode{5.5}{79}{8}{\reftext}{$t := a/b$}
\node{5.5}{64}{2}{\reftext}{$p := a/b$}
\newnode{26.5}{79}{9}{\reftext}{$t := a/b$}
\node{26.5}{64}{3}{\reftext}{$q := a/b$}
\node{16}{49}{4}{\reftext}{$\text{if } (m \!>\! i)$}
\node{5.5}{34}{5}{\reftext}{$q := a/b$}
\node{16}{19}{6}{\reftext}{$x := p\!+\!q$}
\node{16}{5}{7}{\reftext}{$\text{return } x$}

\ncline{->}{n1}{n8}
\ncline{->}{n1}{n9}
\ncline{->}{n8}{n2}
\ncline{->}{n9}{n3}
\ncline{->}{n3}{n4}
\ncline{->}{n2}{n4}
\ncline{->}{n4}{n5}
\ncarc[arcangleA=40,arcangleB=40]{->}{n4}{n6}
\ncline{->}{n5}{n6}
\ncline{->}{n6}{n7}
\end{pspicture}
\end{minipage}
&
\begin{minipage}{3.2cm}
\begin{pspicture}(0,0)(3.2,7.25)
\psset{xunit=1mm,yunit=0.75mm,arrowscale=1.4}

\def\reftext{$c := j \!>\! 0$}

\node{16}{93}{1}{\reftext}{$\text{if } (j \!>\! 0)$}
\node{5.5}{79}{8}{\reftext}{$t := a/b$}
\lnode{5.5}{64}{2}{\reftext}{$p := t$}
\node{26.5}{79}{9}{\reftext}{$t := a/b$}
\lnode{26.5}{64}{3}{\reftext}{$q := t$}
\node{16}{49}{4}{\reftext}{$\text{if } (m \!>\! i)$}
\node{5.5}{34}{5}{\reftext}{$q := a/b$}
\node{16}{19}{6}{\reftext}{$x := p\!+\!q$}
\node{16}{5}{7}{\reftext}{$\text{return } x$}

\ncline{->}{n1}{n8}
\ncline{->}{n1}{n9}
\ncline{->}{n8}{n2}
\ncline{->}{n9}{n3}
\ncline{->}{n3}{n4}
\ncline{->}{n2}{n4}
\ncline{->}{n4}{n5}
\ncarc[arcangleA=40,arcangleB=40]{->}{n4}{n6}
\ncline{->}{n5}{n6}
\ncline{->}{n6}{n7}
\end{pspicture}
\end{minipage}
&
\begin{minipage}{3.2cm}
\begin{pspicture}(0,0)(3.2,7.25)
\psset{xunit=1mm,yunit=0.75mm,arrowscale=1.4}

\def\reftext{$c := j \!>\! 0$}

\node{16}{93}{1}{\reftext}{$\text{if } (j \!>\! 0)$}
\node{5.5}{79}{8}{\reftext}{$t := a/b$}
\node{5.5}{64}{2}{\reftext}{$p := t$}
\node{26.5}{79}{9}{\reftext}{$t := a/b$}
\node{26.5}{64}{3}{\reftext}{$q := t$}
\node{16}{49}{4}{\reftext}{$\text{if } (m \!>\! i)$}
\lnode{5.5}{34}{5}{\reftext}{$q := t$}
\node{16}{19}{6}{\reftext}{$x := p\!+\!q$}
\node{16}{5}{7}{\reftext}{$\text{return } x$}

\ncline{->}{n1}{n8}
\ncline{->}{n1}{n9}
\ncline{->}{n8}{n2}
\ncline{->}{n9}{n3}
\ncline{->}{n3}{n4}
\ncline{->}{n2}{n4}
\ncline{->}{n4}{n5}
\ncarc[arcangleA=40,arcangleB=40]{->}{n4}{n6}
\ncline{->}{n5}{n6}
\ncline{->}{n6}{n7}
\end{pspicture}
\end{minipage}
\\
\pvsid{prog1} & \pvsid{prog2} \& \pvsid{prog3} & \pvsid{prog4} & \pvsid{prog5}
\end{tabular}}
\end{minipage}}
\caption{An example of common subexpression elimination optimization}
\label{fig:gcc-cse-example}
\end{figure*}

The function \pvsid{CSE\_Transformation} is an \emph{optimizing transformation} i.e.
it takes an input program and possibly other parameters, and
applies several primitive transformations to the input program.
\pvsid{CSE\_Transformation} is defined as the following sequence of transformations:
\begin{enumerate}
\item Transform \pvsid{prog1} to \pvsid{prog2}:
Insert a predecessor (\pvsid{IP}) each to the program points that satisfy the \pvsid{OrgAvail} property in \pvsid{prog1}. 
In Figure~\ref{fig:gcc-cse-example}, program points 8 and 9 are
the new predecessors to program points 2 and 3.
Due to space constraints, \pvsid{prog2} and \pvsid{prog3} (a transformed version of \pvsid{prog2})
are represented as same programs in Figure~\ref{fig:gcc-cse-example}.
In \pvsid{prog2}, program points 2 and 3 contain skip statements. The statements
shown at program points 2 and 3 are obtained by the next transformation.
\item Transform \pvsid{prog2} to \pvsid{prog3}:
Let \pvsid{t} be a new variable. Insert an assignment (\pvsid{IA}) \pvsid{t := e} at 
the new program points introduced in the first transformation.
In Figure~\ref{fig:gcc-cse-example}, the assignment $t := a/b$
(where $t$ is a new variable and $a/b$ is the expression under consideration)
is inserted at program points 8 and 9.
\item Transform \pvsid{prog3} to \pvsid{prog4}:
Replace the occurrences of \pvsid{e} at the program points that satisfy the \pvsid{OrgAvail}
property in \pvsid{prog1} by \pvsid{t}.
In Figure~\ref{fig:gcc-cse-example}, the computations of $a/b$ at program points 2 and 3 are replaced by $t$.
\item Transform \pvsid{prog4} to \pvsid{prog5}:
Replace the occurrences of the expression (\pvsid{RE}) \pvsid{e} at the program points that satisfy the \pvsid{Redund} 
property in \pvsid{prog1} by \pvsid{t}.
In Figure~\ref{fig:gcc-cse-example}, the computation of $a/b$ at program point 5 is replaced by $t$.
\end{enumerate}

In Section~\ref{sec:program-transformations}, we explain the definitions of the transformation
primitives: insertion of predecessors (\pvsid{IP}), insertion of assignments (\pvsid{IA}), and 
replacement of expressions (\pvsid{RE}). An extensive set of transformation primitives is defined in~\cite{kanade:thesis}.

\subsection*{Soundness conditions}





Consider an application of \pvsid{RE} : \pvsid{RE(prog,points,e,v)}.
It returns a transformed version of the input program \pvsid{prog}
where the occurrences of the expression \pvsid{e} at the program points denoted by \pvsid{points}
are replaced by a variable~\pvsid{v}. 
The replacement of an occurrence of \pvsid{e} at a program point $p$ by \pvsid{v}
preserves semantics if \pvsid{v} and \pvsid{e} have same value
just before $p$. The soundness conditions for \pvsid{RE} are as follows:
\begin{enumerate}
\item[(1)] The statement at $p$ is an assignment statement. \emph{AND}
\item[(2)] The variable \pvsid{v} is not an operand of the expression \pvsid{e}. \emph{AND}
\item[(3)] The expression \pvsid{e} is computed at $p$ (\pvsid{Antloc}). \emph{AND}
\item[(4)] Depending on the nature of \pvsid{e}, we have the following conditions:
\begin{enumerate}
\item If the expression is a just a variable (say \pvsid{x}) then 
along each backward path starting from the predecessors of $p$ one of the following should hold:
\begin{itemize}
\item An assignment $\pvsid{v} := \pvsid{x}$ or an assignment $\pvsid{x} := \pvsid{v}$ is encoutered
and in between there is no other assignment to any of the variables \pvsid{v} or \pvsid{x}.
\item The assignments $\pvsid{v} := \pvsid{h}$ and $\pvsid{x} := \pvsid{h}$ for some expression \pvsid{h} are encountered
(in some order) at program points say $p_1$ and $p_2$.
In between $p_1$ and $p_2$, the expression \pvsid{h} is \pvsid{Transp} ensuring that
both \pvsid{v} and \pvsid{x} has the same value. Upto the first of the program points $p_1$ and $p_2$,
along the backward path from $p$, neither \pvsid{v} nor \pvsid{x} is assigned.
\end{itemize}
\item Otherwise, along each backward path starting from the predecessors of $p$,
an assignment $\pvsid{v} := \pvsid{e}$ is encoutered with no assignment to \pvsid{v} or
the variable operands of \pvsid{e} in between.
\end{enumerate}
\end{enumerate}

Similarly soundness conditions for other transformation primitives can also be defined.
It can be observed that the soundness conditions can be naturally expressed in $\ctlbp$. 
For each transformation primitive, we identify the soundness conditions
and separately prove that if the soundness conditions are satisfied by an application
of the primitive then the transformed program is semantically equivalent to the input program.
These are one-time proofs. Since the primitives are small-step transformations, the proofs
of semantics preservation for them are easier than similar proofs for (large-step) optimizing transformations.

\subsection*{Verification scheme}

\newcommand{\abstraction}{P}

Consider an optimizing transformation $T$ applied to an input program $\abstraction_1$:
\begin{eqnarray*}
T(\abstraction_1) & \;\;\defas\;\; & 
\texttt{LET}\; \abstraction_2 = T_1(\abstraction_1,\pi_1),\;
\cdots,\;
\abstraction_k = T_{k-1}(\abstraction_{k-1},\pi_{k-1})\;\;
\texttt{IN}\;\; T_k(\abstraction_k,\pi_k)
\end{eqnarray*}
where a transformation primitive $T_i$ is applied to a program $\abstraction_i$ at
program points $\pi_i$. Let other parameters of the transformation primitives be implicit.

Let $\varphi_1, \ldots, \varphi_k$ respectively be the soundness conditions of
the primitives $T_1, \ldots, T_k$. The soundness of $T$ can be established
by proving that for each $i \in \{1,\ldots,k\}$, 
the verification condition $\varphi_i(\abstraction_i,\pi_i)$ is satisfied.
Thus, proving soundness of an optimization reduces to showing that the soundness
conditions of the underlying transformation primitives are satisfied.
Optimizations with similar objectives comprise similar transformations.
For example, ``replacement of some occurrences of an expression by a variable''
is a transformation which is common to optimizations like common subexpression elimination,
lazy code motion, loop invariant code motion, and several others whose aim is to avoid
unnecessary recomputations of a value.
Thus identification of transformation primitives and their soundness conditions simplifies
both specification and verification of a class of optimizations.

To prove a verification condition $\varphi_i(\abstraction_i,\pi_i)$, we use
the program analysis and properties of the preceding transformations.
However, the program analysis is interpreted on the input program and
properties of the preceding transformations are interpreted on previous versions of the program.
Thus we require a logic to correlate temporal properties across transformations of programs.
We therefore develop the temporal trasformation logic in this paper.

\section{Primitive graph transformations}
\label{sec:graph-transformations}

A structural transformation of a program involves a transformation of the control flow graph of the program.
For example, the insertion of predecessors (\pvsid{IP}) transformation used in the CSE specification (Figure~\ref{fig:cse-spec})
is a structural transformation. In this section, we define several primitive graph transformations which 
are then used, in the next section, to define primitive program transformations.

\subsection{Notation}

We use the boolean matrix algebra in a novel way to define graph transformations.
We now introduce the notation and terminology for boolean matrix algebraic operations and graph transformations.

\subsubsection*{Boolean matrix algebra}

Let $\bool = \{\zero, \one\}$ be the set of boolean values.
Consider a boolean algebra $\mathcal{A} = (\bool,\; \neg,\; \vee,\; \wedge,\; \zero,\; \one)$
where $\neg$, $\vee$, and $\wedge$ are respectively the boolean negation, disjunction, and conjunction operators.
Let $\zero$ and $\one$ respectively be the identity elements of $\vee$ and $\wedge$ operators.

For any two positive natural numbers $n$ and $m$, consider a function
\[
\bool_{n,m}: \{1,\ldots,n\} \times \{1,\ldots,m\} \rightarrow \bool.
\]
Let us consider $\bool_{n,m}$ as an $(n \Tm m)$ boolean matrix with $n$ rows and $m$ columns. 
For a matrix $U$, let $[U]^j_i$ denote the element at the $i$th row and the $j$th column.
Let $[U]_i$ denotes the $i$th row vector of $U$ and $[U]^j$ denote the $j$th column vector of $U$.

Consider a boolean matrix algebra
$\mathcal{A}_{n,m}$ defined as follows:
\[
\mathcal{A}_{n,m} = \left(\bool_{n,m},\; \Ng{\phantom{A}},\; \Tr{\phantom{A}},\; 
\Pl,\; \St,\; \Mm,\; -,\; \zerom_{n,m},\; \onem_{n,m} \right)
\]

The functions in the algebra are defined as follows:
\begin{enumerate}
\item $\Ng{U}$ denotes the negation of $U$ which extends the boolean negation to matrices.
\item $\Tr{U}$ denotes the transpose of $U$.
\item If $U$ and $V$ are two $(n \Tm m)$ matrices,
the addition ``$U \Pl V$'', the product ``$U \St V$'', and
the subtraction ``$U \Mn V$'' are defined as follows:
\[
\begin{array}{c@{\quad\quad}c@{\quad\quad}c}
\;[U \Pl V]^j_i \defas [U]^j_i \vee [V]^j_i &
\;[U \St V]^j_i \defas [U]^j_i \wedge [V]^j_i &
\;[U \Mn V]^j_i \defas [U]^j_i \wedge \neg [V]^j_i
\end{array}
\]
\item 
If $U$ is an $(n \Tm m)$ matrix and $V$ is an $(m \Tm p)$ matrix,
the multiplication ``$U \Mm V$'' gives an $(n \Tm p)$ matrix as follows:
\begin{equation*}
\label{eqn:def-matmult}
[U \Mm V]^j_i \;\defas\; \bigvee_{1 \leq k \leq m} \; (\; [U]^k_i \wedge [V]^j_k \;)
\end{equation*}
\end{enumerate}

$\zerom_{n,m}$ is an $(n \Tm m)$ matrix whose all elements are $\zero$s and is 
the identity element of the operator $\Pl$.
$\onem_{n,m}$ is an $(n \Tm m)$ matrix whose all elements are $\one$s and is 
the identity element of the operator $\St$.

For readability, we do not explicate dimensions of matrices and assume appropriate
dimensions as per their usage. We use bold letters to 
denote boolean vectors. Boolean vectors are also considered as column matrices.
The two special vectors $\veczero$ and $\vecone$ have all the elements as $\zero$s and
$\one$s respectively.

The precedence of the operators in the decreasing order is as follows:
$
[\;\Ng{\phantom{A}} \;\; \Tr{\phantom{l}}\;]
[\phantom{l} \Mm \phantom{l}]
[\phantom{l} \Pl \phantom{l} \Mn \phantom{l} \St \phantom{l}]
$.
The operators within a pair of brackets have the same precedence
and their usage needs to be disambiguated with appropriate parenthesization.

\subsubsection*{Graph transformations}
A graph $\graph$ is specified by a set of nodes $\mystates$ 
and a boolean adjacency matrix $\mat$ representing the edges of $\graph$.
We assume an implicit ordering of nodes to access elements in matrices
and vectors associated with a graph. The ordering is considered to be constant for a graph.

Consider a graph $\graph = (\mystates, \mat)$ and its transformed version $\graph' = (\mystates',\mat')$.
A \emph{correspondence relation} $\statemap$ for a transformation of $\graph$ to $\graph'$ gives the correspondence between 
the nodes of $\graph'$ and the nodes of $\graph$. 
It is denoted as an $(|\mystates'| \Tm |\mystates|)$ boolean matrix.
We define the transformation of $\graph$ to $\graph'$ in terms of
(1) the nature of the correspondence relation $\statemap$ and
(2) an equation expressing the adjacency matrix $\mat'$ of the transformed graph in terms
the adjacency matrix $\mat$ of the input graph.

Below we describe seven types of primitive graph transformations.
These transformations are orthogonal with each other i.e. none of them can be expressed
in terms of the others. 

\subsection{Node Splitting}
\label{sec:trans.node-split}

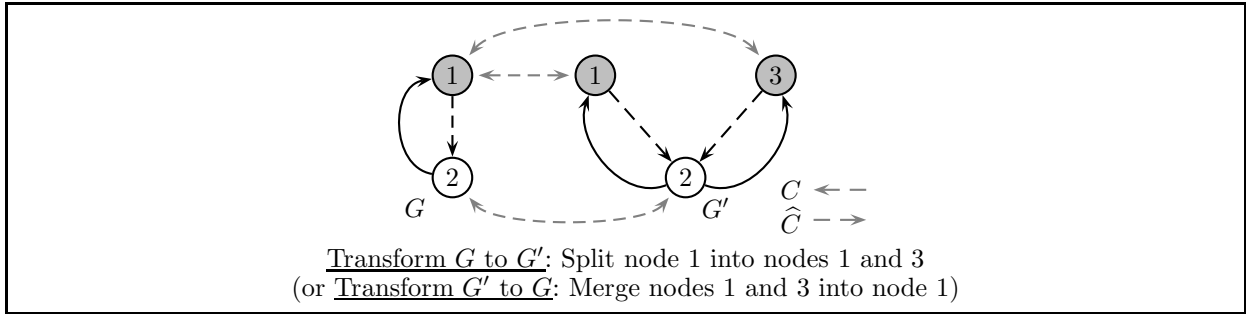
\begin{figure}
\centering
\fbox{
\begin{minipage}{\minipagewidth}
\centering
\begin{tabular}{c}
\begin{pspicture}(0,0)(6.6,3)
\psset{arrowscale=1.4}
\psset{xunit=1mm,yunit=0.8mm}
\rput(10,27){\circlenode[fillstyle=solid,fillcolor=lightgray]{n1}{$1$}}
\rput(10,10){\circlenode{n2}{$2$}}
\ncline[linestyle=dashed]{->}{n1}{n2}
\ncarc[ncurv=0.9,arcangleA=80,arcangleB=80]{->}{n2}{n1}

\rput(5,5){$\graph$}
\rput(45,5){$\graph'$}

\rput(55,8){$\statemap$}
\psline[linecolor=gray,linestyle=dashed]{<-}(58,8)(65,8)
\rput(55,3){$\Tr{\statemap}$}
\psline[linecolor=gray,linestyle=dashed]{->}(58,3)(65,3)

\rput(29,27){\circlenode[fillstyle=solid,fillcolor=lightgray]{m1}{$1$}}
\rput(53,27){\circlenode[fillstyle=solid,fillcolor=lightgray]{m3}{$3$}}
\rput(41,10){\circlenode{m2}{$2$}}
\ncline[linestyle=dashed]{->}{m1}{m2}
\ncline[linestyle=dashed]{->}{m3}{m2}
\ncarc[ncurv=0.8,arcangleA=70,arcangleB=60]{->}{m2}{m1}
\ncarc[ncurv=0.8,arcangleA=-70,arcangleB=-60]{->}{m2}{m3}

\ncline[linecolor=gray,linestyle=dashed,nodesep=2pt]{<->}{n1}{m1}
\ncarc[linecolor=gray,ncurv=0.4,arcangleA=50,arcangleB=50,linestyle=dashed,nodesep=2pt]{<->}{n1}{m3}
\ncarc[linecolor=gray,ncurv=0.4,arcangleA=-50,arcangleB=-50,linestyle=dashed,nodesep=2pt]{<->}{n2}{m2}
\end{pspicture}\\
\underline{Transform $\graph$ to $\graph'$}: Split node $1$ into nodes $1$ and $3$ \\
(or \underline{Transform $\graph'$ to $\graph$}: Merge nodes $1$ and $3$ into node $1$)
\end{tabular}
\end{minipage}
}
\caption
{An Example of Node Splitting (or Node Merging)}
\label{fig:node-split}
\end{figure}

A {node splitting} transformation of a graph $\graph$ splits at least one node of $\graph$ into 
multiple nodes. The edges associated with a node being split are
translated into similar edges for the corresponding nodes in the transformed graph $\graph'$.
All the other edges of $G$ are preserved. 
For example, consider the graphs $\graph$ and $\graph'$ shown in Figure~\ref{fig:node-split}.
Node $1$ of $\graph$ is split into nodes $1$ and $3$ to get graph $\graph'$.
Edge $\pair{1}{2}$ of $\graph$ is translated into edges $\pair{1}{2}$ and $\pair{3}{2}$
of~$\graph'$. Similarly, edge $\pair{2}{1}$ of $\graph$ is translated into edges
$\pair{2}{1}$ and $\pair{2}{3}$ of $\graph'$.

In Figure~\ref{fig:node-split}, the correspondence relation $\statemap$ between
the nodes of $\graph'$ and $\graph$ is shown by the dashed gray arrows from
right to left. A correspondence relation arrow from a node $p$ to a node $q$ means that 
the element at the $p$th row and $q$th column in the correspondence matrix $\statemap$ is (boolean) $\one$
i.e. $[\statemap]^q_p = \one$.
The matrix $\Tr{\statemap}$ is represented by the dashed gray arrows from left to right.

The correspondence between the edges of the two graphs can be
\emph{traced diagrammatically}. For example, 
suppose we go from node $3$ of $\graph'$ to node $1$ of $\graph$ by following a $\statemap$ arrow,
then we follow edge $\pair{1}{2}$ in $\graph$, and lastly, we follow the $\Tr{\statemap}$ arrow from 
node $2$ of $\graph$ to node~$2$ of $\graph'$.
This gives us edge $\pair{3}{2}$ in $\graph'$. 
Similarly, all the other edges of $\graph'$ can be traced.

The composition of arrows and edges (which are respectively, relations between the nodes of $\graph'$ and $\graph$,
and between the nodes of the individual graphs)
can be represented very concisely as \emph{matrix multiplications}, when these relations are expressed as \emph{adjacency matrices}.
This allows us to express the adjacency matrix of the transformed graph
in terms of the adjacency matrix of the input graph.

\begin{definition}
\label{def:ns}
The transformation of a graph $\graph = (\mystates,\mat)$ to a graph $\graph' = (\mystates',\mat')$
is called a \emph{node splitting} transformation if
\begin{enumerate}
\item The correspondence relation $\statemap$ between the nodes of $\graph'$ and $\graph$ is
a total, onto, and many-to-one relation and
\item $\statemap \Mm \mat \Mm \Tr{\statemap} = \mat'$.
\end{enumerate}
\end{definition}

The first condition specifies the nature of the correspondence relation.
The correspondence relation is required to be total and onto.
Totality ensures that each node of the transformed graph corresponds to
a node of the input graph i.e. no new node is introduced in the transformed graph.
The onto nature of the correspondence relation ensures that for each node of the input graph
there is a node of the transformed graph i.e. no node of the input graph is deleted.
The many-to-one nature of the correspondence relation ensures that at least one node of the transformed
graph corresponds to multiple nodes of the input graph i.e. the transformation splits at least one node 
of the input graph and is not vacuous.
We require the non-vacuous nature of the relation so that the definitions of the graph transformations are 
unambiguous and orthogonal.
The second condition defines the adjacency matrix of the transformed graph in
terms of the adjacency matrix of the input graph and the correspondence relation.

\subsection{Node Merging}
\label{sec:trans.node-merge}

A {node merging} transformation of a graph $\graph$ merges multiple nodes of $\graph$
into a single node. The edges associated with the nodes being merged are coalesced into
similar edges for the merged node. All the other edges of $\graph$ are preserved.
For example, consider the graphs shown in Figure~\ref{fig:node-split} once again.
This time consider $\graph$ to be a transformation of $\graph'$.
Nodes $1$ and $3$ of $\graph'$ are merged into node $1$ of
$\graph$. The correspondence relation $\statemap$ is now given by 
the dashed gray arrows from left to right (opposite to that of the node splitting example). 
Edges $\pair{1}{2}$ and $\pair{3}{2}$ of $\graph'$ are coalesced into
edge $\pair{1}{2}$ of $\graph$. Similarly, edges $\pair{2}{1}$
and $\pair{2}{3}$ of $\graph'$ are coalesced into edge $\pair{2}{1}$ of $\graph$.
As described earlier, the correspondence between the edges of the two graphs
can be traced diagrammatically.

\begin{definition}
\label{def:nm}
The transformation of a graph $\graph = (\mystates,\mat)$ to a graph $\graph' = (\mystates',\mat')$
is called a \emph{node merging} transformation if
\begin{enumerate}
\item The correspondence relation $\statemap$ between the nodes of $\graph'$ and $\graph$ is
a total, onto, and one-to-many relation and
\item $\statemap \Mm \mat \Mm \Tr{\statemap} = \mat'$.
\end{enumerate}
\end{definition}

The first condition ensures that there are no new nodes introduced in the transformed graph (totality)
and no nodes of the input graph are deleted (ontoness). The one-to-many nature ensures that
more than one nodes are merged and the transformation is not vacuous.
The second condition defines the adjacency matrix of the graph in terms of the adjacency
matrix of the input graph and the correspondence relation.

\subsection{Edge Addition}
\label{sec:trans.edge-add}

\begin{figure}
\centering
\fbox{
\begin{minipage}{\minipagewidth}
\centering
\begin{tabular}{c}
\begin{pspicture}(0,0)(7.6,3)
\psset{arrowscale=1.4}
\psset{xunit=1mm,yunit=0.8mm}
\rput(6,27){\circlenode{n1}{$1$}}
\rput(30,27){\circlenode{n3}{$3$}}
\rput(18,10){\circlenode{n2}{$2$}}
\ncline{->}{n1}{n2}
\ncline{->}{n3}{n2}
\ncarc[ncurv=0.8,arcangleA=70,arcangleB=60]{->}{n2}{n1}
\ncarc[ncurv=0.8,arcangleA=-70,arcangleB=-60]{->}{n2}{n3}

\rput(14,5){$\graph$}
\rput(64,5){$\graph'$}

\rput(47,27){\circlenode{m1}{$1$}}
\rput(71,27){\circlenode{m3}{$3$}}
\rput(59,10){\circlenode{m2}{$2$}}
\ncline{->}{m1}{m2}
\ncline{->}{m3}{m2}
\ncarc[ncurv=0.8,arcangleA=70,arcangleB=60]{->}{m2}{m1}
\ncarc[ncurv=0.8,arcangleA=-70,arcangleB=-60]{->}{m2}{m3}
\ncline[linestyle=dashed]{->}{m1}{m3}

\ncarc[linecolor=gray,ncurv=0.4,arcangleA=50,arcangleB=50,linestyle=dashed,nodesep=2pt]{<->}{n1}{m1}
\ncarc[linecolor=gray,border=1pt,ncurv=0.4,arcangleA=50,arcangleB=50,linestyle=dashed,nodesep=2pt]{<->}{n3}{m3}
\ncarc[linecolor=gray,ncurv=0.4,arcangleA=-50,arcangleB=-50,linestyle=dashed,nodesep=2pt]{<->}{n2}{m2}
\end{pspicture}\\
\underline{Transform $\graph$ to $\graph'$}: Add an edge $\pair{1}{3}$\\
(or \underline{Transform $\graph'$ to $\graph$}: Delete the edge $\pair{1}{3}$)
\end{tabular}
\end{minipage}
}
\caption
{An Example of Edge Addition (or Edge Deletion)}
\label{fig:edge-add}
\end{figure}
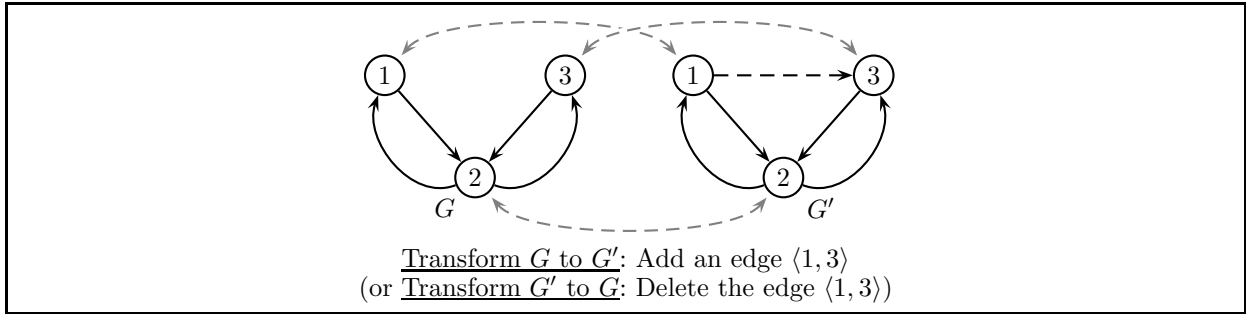

An {edge addition} transformation of a graph $\graph$ adds at least one edge to $\graph$.
All the edges of $\graph$ are preserved.
For example, consider the graphs $\graph$ and $\graph'$ shown in Figure~\ref{fig:edge-add} such that $\graph'$ is obtained by
adding edge $\pair{1}{3}$ to $\graph$.
The correspondence relation $\statemap$ is given by the dashed gray arrows from right to left.
The correspondence between the rest of the edges can be traced diagrammatically.

\begin{definition}
\label{def:ea}
The transformation of a graph $\graph = (\mystates,\mat)$ to a graph $\graph' = (\mystates',\mat')$
is called an \emph{edge addition} transformation if
\begin{enumerate}
\item The correspondence relation $\statemap$ between the nodes of $\graph'$ and $\graph$ is
a bijection and
\item There exists a $(|\mystates| \times |\mystates|)$ matrix $E$ s.t.
$\mat \Pl E > \mat$ and 
$\statemap \Mm \mat \Mm \Tr{\statemap} \; \Pl \; \statemap \Mm E \Mm \Tr{\statemap} = \mat'$
where the matrix $E$ denotes the edges to be added to $\graph$.
\end{enumerate}
\end{definition}

The first condition ensures that there is a one-to-one correspondence between the nodes of $\graph'$ and $\graph$.
The second condition ensures that at least one new edge is added.

\subsection{Edge Deletion}
\label{sec:trans.edge-del}

An {edge deletion} transformation of a graph $\graph$ deletes at least one edge of $\graph$.
All the other edges of $\graph$ are preserved.
For example, consider the graphs shown in Figure~\ref{fig:edge-add} once again.
This time consider $\graph$ to be a transformation of $\graph'$.
Edge $\pair{1}{3}$ is deleted from $\graph'$ to get $\graph$.
The correspondence relation $\statemap$ is now given as the dashed gray
arrows from left to right (opposite to that of the edge addition example).
The correspondence between the rest of the edges can be traced diagrammatically.

\begin{definition}
\label{def:ed}
The transformation of a graph $\graph = (\mystates,\mat)$ to a graph $\graph' = (\mystates',\mat')$
is called an \emph{edge deletion} transformation if
\begin{enumerate}
\item The correspondence relation $\statemap$ between the nodes of $\graph'$ and $\graph$ is
a bijection and
\item There exists a $(|\mystates| \times |\mystates|)$ matrix $E$ s.t.
$\mat \Mn E < \mat$ and 
$\statemap \Mm \mat \Mm \Tr{\statemap} \; \Mn \; \statemap \Mm E \Mm \Tr{\statemap} = \mat'$
where the matrix $E$ denotes the edges to be deleted from $\graph$.
\end{enumerate}
\end{definition}

The first condition ensures that there is a one-to-one correspondence between the nodes of $\graph'$ and $\graph$.
The second condition ensures that at least one edge is deleted.

\subsection{Node Addition}
\label{sec:trans.node-add}

\begin{figure}
\centering
\fbox{
\begin{minipage}{\minipagewidth}
\centering
\begin{tabular}{c}
\begin{pspicture}(0,0)(7.4,3.6)
\psset{arrowscale=1.4}
\psset{xunit=1mm,yunit=1mm}
\rput(6,27){\circlenode{n1}{1}}
\rput(30,27){\circlenode[fillstyle=solid,fillcolor=lightgray]{n3}{3}}
\rput(18,10){\circlenode[fillstyle=solid,fillcolor=lightgray]{n2}{2}}
\rput(22,20){$E$}
\ncline[linestyle=dashed]{->}{n3}{n2}
\ncline{->}{n1}{n2}
\ncarc[ncurv=0.8,arcangleA=70,arcangleB=60]{->}{n2}{n1}
\ncarc[ncurv=0.8,arcangleA=-70,arcangleB=-60]{->}{n2}{n3}
\ncline{->}{n1}{n3}

\rput(14,5){$\graph$}
\rput(61,1){$\graph'$}

\rput(43,27){\circlenode{m1}{1}}
\rput(67,27){\circlenode{m3}{3}}
\rput(55,4){\circlenode{m2}{2}}
\rput(61,16){\circlenode[linestyle=dashed]{m4}{4}}
\ncline{->}{m1}{m2}
\rput(61,23){$E_I$}
\ncline[linestyle=dashed]{->}{m3}{m4}
\rput(61,9){$E_O$}
\ncline[linestyle=dashed]{->}{m4}{m2}
\ncarc[ncurv=0.8,arcangleA=70,arcangleB=60]{->}{m2}{m1}
\ncarc[ncurv=0.8,arcangleA=-70,arcangleB=-60]{->}{m2}{m3}
\ncline{->}{m1}{m3}

\rput(52,21){$N_P$}
\ncarc[nodesep=2pt,linecolor=gray,linestyle=dotted,border=1pt,dotsep=1pt,ncurv=1,arcangleA=0,arcangleB=20]{<-}{m4}{n3}
\rput(35,9){$N_S$}
\ncline[nodesep=2pt,linecolor=lightgray,border=1pt]{<-}{n2}{m4}
\ncarc[linecolor=gray,ncurv=0.4,arcangleA=50,arcangleB=50,linestyle=dashed,nodesep=2pt]{<->}{n1}{m1}
\ncarc[linecolor=gray,ncurv=0.4,border=1pt,arcangleA=50,arcangleB=50,linestyle=dashed,nodesep=2pt]{<->}{n3}{m3}
\ncarc[linecolor=gray,ncurv=0.4,arcangleA=-30,arcangleB=-30,linestyle=dashed,nodesep=2pt]{<->}{n2}{m2}
\end{pspicture}\\
\underline{Transform $\graph$ to $\graph'$}: Add node $4$ along the edge $\pair{3}{2}$\\
(or \underline{Transform $\graph'$ to $\graph$}: Delete node $4$)
\end{tabular}
\end{minipage}
}
\caption
{An Example of Node Addition (or Node Deletion)}
\label{fig:node-add}
\end{figure}
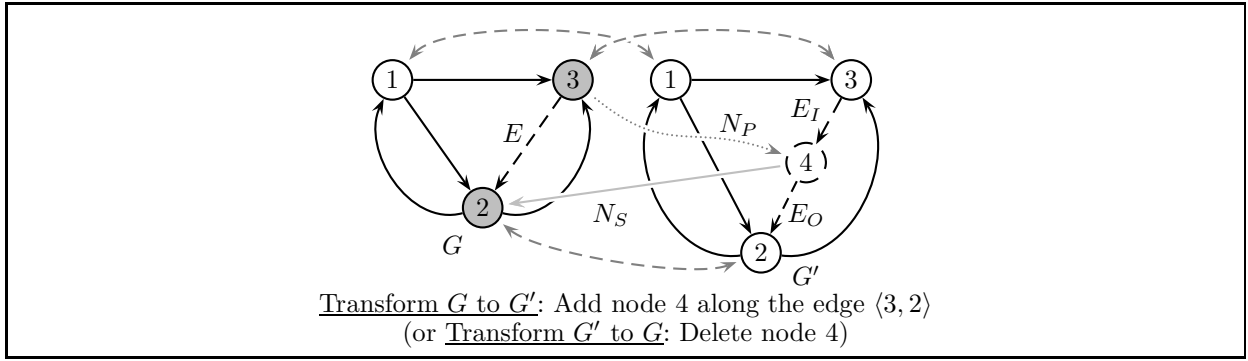

For a graph $\graph$, given a set of edges $E$, a {node addition} transformation adds a new node along each edge in $E$.
It splits the edges in $E$ and adds the new nodes as successors to the source nodes of 
the edges in $E$ and as predecessors to the target nodes of the edges in $E$.
The rest of the edges of $\graph$ are preserved.

For example, consider the two graphs shown in Figure~\ref{fig:node-add}.
The graph $\graph'$ is obtained by adding node $4$ along edge $\pair{3}{2}$ of $\graph$.
Edge $\pair{3}{2}$ is split into two edges $\pair{3}{4}$ and $\pair{4}{2}$ making
node $4$ a successor of node $3$ and a predecessor of node $2$.
The correspondence relation $\statemap$ is given as the dashed gray arrows from right to left.
The lightgray solid arrows from right to left marked as $N_S$ map
the newly added nodes to the target nodes of the edges in $E$.
The lightgray dotted arrows from left to right marked as $N_P$ map
the source nodes of the edges in $E$ to the newly added nodes.

The correspondence between the edges of the two graphs can be traced diagrammatically.
Edge $\pair{3}{4}$ is obtained by following the $\statemap$ arrow from node 3 of $\graph'$
to node 3 of $\graph$ and then following the $N_P$ arrow from node 3 to node 4 of $\graph'$.
Edge $\pair{4}{3}$ is obtained by following the $N_S$ arrow from node $4$ of $\graph'$
to node $2$ of $\graph$ and then following the $\Tr{\statemap}$ arrow from node $2$ of $\graph$
to node 2 of $\graph'$. In order to form edges corresponding to the edges of $\graph$,
we traverse all the edges of $\graph$ except the edges belonging to $E$.

\begin{definition}
\label{def:na}
The transformation of a graph $\graph = (\mystates,\mat)$ to a graph $\graph' = (\mystates',\mat')$
is called a \emph{node addition} transformation if
\begin{enumerate}
\item The correspondence relation $\statemap$ between the nodes of $\graph'$ and $\graph$ is
a partial, onto, and one-to-one relation and
\item Let the set of edges to be split be represented by an adjacency matrix $E$.
Let $\vec{E_S} = E \Mm \vecone$ be a vector denoting the source nodes of the edges in $E$. 
Let $\vec{E_T} = \Tr{E} \Mm \vecone$ be a vector denoting the target nodes of the edges in $E$ and 
$\vec{\newstates} = \vecone \Mn \statemap \Mm \vecone$ be a vector denoting the new nodes. 
There exist a $(|\mystates| \!\times\! |\mystates|)$ matrix~$E$,
a $(|\mystates| \!\times\! |\mystates'|)$ matrix $N_P$, and
a $(|\mystates'| \!\times\! |\mystates|)$ matrix $N_S$ such that the following conditions hold:
\begin{enumerate}
\item $E \leq \mat$,
\item $N_P$ is a total and onto relation from the nodes denoted by $\vec{E_S}$ to the nodes denoted by $\vec{\newstates}$,
\item $N_S$ is a total and onto relation from the nodes denoted by $\vec{\newstates}$ to the nodes denoted by $\vec{E_T}$, 
\item $E = N_P \Mm N_S$, and
\item $(\statemap \Mm \mat \Mm \Tr{\statemap} \;\Mn\; \statemap \Mm E \Mm \Tr{\statemap}) 
\;\Pl\; \underbrace{\statemap \Mm N_P}_{E_I} \;\Pl\; \underbrace{N_S \Mm \Tr{\statemap}}_{E_O} = \mat'$.
\end{enumerate}
\end{enumerate}
\end{definition}

The first condition ensures that there is at least one new node in the transformed graph (partiality)
but no nodes of the input graph are deleted (ontoness).
The one-to-one nature of the correspondence relation ensures that no nodes of the input graph are split or merged.

The second condition states the following:
(a)~The set $E$ of the edges to be split is a subset of the edges of the input graph.
(b)--(c)~All the edges in $E$ are split.
(d)~The edge splitting does not lead to any extraneous edges and any edges between the new nodes in the transformed graph.
(e)~It defines the adjacency matrix $\mat'$ of the transformed graph in terms
the adjacency matrix $\mat$ of the input graph, the matrix $E$, and 
the relations $\statemap$, $N_P$ and $N_S$. $E_I$ and $E_O$ respectively denote 
the incoming and outgoing edges for the new nodes as shown in Figure~\ref{fig:node-add}.

\subsection{Node Deletion}
\label{sec:trans.node-del}

A {node deletion} transformation of a graph $\graph$ deletes at least one node.
The incoming edges of the node being deleted are composed with its outgoing edges. 
The rest of the edges of $\graph$ are preserved.
For example, consider the two graphs shown in Figure~\ref{fig:node-add} once again.
This time consider $\graph$ to be a transformation of $\graph'$.
Node $4$ of $\graph'$ is deleted and edges $\pair{3}{4}$ and $\pair{4}{2}$
are joined to get edge $\pair{3}{2}$ of $\graph$. The correspondence
relation $\statemap$ is given as the dashed gray arrows from left to right
(opposite to that of the node addition example). 
The dotted gray arrows from left to right marked as $N_P$ map the
nodes in $\graph'$ corresponding to the predecessors of the node being deleted
to the node being deleted.
The solid gray arrows from right to left marked as $N_S$ map 
the node being deleted to the
nodes in $\graph'$ corresponding to the successors of the node being deleted.
The correspondence between the edges
of the two graphs can be traced diagrammatically.

\begin{definition}
\label{def:nd}
The transformation of a graph $\graph = (\mystates,\mat)$ to a graph $\graph' = (\mystates',\mat')$
is called a \emph{node deletion} transformation if all of the following conditions hold:
\begin{enumerate}
\item The correspondence relation $\statemap$ between the nodes of $\graph'$ and $\graph$ is
a total and one-to-one but not onto.
\item Let $\vec{\delstates} = \vecone \;\Mn\; \Tr{\statemap} \Mm \vecone$ be a vector the nodes being deleted.
Let $\vec{\delstates_P} = \mat \Mm \vec{\delstates}$ be a vector denoting the predecessors of the nodes being deleted and
$\vec{\delstates_S} = \Tr{\mat} \Mm \vec{\delstates}$ be a vector denoting the successors of the nodes being deleted.
We require that no two nodes having an edge between them can be deleted.
We state this requirement as $\vec{\delstates} \;\St\; (\vec{\delstates_P} \;\Pl\; \vec{\delstates_S}) = \veczero$.
\item 
Let $\vec{\delstates_{\boldsymbol{P'}}} = \statemap \Mm \vec{\delstates_P}$ and $\vec{\delstates_{\boldsymbol{S'}}} = 
\statemap \Mm \vec{\delstates_S}$ be
the nodes in $\graph'$ corresponding respectively to nodes denoted by $\vec{\delstates_P}$ and $\vec{\delstates_S}$.
Further, let $E = (\mat \Mm \mat) \;\St\; (\vec{\delstates_P} \times \vec{\delstates_S})$ where $\times$ is the
Cartesian product operator. The set $E$ gives the set of edges obtained by joining the incoming and
outgoing edges of the nodes being deleted.
There exist a $(|\mystates'| \times |\mystates|)$ matrix $N_P$, and
a $(|\mystates| \times |\mystates'|)$ matrix $N_S$ such that the following conditions hold:
\begin{enumerate}
\item $N_P$ is a total and onto relation from the nodes denoted by $\vec{\delstates_{\boldsymbol{P'}}}$ 
to the nodes denoted by $\vec{\delstates}$,
\item $N_S$ is a total and onto relation from the nodes denoted by $\vec{\delstates}$ 
to the nodes denoted by $\vec{\delstates_{\boldsymbol{S'}}}$,
\item $\statemap \Mm E \Mm \Tr{\statemap} = N_P \Mm N_S$, and
\item $\statemap \Mm \mat \Mm \Tr{\statemap} \;\Pl\; N_P \Mm N_S = \mat'$.
\end{enumerate}
\end{enumerate}
\end{definition}

The first condition states that there are no new nodes (totality) in the transformed graph 
and no nodes of the input graph are split or merged (one-to-one nature). 
However, at least one node of the input graph is deleted (the correspondence relation not being onto).
The second condition states that no two nodes with an edge between them are deleted.

The third condition states the following:
(a)--(b) The incoming edges of the nodes being deleted are joined with the outgoing edges.
(c)~The node deletion does not create any other edges (except those allowed by (a) and (b)) in the transformed graph.
(d)~It defines the adjacency matrix $\mat'$ of the transformed graph in terms
of the adjacency matrix $\mat$ of the input graph and relations $\statemap$, $N_P$, and $N_S$.

\subsection{Isomorphic Transformation}
\label{sec:trans.iso}

An isomorphic transformation transforms a graph into an isomorphic graph.

\begin{definition}
\label{def:im}
The transformation of a graph $\graph = (\mystates,\mat)$ to a graph $\graph' = (\mystates',\mat')$
is called an \emph{isomorphic transformation} if
\begin{enumerate}
\item The correspondence $\statemap$ between the nodes of $\graph'$ and $\graph$ is a bijection and
\item $\statemap \Mm \mat \Mm \Tr{\statemap} = \mat'$.
\end{enumerate}
\end{definition}

\section{Primitive program transformations}
\label{sec:program-transformations}

\renewcommand{\node}[5]
{\rput(#1,#2){\rnode{n#3}{\psframebox[fillcolor=white,fillstyle=solid]{\phantom{\!#4\!}}}}
\rput(#1,#2){#5}
\rput(#1,#2){
\ifthenelse{#3>9}{\rput(-7.5,0){#3}}{\rput(-6.5,0){#3}}
}}

\renewcommand{\lnode}[5]
{\rput(#1,#2){\rnode{n#3}{\psframebox[fillstyle=solid,fillcolor=verylightgray]{\phantom{\!#4\!}}}}
\rput(#1,#2){#5}
\rput(#1,#2){
\ifthenelse{#3>9}{\rput(-10.5,0){#3}}{\rput(-7,0){#3}}
}}

A primitive program transformation (or simply a transformation primitive) 
changes an input program in a small well-defined step.
In this section, we present only the transformation primitives used in the CSE specification (Figure~\ref{fig:cse-spec}):
insertion of predecessors (\pvsid{IP}), insertion of assignments (\pvsid{IA}), and replacement of expressions (\pvsid{RE}).
To define a primitive, we need to define the associated control flow and content transformations.
A control flow transformation is modeled in terms of 
the primitive graph transformations (Section~\ref{sec:graph-transformations}).
This requires setting up various relations (matrices) used in the definitions of the primitive graph transformations.

Apart from the primitives presented here, we have also defined several other primitives
viz. insertion of successors (\pvsid{IS}), edge splitting (\pvsid{SE}), deletion of statements (\pvsid{DS}),
and replacement of variable operands (\pvsid{RV}).
The formal definitions all the primitives are available in~\cite{kanade:thesis}.
These primitives have been used for defining most of the classical 
optimizations like common subexpression elimination, optimal code placement,
loop invariant code motion, lazy code motion, and full and partial dead code elimination~\cite{kanade:thesis}.

\subsection{Insertion of predecessors}

\newcommand{\statemapMat}{
\scalebox{0.7}{
\begin{tabular}{c}
\colorbox{verylightgray}{
$
\begin{array}{@{}lllllll}
1 & 0 & 0 & 0 & 0 & 0\\
0 & 1 & 0 & 0 & 0 & 0\\
0 & 0 & 1 & 0 & 0 & 0\\
0 & 0 & 0 & 1 & 0 & 0\\
0 & 0 & 0 & 0 & 1 & 0\\
0 & 0 & 0 & 0 & 0 & 1\\[-1mm]
\end{array}
$}\\
\colorbox{lightgray}{
$
\begin{array}{@{}lllllll}
0 & 0 & 0 & 0 & 0 & 0\\[-1mm]
\end{array}
$}
\end{tabular}
}}

\newcommand{\NS}{
\scalebox{0.7}{
\begin{tabular}{c}
\colorbox{verylightgray}{
$
\begin{array}{@{}lllllll}
0 & 0 & 0 & 0 & 0 & 0\\
0 & 0 & 0 & 0 & 0 & 0\\
0 & 0 & 0 & 0 & 0 & 0\\
0 & 0 & 0 & 0 & 0 & 0\\
0 & 0 & 0 & 0 & 0 & 0\\
0 & 0 & 0 & 0 & 0 & 0\\[-1mm]
\end{array}
$}\\
\colorbox{lightgray}{
$
\begin{array}{@{}lllllll}
0 & 1 & 0 & 0 & 0 & 0\\[-1mm]
\end{array}
$}
\end{tabular}}
}

\newcommand{\NP}{
\scalebox{0.7}{
\begin{tabular}{@{}l@{}l}
\colorbox{verylightgray}{
$
\begin{array}{@{}lllllll@{}}
0 & 0 & 0 & 0 & 0 & 0\\
0 & 0 & 0 & 0 & 0 & 0\\
0 & 0 & 0 & 0 & 0 & 0\\
0 & 0 & 0 & 0 & 0 & 0\\
0 & 0 & 0 & 0 & 0 & 0\\
0 & 0 & 0 & 0 & 0 & 0\\[-1mm]
\end{array}
$} &
\colorbox{lightgray}{
$
\begin{array}{@{}l}
1\\
0\\
0\\
0\\
1\\
0\\[-1mm]
\end{array}
$}
\end{tabular}
}}

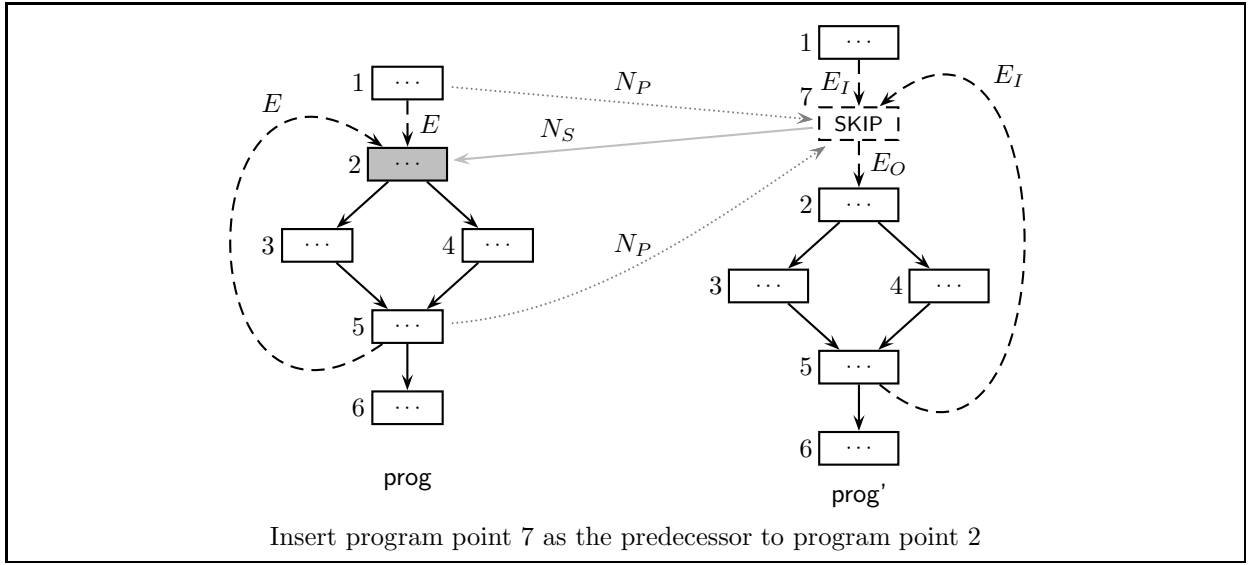
\begin{figure}
\centering
\fbox{
\begin{minipage}{\minipagewidth}
\centering
\begin{tabular}{c}
\begin{pspicture}(0,0)(11.8,6.7)
\psset{xunit=1mm,yunit=0.9mm}
\psset{arrowscale=1.4}
\renewcommand{\s}{{\small $\pvsid{\;SKIP\;}$}}
\node{30}{64}{1}{\s}{\small ${\cdots}$}
\dnode{30}{52}{2}{\s}{{\small ${\cdots}$}}
\node{18}{40}{3}{\s}{\small ${\cdots}$}
\node{42}{40}{4}{\s}{\small ${\cdots}$}
\node{30}{28}{5}{\s}{\small ${\cdots}$}
\node{30}{16}{6}{\s}{\small ${\cdots}$}
\rput(30,5){\pvsid{prog}}
\ncline[linestyle=dashed]{->}{n1}{n2}
\ncline{->}{n2}{n3}
\ncline{->}{n2}{n4}
\ncline{->}{n3}{n5}
\ncline{->}{n4}{n5}
\ncline{->}{n5}{n6}
\ncarc[linestyle=dashed,ncurv=2.5,arcangleA=-235,arcangleB=-235]{->}{n5}{n2}

\renewcommand{\newnode}[5]
{\rput(#1,#2){\rnode{m#3}{\psframebox[linestyle=dashed]{\phantom{#4}}}}
\rput(#1,#2){#5}
\rput(#1,#2){
\ifthenelse{#3>9}{\rput(-13,0){#3}}{\rput(-7,4){#3}}
}}

\tnode{90}{70}{1}{\s}{\small ${\cdots}$}
\newnode{90}{58}{7}{\s}{\small $\pvsid{SKIP}$}
\tnode{90}{46}{2}{\s}{\small ${\cdots}$}
\tnode{78}{34}{3}{\s}{\small ${\cdots}$}
\tnode{102}{34}{4}{\s}{\small ${\cdots}$}
\tnode{90}{22}{5}{\s}{\small ${\cdots}$}
\tnode{90}{10}{6}{\s}{\small ${\cdots}$}
\rput(90,3){\pvsid{prog'}}
\ncline[linestyle=dashed]{->}{m1}{m7}
\ncline[linestyle=dashed]{->}{m7}{m2}
\ncline{->}{m2}{m3}
\ncline{->}{m2}{m4}
\ncline{->}{m3}{m5}
\ncline{->}{m4}{m5}
\ncline{->}{m5}{m6}
\ncarc[ncurv=1.8,arcangleA=230,arcangleB=230,linestyle=dashed]{->}{m5}{m7}

\ncline[nodesep=2pt,linecolor=lightgray,linestyle=solid,border=1pt,dotsep=1pt]{<-}{n2}{m7}
\ncarc[arcangleA=-20,arcangleB=-10,nodesep=2pt,linecolor=gray,linestyle=dotted,border=1pt,dotsep=1pt]{->}{n5}{m7}
\ncline[nodesep=2pt,linecolor=gray,linestyle=dotted,border=1pt,dotsep=1pt]{->}{n1}{m7}


\rput(60,40){$N_P$}
\rput(60,64){$N_P$}
\rput(50,57){$N_S$}
\rput(12,61){$E$}
\rput(33,58){$E$}
\rput(87,64){$E_I$}
\rput(110,65){$E_I$}
\rput(94,52){$E_O$}
\end{pspicture}\\
Insert program point $7$ as the predecessor to program point $2$
\end{tabular}
\end{minipage}
}
\caption{An example of insertion of predecessors transformation}
\label{fig:insert-pred}
\end{figure}

Consider the two programs shown in Figure~\ref{fig:insert-pred} such that
$\pvsid{prog'} = \pvsid{IP(prog,succs,newpoints)}$.
The program \pvsid{prog'} is obtained by inserting the new program point 7 as the predecessor to program point 2. 
Let us use the ordered sequence $\langle 1, \ldots, 6 \rangle$ for indexing vectors/matrices associated
with \pvsid{prog}. Let $\pvsid{succs} = \langle 0,1,0,0,0,0 \rangle$
represent a set containing program point $2$. 
Let $\pvsid{newpoints} = \langle 0,0,0,0,0,0,1 \rangle$ denote the set of new program points
to be inserted as predecessors to \pvsid{succs}. 
The ordered sequence for indexing vectors/matrices for \pvsid{prog'}
is $\langle 1, \ldots, 6, 7 \rangle$.
The new program point 7 is placed at the end of the list.

We model a transformation of the control flow graph of a program by an application of \pvsid{IP}
as a node addition transformation (Definition~\ref{def:na}).
Given the arguments of \pvsid{IP}, we set up the adjacency matrices
for the relations $\statemap$, $N_S$, and $N_P$. For the transformation in Figure~\ref{fig:insert-pred},
we have the following:

\begin{center}
\begin{tabular}{ll}
\begin{minipage}{12cm}
The relation $\statemap$ is represented as the matrix shown here.
The rows correspond to program points $1, \ldots, 7$ (of the transformed program) and 
the columns correspond to program points $1, \ldots, 6$ (of the input program).
Note that since the new program point does not correspond to any program point in the input
graph, the last row (corresponding to program point 7) has all (boolean) $\zero$s.
\end{minipage}
&
\statemapMat
\end{tabular}
\end{center}

\begin{center}
\begin{tabular}{ll}
\begin{minipage}{12cm}
The matrix $\pvsid{Succs}$ is a $(1 \Tm 6)$ matrix which is appended to a $(6 \Tm 6)$
matrix containing all $\zero$s to get the $(7 \Tm 6)$ matrix \pvsid{N\_S} 
which maps program point 7 of \pvsid{prog'} to program point 2 of \pvsid{prog}.
\end{minipage}
&
\NS
\end{tabular}
\end{center}

\begin{center}
\begin{tabular}{ll}
\begin{minipage}{12cm}
The relation $N_P$ (denoted as a matrix) maps program points 1 and 5 (the predecessors of program point 2 in \pvsid{prog})
to program point 7. Given the vector \pvsid{succs}, we can identify the adjacency matrix $E$
of the incoming edges to the program points denoted by \pvsid{succs}.
The matrix $N_P$ is then obtained as $E \Mm \Tr{N_S}$.
\end{minipage}
& 
\NP
\end{tabular}
\end{center}

It can be verified that the matrices satisfy the conditions about the nature
of the corresponding relations stated in Definition~\ref{def:na}. 
For example, the correspondence matrix $\statemap$ denotes a partial 
(at least one row has all $\zero$s),
onto (each column has at least one non-zero element), 
and one-to-one relation (each column as well as each row has at most one non-zero element).
Clearly, the adjacency matrix $\mat'$ of the control flow graph
of \pvsid{prog'} can be obtained by substituting these matrices and the adjacency
matrix $\mat$ of the control flow graph of \pvsid{prog} in Definition~\ref{def:na}.

An insertion of predecessors transformation inserts SKIP statements at the newly inserted program points.
If the target of a goto or a conditional statement belongs to the set represented by \pvsid{succs}
then the target is updated to its new predecessor program point (identified using the $N_S$ relation).
For the example shown in Figure~\ref{fig:insert-pred}, the target of the conditional statement
at program point 5 will be updated to program point 7 in the transformed program.
All other statements remain unchanged.

\subsection{Insertion of assignments}

The insertion of assignments transformation primitive (\pvsid{IA}) takes a program,
a set of program points, a variable, and an expression as its arguments.
An application \pvsid{IA(prog,points,v,e)} of \pvsid{IA} to a program \pvsid{prog}
inserts an assignment $\pvsid{v} := \pvsid{e}$ at the program points denoted by 
the vector \pvsid{points}. The rest of the statements remain unchanged.
The control flow graph of the input program also remains unchanged and is defined
as an isomorphic graph transformation (Definition~\ref{def:im}).

\subsection{Replacement of expressions}

The replacement of expressions transformation primitive (\pvsid{RE}) takes a program,
a set of program points, an expression, and a variable as its arguments.
An application \pvsid{RE(prog,points,e,v)} of \pvsid{RE} to a program \pvsid{prog}
replaces the occurrences of the expression \pvsid{e} at the program points denoted by 
the vector \pvsid{points}. The rest of the statements remain unchanged.
The control flow graph of the input program also remains unchanged and is defined
as an isomorphic graph transformation (Definition~\ref{def:im}).

\section{Program transformations as transformations of Kripke structures}
\label{sec:kripke-transformations}

Kripke structures serve as a natural modeling paradigm when the properties of interest are temporal in nature.
In order to interpret temporal formulae, we abstract a program as a Kripke structure.
The control flow graph of the program gives the transition relation of the Kripke structure.
The atomic propositions of the Kripke structure correspond to the local data flow properties
and the labeling function corresponds to the valuations of the local properties.

In this section, we model program transformations as transformations of Kripke structures.
A program transformation may consist of a structural transformation of the control flow graph
and a content transformation of the control flow graph nodes (statements).
Similar to the program transformations,
a structural transformation of a Kripke structure is defined using the primitive graph transformations.
A content transformation is modeled by specifying modifications of atomic propositions and their labeling.

\subsection{Transformations of Kripke structures}
\begin{definition}
\label{def:kripke}
A \emph{Kripke structure} $\mykripke$ is a tuple $(\graph,\myprops,\mylabeling)$
where $\graph$ is a directed graph,
$\myprops$ is a set of atomic propositions, and
$\mylabeling: \myprops \rightarrow \bool_{n}$ 
is a labeling function that maps each atomic proposition in $P$ to a boolean vector of size $n$.
The $i$th element of the vector $\mylabeling(p)$ is (boolean) $1$ iff 
the proposition $p$ holds at the $i$th node.
$\bool_{n}$ denotes the set of boolean vectors of size $n$ where $n$ is the number of nodes in $\graph$.
\end{definition}

Following the usual assumption in program analysis
but without loss of generality, we assume that $\graph$ has a single entry and a single exit.
Let $\mat$ be the adjacency matrix of $\graph$.
We consider an implicit ordering of the states of $\graph$ for indexing vectors/matrices associated with $\mykripke$.
We do not explicitly represent the states of a Kripke structure.
We require $\mat$ to be a total relation i.e. every state should have an outgoing edge.
For an atomic proposition $p$, we use $\vec{p}$ to denote $\mylabeling(p)$.
There are two special atomic propositions $\varepsilon$
and $\omega$ such that $\entry$ and $\exit$ respectively hold only at the entry and the exit states.

\begin{definition}
\label{def:k-trans}
A \emph{Kripke transformation} maps a Kripke structure $\mykripke = (\graph,\myprops,\mylabeling)$ to 
a Kripke structure $\mykripke' = (\graph',\myprops',\mylabeling')$
and is defined as follows:
\begin{enumerate}
\item[{(1)}] A graph transformation that maps $\graph$ to $\graph'$. 
Let $\statemap$ as the correspondence matrix of the transformation.
\item[{(2)}] For each atomic proposition $p' \in \myprops'$, the labeling $\mylabeling'(p')$
(denoted as $\vec{p'}$) is defined in terms of the labeling of an atomic proposition $p \in \myprops$
and other atomic propositions of $\mykripke$ and $\mykripke'$ as follows:
\begin{equation}
\label{eqn:k-trans}
(\statemap \Mm \vec{p} \;\Pl\; \vec{u'}) \;\St\; \Ng{\vec{v'}} \;\St\; \vec{w'} = \vec{p'}
\text{ where } u', v', \text{ and } w' \text{ are temporal formulae}
\end{equation}
\begin{enumerate}
\item[{(a)}] $\statemap \Mm \vec{p}$ indicates the states in $\mykripke'$ which correspond
to the states labeled by $p$ in $\mykripke$,
\item[{(b)}] $\vec{u'}$ indicates the states in $\mykripke'$ which may not correspond
to the states labeled with $p$ in $\mykripke$ but are labeled with $p'$ in $\mykripke'$,
\item[{(c)}] $\vec{v'}$ indicates the states in $\mykripke'$ which are not labeled with $p'$, and
\item[
{(d)}] $\vec{w'}$ indicates the states in $\mykripke'$ which are labeled with $p'$ and also with $w'$.
\end{enumerate}
\end{enumerate}
\end{definition}

In general, the formulae $u'$, $v'$, and $w'$ in Equation (\ref{eqn:k-trans}) can be arbitrary \ctlbp\ formulae
possibly involving atomic propositions from both the input and the transformed Kripke structures.
If there are $m$ atomic propositions in $\mykripke'$ then we get a system of $m$ simultaneous
equations. If the system is h-monotonic~\cite{DBLP:conf/aplas/KanadeKS05} or simply monotonic
then a solution to the equations can be computed iteratively and it completely
determines the labeling of the transformed Kripke structure.

In our framework for verification of optimizations, we arrive at formulations of
Kripke transformations only indirectly by modeling primitive program transformations as primitive Kripke transformations.
Therefore, we do not solve formulations of Kripke transformations explicitly.
Further, definitions of atomic propositions of the transformed Kripke structure
involve only simple $\ctlbp$ formulae. 

We identify a Kripke transformation with the type of the graph transformation involved.
Thus we have node splitting, node merging, edge addition, edge deletion,
node addition, node deletion, and isomorphic Kripke transformations 
whose component graph transformations are as defined in Section~\ref{sec:graph-transformations}.


\subsection{Insertion of predecessors as node addition Kripke transformation}
\label{sec:ip-as-na}

In Section~\ref{sec:motivation}, we defined some local data flow properties viz.
\pvsid{Antloc}, \pvsid{Transp}, and \pvsid{Comp}. 
We now introduce a few more properties.
A variable \pvsid{v} is \pvsid{Def} at a program point if \pvsid{v} is assigned at the program point.
A variable \pvsid{v} is \pvsid{Use} at a program point if \pvsid{v} appears in an expression at the program point.
A variable \pvsid{v} and an expression \pvsid{e} satisfy \pvsid{AssignStmt} property at a program point if
the statement at the program point is $\pvsid{v} := \pvsid{e}$.

Let $\pvsid{prog}_1 = \pvsid{IP(prog,succs,newpoints)}$ be an application of the insertion of predecessors primitive. 
Since an insertion of predecessors transformation involves a node addition graph transformation,
we can abstract it as a node addition Kripke transformation (part 1 of Definition~\ref{def:k-trans}).

We now define the local data flow properties (atomic propositions)
of the transformed program $\pvsid{prog}_1$ in terms of the local data flow properties of 
the input program \pvsid{prog} (part 2 of Definition~\ref{def:k-trans}).
Recall that \pvsid{IP} inserts SKIP statements at the new program points
and does not change statements at any other program point.
Let $\statemap$ be the correspondence matrix for the (graph) transformation.

For any expression \pvsid{e} and a variable \pvsid{v} occurring in \pvsid{prog} (and also in $\pvsid{prog}_1$),
the correlation of the atomic propositions for the transformation is defined as follows:
\begin{equation*}
\label{eqn:ip-k-trans-2}
\setstretch{1.2}
\begin{array}{lcl}
{C} \Mm \pvsid{Antloc(prog,e)} & = & \pvsid{Antloc}(\pvsid{prog}_1,\pvsid{e})\\
{C} \Mm \pvsid{Transp(prog,e)} \;\Pl\; \pvsid{newpoints} & = & \pvsid{Transp}(\pvsid{prog}_1,\pvsid{e})\\
{C} \Mm \pvsid{Comp(prog,e)} & = & \pvsid{Comp}(\pvsid{prog}_1,\pvsid{e})\\
{C} \Mm \pvsid{Mod(prog,e)} & = & \pvsid{Mod}(\pvsid{prog}_1,\pvsid{e})\\
{C} \Mm \pvsid{Use(prog,v)} & = & \pvsid{Use}(\pvsid{prog}_1,\pvsid{v})\\
{C} \Mm \pvsid{Def(prog,v)} & = & \pvsid{Def}(\pvsid{prog}_1,\pvsid{v})\\
{C} \Mm \pvsid{AssignStmt(prog,v,e)} & = & \pvsid{AssignStmt}(\pvsid{prog}_1,\pvsid{v},\pvsid{e})
\end{array}
\end{equation*}

Note that since the new program points (denoted by \pvsid{newpoints}) contain SKIP statement
in $\pvsid{prog}_1$, for any expression \pvsid{e}, $\pvsid{Transp}(\pvsid{prog}_1,\pvsid{e})$ holds at \pvsid{newpoints}.

Let $\pvsid{succs'} = \pvsid{C} \Mm \pvsid{succs}$ be a new atomic proposition
denoting the program points in $\pvsid{prog}_1$ which correspond to the \pvsid{succs} program points in \pvsid{prog}.
For each program point $i$ from \pvsid{succs'},
there exists a program point $j$ from \pvsid{newpoints} such that $i$ is the only predecessor of $j$ and
$j$ is the only successor of $i$. Hence,
\begin{equation*}
\label{eqn:ip-k-trans-3}
\setstretch{1.2}
\begin{array}{lcl}
\pvsid{AX}(\pvsid{prog}_1\pvsid{`cfg}, \pvsid{succs'}) & = & \pvsid{newpoints}\\
\pvsid{AY}(\pvsid{prog}_1\pvsid{`cfg}, \pvsid{newpoints}) & = & \pvsid{succs'}
\end{array}
\end{equation*}

Further, for each primitive program transformation, we lemmatize several other properties
of the transformation or the transformed program. For instance, we introduce a lemma
to state that the statements at the newly inserted program points \pvsid{newpoints} are SKIP statements.
The correctness of such lemmas can be proved from the definitions of the primitives.

\subsection{Insertion of assignments as isomorphic transformation}
\label{sec:ia-as-im}

\newcommand{\myarrayscale}[1]{\scalebox{1}{#1}}

Let $\pvsid{prog}_1 = \pvsid{IA(prog,points,v,e)}$ be an application of the insertion of assignments primitive.
Since an insertion of assignments transformation does not change the control flow graph of the input program,
we can abstract it as an isomorphic Kripke transformation (part 1 of Definition~\ref{def:k-trans}).

Recall that \pvsid{IA} inserts an assignment $\pvsid{v} := \pvsid{e}$ at the program points denoted by \pvsid{points}
and does not change statements at other program points.
We assume that the statement at any program point in \pvsid{points} in \pvsid{prog} is SKIP.
This condition is defined as a part of the soundness conditions for \pvsid{IA}.

Below we correlate the local data flow properties of the input program \pvsid{prog}
and the transformed program $\pvsid{prog}_1$ (part 2 of Definition~\ref{def:k-trans}). \pvsid{e1} is an expression
and \pvsid{v1} is a variable in the following:
\begin{equation*}
\label{eqn:ia-k-trans-2}
\setstretch{1.2}
\begin{array}{lcll}
\pvsid{Antloc(prog,e1)} \;\Pl\; \pvsid{points} & = & \pvsid{Antloc}(\pvsid{prog}_1,\pvsid{e1}) & 
\myarrayscale{\text{if } \pvsid{e} = \pvsid{e1}}\\
\pvsid{Antloc(prog,e1)} & = & \pvsid{Antloc}(\pvsid{prog}_1,\pvsid{e1}) & 
\myarrayscale{otherwise}\\[1.5mm]
\pvsid{Transp(prog,e1)} \;\Mn\; \pvsid{points} & = & \pvsid{Transp}(\pvsid{prog}_1,\pvsid{e1}) & 
\myarrayscale{\text{if \pvsid{v} is an operand of \pvsid{e1}}}\\
\pvsid{Transp(prog,e1)} & = & \pvsid{Transp}(\pvsid{prog}_1,\pvsid{e1}) & 
\myarrayscale{\text{otherwise}}\\[1.5mm]
\pvsid{Comp(prog,e1)} \;\Mn\; \pvsid{points} & = & \pvsid{Comp}(\pvsid{prog}_1,\pvsid{e1}) & 
\myarrayscale{\text{if \pvsid{v} is an operand of \pvsid{e1}}}\\
\pvsid{Comp(prog,e1)} \;\Pl\; \pvsid{points} & = & \pvsid{Comp}(\pvsid{prog}_1,\pvsid{e1}) & 
\myarrayscale{\text{else if } \pvsid{e} = \pvsid{e1}}\\
\pvsid{Comp(prog,e1)} & = & \pvsid{Comp}(\pvsid{prog}_1,\pvsid{e1}) & 
\myarrayscale{otherwise}\\[1.5mm]
\pvsid{Mod(prog,e1)} \;\Pl\; \pvsid{points} & = & \pvsid{Mod}(\pvsid{prog}_1,\pvsid{e1}) & 
\myarrayscale{\text{if \pvsid{v} is an operand of \pvsid{e1}}}\\ 
\pvsid{Mod(prog,e1)} & = & \pvsid{Mod}(\pvsid{prog}_1,\pvsid{e1}) & 
\myarrayscale{otherwise}\\[1.5mm]
\pvsid{Use(prog,v1)} \;\Pl\; \pvsid{points} & = & \pvsid{Use}(\pvsid{prog}_1,\pvsid{v1}) & 
\myarrayscale{\text{if \pvsid{v} is an operand of \pvsid{e1}}}\\
\pvsid{Use(prog,v1)} & = & \pvsid{Use}(\pvsid{prog}_1,\pvsid{v1}) & 
\myarrayscale{otherwise}\\[1.5mm]
\pvsid{Def(prog,v1)} \;\Pl\; \pvsid{points} & = & \pvsid{Def}(\pvsid{prog}_1,\pvsid{v1}) & 
\scalebox{0.9}{\text{if } \pvsid{v} = \pvsid{v1}}\\ 
\pvsid{Def(prog,v1)} & = & \pvsid{Def}(\pvsid{prog}_1,\pvsid{v1}) & 
\myarrayscale{otherwise}
\end{array}
\end{equation*}
\renewcommand{\myarrayscale}[1]{\scalebox{0.96}{#1}}
\begin{equation*}
\label{eqn:ia-k-trans-3}
\setstretch{1.2}
\begin{array}{lcll}
\myarrayscale{\pvsid{AssignStmt(prog,v1,e1)}} \;\Pl\; \scalebox{0.9}{\pvsid{points}} & = & 
\pvsid{AssignStmt}(\pvsid{prog}_1,\pvsid{v1},\pvsid{e1})&
\myarrayscale{\text{if \pvsid{v} \!=\! \pvsid{v1} and \pvsid{e} \!=\! \pvsid{e1}}}\\
\pvsid{AssignStmt(prog,v1,e1)} & = & \pvsid{AssignStmt}(\pvsid{prog}_1,\pvsid{v1},\pvsid{e1}) & 
\myarrayscale{otherwise}
\end{array}
\end{equation*}

\subsection{Replacement of expressions as isomorphic transformation}
\label{sec:re-as-im}

Let $\pvsid{prog}_1 = \pvsid{RE(prog,points,e,v)}$ be an application of the replacement of expressions primitive.
Since a replacement of expressions transformation does not change the control flow graph of the input program,
we can abstract it as an isomorphic Kripke transformation (part 1 of Definition~\ref{def:k-trans}).

Recall that \pvsid{RE} replaces the occurrences of expression \pvsid{e} at the program points 
denoted by \pvsid{points} by variable \pvsid{v} and does not change statements at other program points.

\renewcommand{\myarrayscale}[1]{\scalebox{1}{#1}}
Below we correlate local data flow properties of the input program \pvsid{prog}
and the transformed program $\pvsid{prog}_1$ (part 2 of Definition~\ref{def:k-trans}).
\pvsid{e1} is an expression and \pvsid{v1} is a variable in the following:
\begin{equation*}
\label{eqn:re-k-trans-2}
\setstretch{1.2}
\begin{array}{lcll}
\pvsid{Antloc(prog,e1)} \;\Mn\; \pvsid{points} & = & \pvsid{Antloc}(\pvsid{prog}_1,\pvsid{e1}) & 
\myarrayscale{\text{if } \pvsid{e} = \pvsid{e1}}\\
\pvsid{Antloc(prog,e1)} \;\Pl\; \pvsid{points} & = & \pvsid{Antloc}(\pvsid{prog}_1,\pvsid{e1}) &
\text{if the expression \pvsid{e1} is just the variable \pvsid{v}}\\
\pvsid{Antloc(prog,e1)} & = & \pvsid{Antloc}(\pvsid{prog}_1,\pvsid{e1}) &
\myarrayscale{otherwise}\\[1.5mm]
\pvsid{Transp(prog,e1)} & = & \pvsid{Transp}(\pvsid{prog}_1,\pvsid{e1})\\[1.5mm]
\pvsid{Comp(prog,e1)} \;\Mn\; \pvsid{points} & = & \pvsid{Comp}(\pvsid{prog}_1,\pvsid{e1}) &
\myarrayscale{\text{if } \pvsid{e} = \pvsid{e1}}\\ 
\pvsid{Comp(prog,e1)} & \leq & \pvsid{Comp}(\pvsid{prog}_1,\pvsid{e1}) &
\text{if the expression \pvsid{e1} is just the variable \pvsid{v}}\\
\pvsid{Comp(prog,e1)} & = & \pvsid{Comp}(\pvsid{prog}_1,\pvsid{e1}) &
\myarrayscale{otherwise}\\[1.5mm]
\pvsid{Mod(prog,e1)} & = & \pvsid{Mod}(\pvsid{prog}_1,\pvsid{e1}) & \\[1.5mm]
\pvsid{Use(prog,v1)} \;\Mn\; \pvsid{points} & = & \pvsid{Use}(\pvsid{prog}_1,\pvsid{v1}) &
\myarrayscale{\text{if \pvsid{v1} is an operand of \pvsid{e}}}\\ 
\pvsid{Use(prog,v1)} \;\Pl\; \pvsid{points} & = & \pvsid{Use}(\pvsid{prog}_1,\pvsid{v1}) &
\myarrayscale{\text{else if } \pvsid{v} \!=\! \pvsid{v1}}\\
\pvsid{Use(prog,v1)} & = & \pvsid{Use}(\pvsid{prog}_1,\pvsid{v1}) &
\myarrayscale{otherwise}\\[1.5mm]
\pvsid{Def(prog,v1)} & = & \pvsid{Def}(\pvsid{prog}_1,\pvsid{v1})
\end{array}
\end{equation*}

Note the inequality in the second correlation of the \pvsid{Comp} property.
Suppose the expression $\pvsid{e1} \neq \pvsid{e}$ is just the variable \pvsid{v}.
If \pvsid{e1} is \pvsid{Comp} at a program point $i$ in the input program
then it is also \pvsid{Comp} at $i$ in the transformed program.
Therefore, $\pvsid{Comp(prog,e1)} \leq \pvsid{Comp}(\pvsid{prog}_1,\pvsid{e1})$.
At the program points denoted by \pvsid{points}, the expression \pvsid{e1} is \pvsid{Antloc} in the transformed program.
However, depending on whether the left-hand side variable at a program point
is \pvsid{v} or not, it may or may not be \pvsid{Comp} at that point. 
Since we do not require a complete characterization of $\pvsid{Comp}(\pvsid{prog}_1,\pvsid{e1})$ in the proofs,
we prefer the inequality as the correlation between \pvsid{Comp} program points.
For similar reasons, we prefer an inequality in the second correlation of \pvsid{AssignStmt} property:
\begin{equation*}
\label{eqn:re-k-trans-3}
\setstretch{1.2}
\begin{array}{lcll}
\myarrayscale{\pvsid{AssignStmt(prog,v1,e1)}} \;\Mn\; \scalebox{0.9}{\pvsid{points}} & = & \pvsid{AssignStmt(prog',v1,e1)} &
\myarrayscale{\text{if } \pvsid{e} = \pvsid{e1}}\\ 
\pvsid{AssignStmt(prog,v1,e1)} & \leq & \pvsid{AssignStmt(prog',v1,e1)} &
\text{if the expression \pvsid{e1} is just the variable \pvsid{v1}}\\
\pvsid{AssignStmt(prog,v1,e1)} & = & \pvsid{AssignStmt(prog',v1,e1)} &
\myarrayscale{otherwise}
\end{array}
\end{equation*}

\section{Temporal transformation logic}
\label{sec:ttl}

In this section, we introduce the logic TTL. We present the syntax and the semantics of TTL operators. 
We then define the axioms and the basic inference rules.
For each type of primitive Kripke transformations, we present the inference rules
to correlate temporal properties between a Kripke structure and its transformation under the specific primitive transformation.

\subsection{Syntax and semantics}
We give boolean matrix algebraic semantics to $\ctlbp$ formulae. For a formula $\varphi$, we denote the set of states
where the the formula is satisfied by a boolean vector $\vec{\varphi}$.
For an atomic proposition $p$, the labeling $\mylabeling(p)$ is denoted by a boolean vector $\vec{p}$.
The one-step temporal modalities can be defined by matrix multiplication. For example,
the semantics of the formula $\EX(\varphi)$ is given by $\mat \;\Mm\; \vec{\varphi}$ where
$\mat$ is the boolean adjacency matrix of the transition relation of the Kripke structure
and the operator ``$\Mm$'' denotes matrix multiplication.
The semantics of path operators (until, globally, etc.) are specified by fixed points.
Appendix~\ref{sec:ctlbp} gives the syntax and semantics of $\ctlbp$ formulae.

Consider a Kripke structure $\mykripke$ and its transformed version $\mykripke'$.
Let $\varphi$ and $\varphi'$ be $\ctlbp$ formulae defined respectively over Kripke structures $\mykripke$ and $\mykripke'$.
As a convention, we use primed symbols for a transformed Kripke structure and
unprimed symbols for an input Kripke structure. 
Correlations between $\ctlbp$ formulae $\varphi$ and $\varphi'$ are denoted the following TTL formulae:
\begin{equation}
\label{eqn:ttl-syntax}
\varphi \timp \varphi', \varphi \Rightarrow \varphi', \text{ and } \varphi \leftarrow \varphi'
\end{equation}

The model of a TTL formula is the pair $\langle \mykripke,\mykripke' \rangle$ of Kripke structures.
Let $\statemap$ be the correspondence matrix that relates the nodes of $\mykripke'$ and $\mykripke$.
The semantics of the TTL formulae are defined as follows:
\begin{equation}
\label{eqn:ttl-sem}
\setstretch{1.4}
\begin{array}{lcl}
\langle \mykripke, \mykripke' \rangle \models \varphi \timp \varphi' 
& \text{iff} &
\statemap \Mm\; \vec{\varphi} \leq \vec{\varphi'}\\
\pair{\mykripke}{\mykripke'} \models \varphi \Rightarrow \varphi' 
& \text{iff} &
\statemap \Mm\; \vec{\varphi} = \vec{\varphi'}\\
\pair{\mykripke}{\mykripke'} \models \varphi \leftarrow \varphi' 
& \text{iff} &
\statemap \Mm\; \vec{\varphi} \geq \vec{\varphi'}
\end{array}
\end{equation}

\newcommand{\Anl}{\ensuremath{\mathcal{A}_f}}
\newcommand{\mvdash}{\ensuremath{\vdash_{\mykripke}}}
\newcommand{\mmvdash}{\ensuremath{\vdash_{\mykripke'}}}
\newcommand{\fvdash}{\ensuremath{\vdash_f}}

\subsection{Axioms and basic inference rules}
\label{sec:ax-basics}

\newcommand{\jsep}{\\[-3mm]}

\newcommand{\AXMone}{
\begin{tabular}{c}
 \begin{tabular}{@{\;}c@{\;}}
 \phantom{nothing}\\
 \cline{1-1}\jsep
 $p \timp p'$
 \end{tabular}\\[1mm]
 (AXM1 if $\statemap \Mm \vec{p} \leq \vec{p'}$)
\end{tabular}
}

\newcommand{\AXMtwo}{
\begin{tabular}{c}
\begin{tabular}{@{\;}c@{\;}}
 \phantom{nothing}\\
 \cline{1-1}\jsep
 $p \Rightarrow p'$
 \end{tabular}\\[1mm]
 (AXM2 if $\statemap \Mm \vec{p} = \vec{p'}$)
\end{tabular}
}

\newcommand{\AXMthree}{
\begin{tabular}{c}
 \begin{tabular}{@{\;}c@{\;}}
 \phantom{nothing}\\
 \cline{1-1}\jsep
 $p \leftarrow p'$
 \end{tabular}\\[1mm]
 (AXM3 if $\statemap \Mm \vec{p} \geq \vec{p'}$)
\end{tabular}
}

\newcommand{\FCone}{
\begin{tabular}{@{\;}c@{\;}}
\begin{tabular}{@{}cc@{}}
$\varphi \!\Imp\! \psi$ & $\psi \!\timp\! \psi'$\\
\end{tabular}\\
\cline{1-1}\jsep
$\varphi \!\timp\! \psi'$\\[1mm]
(FC1)
\end{tabular}
}

\newcommand{\FCtwo}{
\begin{tabular}{@{\;}c@{\;}}
\begin{tabular}{@{}cc@{}}
$\varphi \!\timp\! \varphi'$ & $\varphi' \!\Imp\! \psi'$\\
\end{tabular}\\
\cline{1-1}\jsep
$\varphi \!\timp\! \psi'$\\[1mm]
(FC2)
\end{tabular}
}

\newcommand{\FCthree}{
\begin{tabular}{@{\;}c@{\;}}
\begin{tabular}{@{}c@{}}
$\varphi \!\timp\! \psi'$\\
\end{tabular}\\
\cline{1-1}\jsep
$\exists \varphi': \varphi \!\Rightarrow\! \varphi' \wedge \varphi' \!\Imp\! \psi'$\\[1mm]
(FC3)
\end{tabular}
}

\newcommand{\MPone}{
\scalebox{1}{
\begin{tabular}{c}
\begin{tabular}{@{\;}c@{\;}}
 \begin{tabular}{@{}cc@{}}
 $\top \Imp \varphi$ & $\varphi \timp \varphi'$\\
 \end{tabular}\\
 \cline{1-1}\jsep
 $\top' \Imp \varphi'$
\end{tabular}\\[1mm]
(MP1 for NS, NM, EA, ED, IM)
 \end{tabular}}}

\newcommand{\MPtwo}{
\scalebox{1}{
\begin{tabular}{c}
\begin{tabular}{@{\;}c@{\;}}
 \begin{tabular}{@{}cc@{}}
 $\top \Imp \varphi$ & $\varphi \timp \varphi'$\\
 \end{tabular}\\
 \cline{1-1}\jsep
 $\top' \Imp (\varphi' \vee \newstates)$
\end{tabular}\\[1mm]
 (MP2 for NA)
 \end{tabular}}}

\newcommand{\MPthree}{
\scalebox{1}{
\begin{tabular}{c}
\begin{tabular}{@{\;}c@{\;}}
 \begin{tabular}{@{}cc@{}}
 $\top \Imp (\varphi \!\vee\! \delstates)$ & $\varphi \timp \varphi'$\\
 \end{tabular}\\
 \cline{1-1}\jsep
 $\top' \Imp \varphi'$
\end{tabular}\\[1mm]
(MP3 for ND)
 \end{tabular}}}

\newcommand{\MTone}{
\scalebox{1}{
\begin{tabular}{c}
\begin{tabular}{@{\;}c@{\;}}
 \begin{tabular}{@{}c@{\;\;}c@{}}
$\varphi \!\timp\! \varphi'$ & $\varphi' \!\Imp\! \bot'$\\
 \end{tabular}\\
 \cline{1-1}\jsep
$\varphi \!\Imp\! \bot$
\end{tabular}\\[1mm]
(MT1 for NS, NM, EA, ED, IM)
 \end{tabular}}}

\newcommand{\MTtwo}{
\scalebox{1}{
\begin{tabular}{c}
\begin{tabular}{@{\;}c@{\;}}
 \begin{tabular}{@{}c@{\;\;}c@{}}
$\varphi \!\timp\! \varphi'$ & $(\varphi' \!\wedge\! \neg \eta') \!\Imp\! \bot'$\\
 \end{tabular}\\
 \cline{1-1}\jsep
$\varphi \!\Imp\! \bot$
\end{tabular}\\[1mm]
 (MT2 for NA)
 \end{tabular}}}

\newcommand{\MTthree}{
\scalebox{1}{
\begin{tabular}{c}
\begin{tabular}{@{\;}c@{\;}}
 \begin{tabular}{@{}c@{\;\;}c@{}}
$\varphi \!\timp\! \varphi'$ & $\varphi' \!\Imp\! \bot'$\\
 \end{tabular}\\
 \cline{1-1}\jsep
$(\varphi \!\wedge\! \neg \delstates) \Imp \bot$
\end{tabular}\\[1mm]
(MT3 for ND)
 \end{tabular}}}

\newcommand{\DIone}{
\scalebox{1}{
\begin{tabular}{c}
\begin{tabular}{@{\;}c@{\;}}
 \begin{tabular}{@{}cc@{}}
 $\varphi \timp \varphi'$ & $\psi \timp \varphi'$\\
 \end{tabular}\\
 \cline{1-1}\jsep
 $(\varphi \vee \psi) \timp \varphi'$
\end{tabular}\\[1mm]
(DI1)
 \end{tabular}}}

\newcommand{\DItwo}{
\scalebox{1}{
\begin{tabular}{c}
\begin{tabular}{@{\;}c@{\;}}
 \begin{tabular}{@{}c@{}}
 $\varphi \timp \varphi'$\\
 \end{tabular}\\
 \cline{1-1}\jsep
 $\varphi \timp (\varphi' \vee \psi')$
\end{tabular}\\[1mm]
(DI2)
 \end{tabular}}}

\newcommand{\CIone}{
\scalebox{1}{
\begin{tabular}{c}
\begin{tabular}{@{\;}c@{\;}}
 \begin{tabular}{@{}c@{}}
 $\varphi \timp \varphi'$\\
 \end{tabular}\\
 \cline{1-1}\jsep
 $(\varphi \wedge \psi) \timp \varphi'$
\end{tabular}\\[1mm]
(CI1)
 \end{tabular}}}

\newcommand{\CItwo}{
\scalebox{1}{
\begin{tabular}{c}
\begin{tabular}{@{\;}c@{\;}}
 \begin{tabular}{@{}cc@{}}
 $\varphi \timp \varphi'$ & $\varphi \timp \psi'$\\
 \end{tabular}\\
 \cline{1-1}\jsep
 $\varphi \timp (\varphi' \wedge \psi')$
\end{tabular}\\[1mm]
(CI2)
 \end{tabular}}}

\newcommand{\NIone}{
\scalebox{1}{
\begin{tabular}{c}
\begin{tabular}{@{\;}c@{\;}}
 \begin{tabular}{@{}c@{}}
 $\varphi \timp \varphi'$\\
 \end{tabular}\\
 \cline{1-1}\jsep
 $\neg\; \varphi \leftarrow \neg\; \varphi'$
\end{tabular}\\[1mm]
(NI1 if $\statemap \Mm \Ng{\vec{z}} \geq \Ng{\statemap \Mm \vec{z}}$)
 \end{tabular}}}

\newcommand{\NItwo}{
\scalebox{1}{
\begin{tabular}{c}
\begin{tabular}{@{\;}c@{\;}}
 \begin{tabular}{@{}c@{}}
 $\varphi \Rightarrow \varphi'$\\
 \end{tabular}\\
 \cline{1-1}\jsep
 $\neg\; \varphi \timp \neg\; \varphi'$
\end{tabular}\\[1mm]
(NI2 if $\statemap \Mm \Ng{\vec{z}} \leq \Ng{\statemap \Mm \vec{z}}$)
 \end{tabular}}}

\newcommand{\IRone}{
\scalebox{1}{
\begin{tabular}{c}
\begin{tabular}{@{\;}c@{\;}}
 \begin{tabular}{@{}c@{}}
 $\varphi \Rightarrow \varphi'$\\
 \end{tabular}\\
 \cline{1-1}\jsep
 $\varphi \timp \varphi'$
\end{tabular}\\[1mm]
(IR1)
 \end{tabular}}}

\newcommand{\IRtwo}{
\scalebox{1}{
\begin{tabular}{c}
\begin{tabular}{@{\;}c@{\;}}
 \begin{tabular}{@{}c@{}}
 $\varphi \Rightarrow \varphi'$\\
 \end{tabular}\\
 \cline{1-1}\jsep
 $\varphi \leftarrow \varphi'$
\end{tabular}\\[1mm]
(IR2)
 \end{tabular}}}

\newcommand{\IT}{
\scalebox{1}{
\begin{tabular}{c}
\begin{tabular}{@{\;}c@{\;}}
 \begin{tabular}{@{}ccc@{}}
 $\varphi \Imp \psi$ & $\varphi \Rightarrow \varphi'$ & $\psi \Rightarrow \psi'$\\
 \end{tabular}\\
 \cline{1-1}\jsep
 $\varphi' \Imp \psi'$
\end{tabular}\\[1mm]
(IT)
 \end{tabular}}}

\newcommand{\MC}{
\scalebox{1}{
\begin{tabular}{c}
\begin{tabular}{@{\;}c@{\;}}
 \begin{tabular}{@{}cc@{}}
 $\varphi \leftarrow \varphi'$ & $\varphi \timp \psi'$\\
 \end{tabular}\\
 \cline{1-1}\jsep
 $\varphi' \Imp \psi'$
\end{tabular}\\[1mm]
(MC)
 \end{tabular}}}

\newcommand{\BCone}{
\scalebox{1}{
\begin{tabular}{c}
\begin{tabular}{@{\;}c@{\;}}
 \begin{tabular}{@{}cc@{}}
 $\varphi \leftarrow \varphi'$ & $\varphi' \subset \psi'$\\
 \end{tabular}\\
 \cline{1-1}\jsep
 $\varphi \leftarrow \psi'$
\end{tabular}\\[1mm]
(BC1)
 \end{tabular}}}

\newcommand{\BCtwo}{
\scalebox{1}{
\begin{tabular}{c}
\begin{tabular}{@{\;}c@{\;}}
 \begin{tabular}{@{}cc@{}}
 $\varphi \subset \psi$ & $\psi \leftarrow \psi'$\\
 \end{tabular}\\
 \cline{1-1}\jsep
 $\varphi \leftarrow \psi'$
\end{tabular}\\[1mm]
(BC2)
 \end{tabular}}}

\newcommand{\DE}{
\scalebox{1}{
\begin{tabular}{c}
\begin{tabular}{@{\;}c@{\;}}
 \begin{tabular}{@{}c@{}}
 $(\varphi \vee \psi) \timp \varphi'$\\
 \end{tabular}\\
 \cline{1-1}\jsep
 $\varphi \timp \varphi'$
\end{tabular}\\[1mm]
(DE)
 \end{tabular}}}

\newcommand{\CE}{
\scalebox{1}{
\begin{tabular}{c}
\begin{tabular}{@{\;}c@{\;}}
 \begin{tabular}{@{}c@{}}
 $\varphi \timp (\varphi' \wedge \psi')$\\
 \end{tabular}\\
 \cline{1-1}\jsep
 $\varphi \!\timp\! \varphi'$
\end{tabular}\\[1mm]
(CE)
 \end{tabular}}}

\newcommand{\IRT}{
\scalebox{1}{
\begin{tabular}{c}
\begin{tabular}{@{\;}c@{\;}}
 \begin{tabular}{@{}ccc@{}}
 $\varphi' \!\Imp\! \psi'$ & $\varphi \!\Rightarrow\! \varphi'$ & $\psi \!\Rightarrow\! \psi'$\\
 \end{tabular}\\
 \cline{1-1}\jsep
 $\varphi \!\Imp\! \psi$
\end{tabular}\\[1mm]
(IRT if $\vec{z} = \Tr{\statemap} \Mm \statemap \Mm \vec{z}$)
 \end{tabular}}}

\newcommand{\NEone}{
\scalebox{1}{
\begin{tabular}{c}
\begin{tabular}{@{\;}c@{\;}}
 \begin{tabular}{@{}c@{}}
 $\neg\; \varphi \timp \neg\; \varphi'$\\
 \end{tabular}\\
 \cline{1-1}\jsep
 $\varphi \leftarrow \varphi'$
\end{tabular}\\[1mm]
(NE1 if $\statemap \Mm \Ng{\vec{z}} \geq \Ng{\statemap \Mm \vec{z}}$)
 \end{tabular}}}

\newcommand{\NEtwo}{
\scalebox{1}{
\begin{tabular}{c}
\begin{tabular}{@{\;}c@{\;}}
 \begin{tabular}{@{}c@{}}
 $\neg\; \varphi \leftarrow \neg\; \varphi'$\\
 \end{tabular}\\
 \cline{1-1}\jsep
 $\varphi \timp \varphi'$
\end{tabular}\\[1mm]
(NE2 if $\statemap \Mm \Ng{\vec{z}} \leq \Ng{\statemap \Mm \vec{z}}$)
 \end{tabular}}}

\begin{figure}
\centering
\fbox{
\begin{minipage}{\minipagewidth}
\centering
\begin{tabular}{@{}c@{}}
\begin{tabular}{@{}c@{}}
\begin{tabular}{@{}ccc@{}}
\AXMone & \AXMtwo & \AXMthree
\end{tabular}\\[5mm]
{\bf (a) Axioms for introducing atomic judgments}\\[2mm]
\begin{tabular}{c}
\begin{tabular}{@{}cc@{}}
\FCone & \FCtwo 
\end{tabular}\\[5mm]
{\bf (b) Forward Chaining}\\[2mm]
\end{tabular}\\
\begin{tabular}{@{}ccc@{}}
\MPone & \MPtwo & \MPthree
\end{tabular}\\[5mm]
{\bf (c) Modus Ponens for TTL}\\[2mm]
\begin{tabular}{@{}ccc@{}}
\MTone & \MTtwo & \MTthree
\end{tabular}\\[5mm]
{\bf (d) Modus Tollens for TTL}\\[2mm]
\begin{tabular}{@{}cc@{}}
\begin{tabular}{@{}cc@{}}
\DIone & \DItwo
\end{tabular}
&
\begin{tabular}{@{}cc@{}}
\CIone & \CItwo
\end{tabular}\\[5mm]
{\bf (e) Disjunct Introduction} & {\bf (f) Conjunct Introduction}\\[2mm]
\end{tabular}\\
\begin{tabular}{@{}cc@{}}
\NIone & \NItwo
\end{tabular}\\[5mm]
{\bf (g) Negation Introduction}\\[2mm]
\begin{tabular}{@{}cc@{}}
\begin{tabular}{@{}cc@{}}
\IRone & \IRtwo
\end{tabular}
& \IT\\[5mm]
{\bf (h) Implication Relaxation} & {\bf (i) Implication Transfer}\\[2mm]
\end{tabular}\\
\begin{tabular}{@{}cc@{}}
\MC &
\begin{tabular}{@{}cc@{}}
\BCone & \BCtwo
\end{tabular}\\[5mm]
{\bf (j) Mixed Chaining} & {\bf (k) Backward Chaining}\\[2mm]
\end{tabular}
\end{tabular}\\
\begin{tabular}{@{}cc@{}}
\begin{tabular}{@{}cc@{}}
\DE & \CE
\end{tabular} 
& \IRT\\[5mm]
{\bf (l) Disjunct/Conjunct Elimination} & {\bf (m) Implication Reverse Transfer}\\[2mm]
\end{tabular}\\
\begin{tabular}{@{}cc@{}}
\NEone & \NEtwo
\end{tabular}\\[5mm]
{\bf (n) Negation Elimination}
\end{tabular}
\end{minipage}}
\caption{Axioms and basic inference rules}
\label{fig:ax}
\end{figure}

\newcommand{\InfRule}{
\scalebox{1}{
\begin{tabular}{c}
\begin{tabular}{@{\;}c@{\;}}
 \begin{tabular}{@{}ccc@{}}
 $\alpha_1$ & $\cdots$ & $\alpha_k$
 \end{tabular}\\
 \cline{1-1}\jsep
 $\beta$
\end{tabular}\\
(name of the rule and its scope)
 \end{tabular}}}

An inference rule is represented as follows:
\begin{center}
\InfRule
\end{center}
where the judgments $\alpha_1,\ldots,\alpha_k$ above the line denote
the premises and the judgment below the line denotes the conclusion.
A premise $\alpha_i$ can be a $\ctlbp$ formula defined either on an input Kripke 
structure or on its transformation, or it can be a TTL formula.
The conclusion $\beta$ is a TTL formula. An inference rule has a name.
The scope of the rule specifies the type of Kripke transformations for which the inference is valid.
We abbreviate the transformation primitives by their first letters e.g. node splitting is abbreviated as NS.
If no specific scope is given then the rule is applicable to all primitive transformations.

Consider a pair $\langle \mykripke,\mykripke' \rangle$ of Kripke structures
where $\mykripke'$ is a transformation of $\mykripke$ according to some transformation
within the scope of an inference rule. The inference rule states that if the premises hold
for $\langle \mykripke,\mykripke' \rangle$ then the conclusion is also holds for $\langle \mykripke,\mykripke' \rangle$.

In Figure~\ref{fig:ax}, we give the axioms and basic inference rules of the logic.
The correspondence matrix for a pair of Kripke structures $\langle \mykripke,\mykripke' \rangle$
is denoted by $\statemap$. The axiom AXM1 is applicable if an atomic proposition $p$ defined on $\mykripke$
and an atomic proposition $p'$ defined on $\mykripke'$ satisfy the condition $\statemap \Mm \vec{p} \leq \vec{p'}$.
In that case, $p \timp p'$ follows directly. Note that the condition $\statemap \Mm \vec{p} \leq \vec{p'}$
is not part of the logic. The axiom allows introduction of atomic judgments 
(judgments involving atomic propositions only) in a proof. Similarly, AXM2 and AXM3.

Proofs of the basic inference rules are straightforward. We only discuss a few interesting ones here.
The rule MP2 is applicable to node addition transformations only.
The atomic proposition $\newstates$ denotes the set of new nodes in the transformed Kripke structure $\mykripke'$.
It is defined as $\vec{\newstates} = \vecone \Mn \statemap \Mm \vecone$.
MP2 is sound because 
(1)~$\vecone \leq \vec{\varphi}$ : given,
(2)~$\statemap \Mm \vecone \leq \statemap \Mm \vec{\varphi}$ : monotonic,
(3)~$\statemap \Mm \vec{\varphi} \leq \vec{\varphi'}$ : given,
(4)~$\statemap \Mm \vecone \leq \vec{\varphi'}$ : transitive,
(5)~$\statemap \Mm \vecone \Pl (\vecone \Mn \statemap \Mm \vecone) \leq \vec{\varphi'} \Pl \vec{\newstates}$ : 
monotonic and definition of $\newstates$, and
(6)~$\vecone \leq \vec{\varphi'} \Pl \vec{\newstates}$.

The rule MP3 is applicable to node deletion transformations only.
The atomic proposition $\delstates$ denotes the set of nodes that are deleted from the input Kripke structure.
MP3 is sound because
(1)~$\vecone \leq \vec{\varphi} \Pl \vec{\delstates}$ : given,
(2)~$\vecone \Mn \vec{\delstates} \leq \vec{\varphi}$,
(3)~$\statemap \Mm (\vecone \Mn \vec{\delstates}) \leq \statemap \Mm \vec{\varphi}$ : monotonic,
(4)~$\statemap \Mm \vec{\varphi} \leq \vec{\varphi'}$ : given, 
(5)~$\statemap \Mm (\vecone \Mn \vec{\delstates}) = \vecone$ : $\statemap$ is a total
but not onto function and does not relate the nodes of $\mykripke'$ with the nodes being deleted, and
(6)~$\vecone \leq \vec{\varphi'}$ : transitive. 

The rule NI1 is applicable to the primitives which satisfy the constraint
$\statemap \Mm \Ng{\vec{z}} \geq \Ng{\statemap \Mm \vec{z}}$
for any $\vec{z} \in \bool_n$ where $n$ is the number of nodes of the input Kripke structure.
NI1 is sound because
(1)~$\statemap \Mm \vec{\varphi} \leq \vec{\varphi'}$ : given,
(2)~$\Ng{\statemap \Mm \vec{\varphi}} \geq \Ng{\vec{\varphi'}}$ : negation,
(3)~$\statemap \Mm \Ng{\vec{\varphi}} \geq \Ng{\statemap \Mm \vec{\varphi}}$ :
since $\statemap \Mm \Ng{\vec{z}} \geq \Ng{\statemap \Mm \vec{z}}$, and
(4)~$\statemap \Mm \Ng{\vec{\varphi}} \geq \Ng{\vec{\varphi'}}$ : transitive.
The rule NI2 is applicable to the primitives which satisfy the constraint that
$\statemap \Mm \Ng{\vec{z}} \leq \Ng{\statemap \Mm \vec{z}}$.
NI2 is sound because
(1)~$\statemap \Mm \vec{\varphi} = \vec{\varphi'}$ : given,
(2)~$\Ng{\statemap \Mm \vec{\varphi}} = \Ng{\vec{\varphi'}}$ : negation,
(3)~$\statemap \Mm \Ng{\vec{\varphi}} \leq \Ng{\statemap \Mm \vec{\varphi}}$ :
since $\statemap \Mm \Ng{\vec{z}} \leq \Ng{\statemap \Mm \vec{z}}$, and
(4)~$\statemap \Mm \Ng{\vec{\varphi}} \leq \Ng{\vec{\varphi'}}$ : transitive.

The relation $\statemap \Mm \Ng{\vec{z}} \geq \Ng{\statemap \Mm \vec{z}}$ means
that for any $\vec{z} \in \bool_n$, the set of nodes of $\mykripke'$ corresponding
to the complement of $\vec{z}$ is a superset of the set of nodes of $\mykripke'$
obtained by complement of the set of nodes (of $\mykripke'$) corresponding to $\vec{z}$.
Conversely, $\statemap \Mm \Ng{\vec{z}} \leq \Ng{\statemap \Mm \vec{z}}$ means that
the former set is a subset of the later set. We now establish the relations between
$\statemap \Mm \Ng{\vec{z}}$ and $\Ng{\statemap \Mm \vec{z}}$ for the primitive transformations.

\begin{lemma}
\label{lmm:neg-I}
If $\statemap$ is the correspondence matrix of an NS, EA, ED, or ND transformation from $\mykripke$ to $\mykripke'$
and $\vec{z}$ is a vector on $\mykripke$ then $\statemap \Mm \Ng{\vec{z}} \;=\; \Ng{\statemap \Mm \vec{z}}$.
\end{lemma}
\begin{proof}
Consider a boolean vector $\vec{z}$ defined over $\mykripke$.
\[
\setstretch{1.2}
\begin{array}{rlclr}
& \vec{z} \;\Pl\; \Ng{\vec{z}} & = & \vecone &
\mycomment{2}{From~(\ref{1}{eqn:inverse-pl})}{identity element}\pnl
& \statemap \Mm (\vec{z} \;\Pl\; \Ng{\vec{z}}) & = & \statemap \Mm \vecone &
\mycomment{2}{From~(\ref{1}{eqn:mono-3})}{monotonic}\pnl
(i) & \statemap \Mm \vec{z} \;\Pl\; \statemap \Mm \Ng{\vec{z}} 
\phantom{\;\Pl\; (\statemap \Mm \Ng{\vec{z}}) \;\St\; \statemap \Mm \Ng{\vec{z}})}
& = & \vecone &
\mycomment{2}{From~(\ref{1}{eqn:distr-5}); $\statemap$ is a total relation, hence $\statemap \Mm \vecone = \vecone$}
{distributive; $\statemap$ is a total relation, hence $\statemap \Mm \vecone = \vecone$}
\end{array}
\]

Consider the following derivation.
\[
\setstretch{1.2}
\begin{array}{clcl@{\hspace*{0.6cm}}r}
& \vec{z} \;\St\; \Ng{\vec{z}} & = & \veczero &
\mycomment{2}{From~(\ref{1}{eqn:inverse-st})}{identity element}\pnl
& \statemap \Mm (\vec{z} \;\St\; \Ng{\vec{z}}) 
\phantom{\;\Pl\; (\statemap \Mm \Ng{\vec{z}}) \;\St\; \statemap \Mm \Ng{\vec{z}})}
& = & \statemap \Mm \veczero &
\mycomment{2}{From~(\ref{1}{eqn:mono-3})}{monotonic}\pnl
(ii) & \statemap \Mm \vec{z} \;\St\; \statemap \Mm \Ng{\vec{z}} & = & \veczero &
\mycomment{2}{Since $\statemap$ is a function, from~(\ref{1}{eqn:distr-7})}{since the relation $\statemap$ is a function, it distributes}
\end{array}
\]

We first show that $\statemap \Mm \Ng{\vec{z}} \leq \Ng{\statemap \Mm \vec{z}}$.
\[
\setstretch{1.4}
\begin{array}{clcl@{\hspace*{0.8cm}}r}
& \vec{1} & = & \Ng{\vec{0}} & \pnl
& \statemap \Mm \vec{z} \;\Pl\; \statemap \Mm \Ng{\vec{z}} & = &
\Ng{\statemap \Mm \vec{z} \;\St\; \statemap \Mm \Ng{\vec{z}}} & 
\mycomment{2}{From (\ref{1}{eqn:pl}) and (\ref{1}{eqn:st})}{from ($i$) and ($ii$) above}\pnl
(iii) & (\statemap \Mm \vec{z} \;\Pl\; \statemap \Mm \Ng{\vec{z}}) \;\St\; \statemap \Mm \Ng{\vec{z}} & = &
(\Ng{\statemap \Mm \vec{z}} \;\Pl\; \Ng{\statemap \Mm \Ng{\vec{z}}})
\;\St\; \statemap \Mm \Ng{\vec{z}}  &
\mycomment{2}{From~(\ref{1}{eqn:de-morgan-1}) and (\ref{1}{eqn:mono-1})}{monotonic}\pnl
& (\statemap \Mm \vec{z} \St \statemap \Mm \Ng{\vec{z}})
\;\Pl\; (\statemap \Mm \Ng{\vec{z}} \St \statemap \Mm \Ng{\vec{z}}) & = &
(\Ng{\statemap \Mm \vec{z}} \St \statemap \Mm \Ng{\vec{z}})
\;\Pl\; (\Ng{\statemap \Mm \Ng{\vec{z}}} \St \statemap \Mm \Ng{\vec{z}})  &
\mycomment{2}{From~(\ref{1}{eqn:distr-4})}{distributive}\pnl
& \vec{0} \;\Pl\; \statemap \Mm \Ng{\vec{z}} & = &
(\Ng{\statemap \Mm \vec{z}} \St \statemap \Mm \Ng{\vec{z}}) \;\Pl\; \vec{0}  &
\mycomment{2}{From~(\ref{1}{eqn:st}), (\ref{1}{eqn:idem-1}), and (\ref{1}{eqn:inverse-st})}
{from ($ii$), identity element}\pnl
& \statemap \Mm \Ng{\vec{z}} & \leq & \Ng{\statemap \Mm \vec{z}} & \mycomment{2}{From~(\ref{1}{eqn:misc-2})}{}
\end{array}
\]

Instead of $\statemap \Mm \Ng{\vec{z}}$, if we take a product by $\Ng{\statemap \Mm \vec{z}}$ in ($iii$),
we get $\statemap \Mm \Ng{\vec{z}} \geq\; \Ng{\statemap \Mm \vec{z}}$.
\hfill
\end{proof}

\begin{lemma}
\label{lmm:neg-II}
If $\statemap$ is the correspondence matrix of an NM transformation from $\mykripke$ to $\mykripke'$
and $\vec{z}$ is a vector on $\mykripke$ then $\statemap \Mm \Ng{\vec{z}} \;\geq\; \Ng{\statemap \Mm \vec{z}}$.
\end{lemma}

The proof of Lemma~\ref{lmm:neg-II} is similar to that of Lemma~\ref{lmm:neg-I}.
From Lemma~\ref{lmm:neg-I} and Lemma~\ref{lmm:neg-II}, we know that the rule NI1
is applicable to the primitives NS, EA, ED, ND, and NM.

\begin{lemma}
\label{lmm:neg-III}
If $\statemap$ is the correspondence matrix of an NA transformation from $\mykripke$ to $\mykripke'$,
$\vec{z}$ is a vector on $\mykripke$, and $\newstates$ the atomic proposition denoting the new nodes in $\mykripke'$ then 
$\Ng{\statemap \Mm \vec{z}} \;\St\; \Ng{\vec{\newstates}} = \statemap \Mm \Ng{\vec{z}}$.
\end{lemma}
\begin{proof}
Consider a boolean vector $\vec{z}$ defined over $\mykripke$.
\[
\setstretch{1.2}
\begin{array}{clclr}
& \vec{z} \;\Pl\; \Ng{\vec{z}} & = & \vecone &
\mycomment{2}{From~(\ref{1}{eqn:inverse-pl})}{}\pnl
& \statemap \Mm (\vec{z} \;\Pl\; \Ng{\vec{z}}) & = & \statemap \Mm \vecone &
\mycomment{2}{From~(\ref{1}{eqn:mono-3})}{monotonic}\pnl
& \statemap \Mm \vec{z} \;\Pl\; \statemap \Mm \Ng{\vec{z}} 
\phantom{\;\Pl\; (\statemap \Mm \Ng{\vec{z}}) \;\St\; \statemap \Mm \Ng{\vec{z}})}
& = & \statemap \Mm \vecone \leq \vecone &
\mycomment{2}{From~(\ref{1}{eqn:distr-5}); $\statemap$ is a partial function, hence $\statemap \Mm \vecone \leq \vecone$}
{distribute; $\statemap$ being a partial function, $\statemap \Mm \vecone \leq \vecone$}
\end{array}
\]

Consider the following derivation.
\[
\setstretch{1.2}
\begin{array}{rlcl@{\hspace*{3.8cm}}r}
& \vec{z} \;\St\; \Ng{\vec{z}} & = & \veczero &
\mycomment{2}{From~(\ref{1}{eqn:inverse-st})}{}\pnl
& \statemap \Mm (\vec{z} \;\St\; \Ng{\vec{z}}) & = & \statemap \Mm \veczero &
\mycomment{2}{From~(\ref{1}{eqn:mono-3})}{monotonic}\pnl
(i) & \statemap \Mm \vec{z} \;\St\; \statemap \Mm \Ng{\vec{z}} 
\phantom{\;\Pl\; (\statemap \Mm \Ng{\vec{z}}) \;\St\; \statemap \Mm \Ng{\vec{z}})}
& = & \veczero &
\mycomment{2}{Since $\statemap$ is a function, from~(\ref{1}{eqn:distr-7})}{distribute; $\statemap$ is a function}
\end{array}
\]

Since $\vec{\newstates} = \vecone \Mn \statemap \Mm \vecone$, 
we have $\statemap \Mm \vec{z} \Pl \statemap \Mm \Ng{\vec{z}} \Pl \vec{\newstates} = \vecone$.

We first show that 
$\statemap \Mm \Ng{\vec{z}} \leq \Ng{\statemap \Mm \vec{z}} \;\St\; \Ng{\vec{\newstates}}$.
\[
\setstretch{1.4}
\begin{array}{rlcl@{\hspace*{1.5cm}}r}
& \vec{1} & = & \Ng{\vec{0}} & \pnl
& \statemap \Mm \vec{z} \Pl \statemap \Mm \Ng{\vec{z}} \Pl \vec{\newstates} & = &
\Ng{\statemap \Mm \vec{z} \St \statemap \Mm \Ng{\vec{z}}} &
\mycomment{1}{substitute}{}\pnl
& (\statemap \Mm \vec{z} \Pl \statemap \Mm \Ng{\vec{z}} \Pl \vec{\newstates}) 
\St \statemap \Mm \Ng{\vec{z}} & = &
(\Ng{\statemap \Mm \vec{z}} \Pl \Ng{\statemap \Mm \Ng{\vec{z}}}) \St \statemap \Mm \Ng{\vec{z}}  &
\mycomment{2}{From~(\ref{1}{eqn:de-morgan-1}) and (\ref{1}{eqn:mono-1})}{distributive, monotonic}\pnl
(ii) & \vec{0} \Pl \statemap \Mm \Ng{\vec{z}} \Pl (\vec{\newstates} \St \statemap \Mm \Ng{\vec{z}}) & = &
(\Ng{\statemap \Mm \vec{z}} \St \statemap \Mm \Ng{\vec{z}})
\Pl \vec{0}  &
\mycomment{2}{From~(\ref{1}{eqn:distr-4})}{distributive, substitute ($i$)}\pnl
& \statemap \Mm \Ng{\vec{z}} & \leq &
\Ng{\statemap \Mm \vec{z}} \St \statemap \Mm \Ng{\vec{z}} &
\mycomment{1}{}{first disjunct of LHS of ($ii$)}\pnl
(iii) & \statemap \Mm \Ng{\vec{z}} & \leq &
\Ng{\statemap \Mm \vec{z}} &
\mycomment{1}{}{From~(\ref{1}{eqn:misc-2})}
\pnl
& \vec{\newstates} \St \statemap \Mm \Ng{\vec{z}} & \leq &
\Ng{\statemap \Mm \vec{z}} \;\St\; \statemap \Mm \Ng{\vec{z}} &
\mycomment{2}{}{third disjunct of LHS ($ii$)}\pnl
(iv) & \vec{\newstates} & \leq &
\Ng{\statemap \Mm \vec{z}} &
\mycomment{2}{From~(\ref{1}{eqn:mono-14})}{}\pnl
& \statemap \Mm \Ng{\vec{z}} \Pl \vec{\newstates} & \leq &
\Ng{\statemap \Mm \vec{z}} &
\mycomment{2}{From~(\ref{1}{eqn:mono-6})}{($iii$), ($iv$), monotonic, distributive}\pnl
& \statemap \Mm \Ng{\vec{z}} & \leq &
\Ng{\statemap \Mm \vec{z}} \;\St\; \Ng{\vec{\newstates}} &
\mycomment{2}{From~(\ref{1}{eqn:misc-1}) and (\ref{1}{eqn:de-morgan-2})}
{negation; $\statemap \Mm \Ng{\vec{z}}$ and $\vec{\newstates}$ are disjoint}
\end{array}
\]
Instead of $\statemap \Mm \Ng{\vec{z}}$, if we take a product 
by $\Ng{\statemap \Mm \vec{z}} \;\St\; \Ng{\vec{\newstates}}$ in the third step, we get 
$\statemap \Mm \Ng{\vec{z}} \geq \Ng{\statemap \Mm \vec{z}} \;\St\; \Ng{\vec{\newstates}}$.
\hfill
\end{proof}

From Lemma~\ref{lmm:neg-III}, for the primitive NA, $\statemap \Mm \Ng{\vec{z}} \leq \Ng{\statemap \Mm \vec{z}}$.
Thus, additionally from Lemma~\ref{lmm:neg-I}, we know that the rule NI2
is applicable to the primitives NS, EA, ED, ND, and NA.

The rule IRT is applicable to the primitives which satisfy the constraint
$\vec{z} = \Tr{\statemap} \Mm \statemap \Mm \vec{z}$ for any $\vec{z} \in \bool_n$
where $n$ is the number of nodes of the input Kripke structure.
IRT is sound because
(1)~$\vec{\varphi'} \leq \vec{\psi'}$ : given,
(2)~$\Tr{\statemap} \Mm \vec{\varphi'} \leq \Tr{\statemap} \Mm \vec{\psi'}$ : monotonic,
(3)~$\statemap \Mm \vec{\varphi} = \vec{\varphi'}$ : given,
(4)~$\Tr{\statemap} \Mm \statemap \Mm \vec{\varphi} = \Tr{\statemap} \Mm \vec{\varphi'}$ : monotonic,
(5)~$\vec{\varphi} = \Tr{\statemap} \Mm \vec{\varphi'}$ : since $\vec{z} = \Tr{\statemap} \Mm \statemap \Mm \vec{z}$,
(6)~$\statemap \Mm \vec{\psi} = \vec{\psi'}$ : given,
(7)~$\Tr{\statemap} \Mm \statemap \Mm \vec{\psi} = \Tr{\statemap} \Mm \vec{\psi'}$ : monotonic,
(8)~$\vec{\psi} = \Tr{\statemap} \Mm \vec{\psi'}$ : since $\vec{z} = \Tr{\statemap} \Mm \statemap \Mm \vec{z}$, and
(9)~$\vec{\varphi} \leq \vec{\psi}$ : substitute from (5) and (8) in (2).

From the nature of the correspondence relations of the primitives,
we have the following result.
\begin{lemma}
\label{lmm:ns-mapping}
\label{lmm:nm-mapping}
\label{lmm:ea-mapping}
\label{lmm:ed-reverse-mapping}
\label{lmm:ed-mapping}
\label{lmm:na-mapping}
\label{lmm:nd-reverse-mapping}
\label{lmm:im-mapping}
\label{lmm:im-reverse-mapping}
If $\statemap$ is the correspondence matrix of a primitive transformation of $\mykripke$ to $\mykripke'$
and $Z$ and $Z'$ are boolean matrices of appropriate sizes then the following properties hold:
\begin{enumerate}
\item For an NS, EA, ED, NA, or IM transformation,
$Z = \Tr{\statemap} \Mm \statemap \Mm Z$ and $Z = Z \Mm \Tr{\statemap} \Mm \statemap$.
\item For an NM transformation,
$Z \leq \Tr{\statemap} \Mm \statemap \Mm Z$ and $Z \leq Z \Mm \Tr{\statemap} \Mm \statemap$.
\item For an ED, ND, or IM transformation,
$Z' = \statemap \Mm \Tr{\statemap} \Mm Z'$ and $Z' = Z' \Mm \statemap \Mm \Tr{\statemap}$.
\item For an NS transformation,
$Z' \leq \statemap \Mm \Tr{\statemap} \Mm Z'$ and $Z' \leq Z' \Mm \statemap \Mm \Tr{\statemap}$.
\end{enumerate}
\end{lemma}






From Lemma~\ref{lmm:ns-mapping}, it is clear that $\vec{z} = \Tr{\statemap} \Mm \statemap \Mm \vec{z}$
holds for NS, EA, ED, NA, and IM transformations and thus the rule IRT is applicable to them.

\subsection{Transformation specific inference rules}

\newcommand{\nstimpU}{
 \begin{tabular}{c}
 $\varphi \!\timp\! \varphi'$\\
 \cline{1-1}\jsep
 $\Delta(\varphi) \!\timp\! \Delta(\varphi')$
 \end{tabular}
}

\newcommand{\nstimpB}{
 \begin{tabular}{@{}c@{}c@{}}
 \begin{tabular}{c}
 \begin{tabular}{cc}
 $\varphi \!\timp\! \varphi'$ & $\psi \!\timp\! \psi'$\\
 \end{tabular}\\
 \cline{1-1}\jsep
 $\nabla(\varphi,\psi) \!\timp\! \nabla(\varphi',\psi')$
 \end{tabular}
 &
 \;\;$(\text{TTL}_{ns},\text{TTL}_{im})$\\
 \end{tabular}
}

\newcommand{\nmtimpU}{
 \begin{tabular}{c}
 $\varphi \!\timp\! \varphi'$\\
 \cline{1-1}\jsep
 $\Delta(\varphi) \!\timp\! \Delta(\varphi')$ 
 \end{tabular}
}

\newcommand{\nmtimpB}{
 \begin{tabular}{@{}c@{}c@{}}
 \begin{tabular}{c}
 \begin{tabular}{cc}
 $\varphi \!\timp\! \varphi'$ & $\psi \!\timp\! \psi'$\\
 \end{tabular}\\
 \cline{1-1}\jsep
 $\nabla(\varphi,\psi) \!\timp\! \nabla(\varphi',\psi')$
 \end{tabular}
 &
 \;\;$(\text{TTL}_{nm},\text{TTL}_{ea})$\\
 \end{tabular}
}

\newcommand{\eatimpU}{
 \begin{tabular}{c}
 $\varphi \!\timp\! \varphi'$\\
 \cline{1-1}\jsep
 $\Delta(\varphi) \!\timp\! \Delta(\varphi')$
 \end{tabular}
}

\newcommand{\eatimpB}{
 \begin{tabular}{@{}c@{}c@{}}
 \begin{tabular}{c}
 \begin{tabular}{cc}
 $\varphi \!\timp\! \varphi'$ & $\psi \!\timp\! \psi'$\\
 \end{tabular}\\
 \cline{1-1}\jsep
 $\nabla(\varphi,\psi) \!\timp\! \nabla(\varphi',\psi')$
 \end{tabular}
 &
 \;\;$(\text{TTL}_{ea})$\\
 \end{tabular}
}

\newcommand{\edtimpU}{
 \begin{tabular}{c}
 $\varphi \!\timp\! \varphi'$\\
 \cline{1-1}\jsep
 $\Delta(\varphi) \!\timp\! \Delta(\varphi')$
 \end{tabular}
}

\newcommand{\edtimpB}{
 \begin{tabular}{@{}c@{}c@{}}
 \begin{tabular}{c}
 \begin{tabular}{cc}
 $\varphi \!\timp\! \varphi'$ & $\psi \!\timp\! \psi'$
 \end{tabular}\\
 \cline{1-1}\jsep
 $\nabla(\varphi,\psi) \!\timp\! \nabla(\varphi',\psi')$
 \end{tabular}
 &
 \;\;$(\text{TTL}_{ed})$\\
 \end{tabular}
}

\newcommand{\nasimple}{
\begin{tabular}{@{}c@{}}
$\varphi \!\timp\! \varphi'$\\
\cline{1-1}\jsep
$\varphi \!\timp\! (\varphi' \wedge \neg\;\newstates)$
\end{tabular}
}

\newcommand{\natimpU}{
 \begin{tabular}{c}
 \begin{tabular}{@{}cc@{}}
 $\varphi \!\timp\! \varphi'$ 
 & 
 $\newstates \!\Imp\! \varphi'$
 \end{tabular}\\
 \cline{1-1}\jsep
 $\Delta(\varphi) \!\timp\! \Delta(\varphi')$
 \end{tabular}
}

\newcommand{\natimpB}{
 \begin{tabular}{@{}c@{}c@{}}
 \begin{tabular}{c}
 \begin{tabular}{@{}ccc@{}}
 $\varphi \!\timp\! \varphi'$ 
 & 
 $\newstates \!\Imp\! \varphi'$
 &
 $\psi \!\timp\! \psi'$\\
 \end{tabular}\\
 \cline{1-1}\jsep
 $\nabla(\varphi,\psi) \!\timp\! \nabla(\varphi',\psi')$
 \end{tabular}
 &
 \;\;$(\text{TTL}_{na})$\\
 \end{tabular}
}

\newcommand{\ALTnatimpU}{
 \begin{tabular}{c}
 $\varphi \!\timp\! \varphi'$\\
 \cline{1-1}\jsep
 $\Delta(\varphi) \!\timp\! \Delta(\varphi' \vee \newstates)$
 \end{tabular}
}

\newcommand{\ALTnatimpB}{
 \begin{tabular}{@{}c@{}c@{}}
 \begin{tabular}{c}
 \begin{tabular}{cc}
 $\varphi \!\timp\! \varphi'$ 
 &
 $\psi \!\timp\! \psi'$\\
 \end{tabular}\\
 \cline{1-1}\jsep
 $\nabla(\varphi,\psi) \!\timp\! \nabla(\varphi' \vee \newstates,\psi')$
 \end{tabular}
 &
 \;\;$(\text{Alternative TTL}_{na})$\\
 \end{tabular}
}

\newcommand{\ndsimple}{
\begin{tabular}{c}
$\varphi \!\timp\! \varphi'$\\
\cline{1-1}\jsep
$(\varphi \vee \delstates) \!\timp\! \varphi'$
\end{tabular}
}

\newcommand{\ndtimpU}{
 \begin{tabular}{@{}c@{}}
 \begin{tabular}{cc}
  $\psi \!\Imp\! \neg\delstates$ & $\psi \!\timp\! \psi'$
 \end{tabular}\\
 \cline{1-1}\jsep
 $\Delta(\psi) \!\timp\! \Delta(\psi')$
 \end{tabular}
}

\newcommand{\ndtimpEG}{
 \begin{tabular}{@{}c@{}}
 $\varphi \!\timp\! \varphi'$\\
 \cline{1-1}\jsep
 $\Theta(\varphi) \!\timp\! \Theta(\varphi')$
 \end{tabular}
}

\newcommand{\ndtimpB}{
 \begin{tabular}{c}
 \begin{tabular}{@{}ccc@{}}
 $\varphi \!\timp\! \varphi'$ 
 & 
 $\psi \!\Imp\! \neg\delstates$
 &
 $\psi \!\timp\! \psi'$\\
 \end{tabular}\\
 \cline{1-1}\jsep
 $\nabla(\varphi,\psi) \!\timp\! \nabla(\varphi',\psi')$
 \end{tabular}
}

\begin{figure}
\centering
\fbox{
\begin{minipage}{\minipagewidth}
\centering
\begin{tabular}{c}
\begin{tabular}{cc}
\nstimpU & \nstimpB
\end{tabular}\\[5mm]
where $\Delta$ and $\nabla$ are respectively any unary and binary $\ctlbp$ temporal operators\\[1mm]
{\bf (a) Inference rules for node splitting and isomorphic transformations}\\[3mm]
\begin{tabular}{cc}
\nmtimpU & \nmtimpB
\end{tabular}\\[5mm]
where $\Delta$ and $\nabla$ are respectively any \emph{existential} unary and binary $\ctlbp$ temporal operators\\[1mm]
{\bf (b) Inference rules for node merging and edge addition transformations}\\[3mm]
\begin{tabular}{cc}
\edtimpU & \edtimpB
\end{tabular}\\[5mm]
where $\Delta$ and $\nabla$ are respectively any \emph{universal future} unary and binary $\ctlbp$ operators\\[1mm]
{\bf (c) Inference rules for edge deletion transformations}\\[3mm]
\begin{tabular}{ccc}
\nasimple & \natimpU & \natimpB
\end{tabular}\\[5mm]
where $\Delta$ and $\nabla$ are respectively any unary and binary $\ctlbp$ operators\\[1mm]
{\bf (d) Inference rules for node addition transformations}\\[3mm]
\begin{tabular}{cc}
\begin{tabular}{cc}
{\ndsimple} & \ndtimpU\\[5mm]
\ndtimpEG & \ndtimpB
\end{tabular} &
$(\text{TTL}_{nd})$
\end{tabular}\\[11mm]
where $\Delta \equiv \EX, \AX, \EY, \AY$;
$\Theta \equiv \EG, \AG, \EH, \AH$;
and $\nabla \equiv \EU, \AU, \EW, \AW, \ES, \AS$\\[1mm]
{\bf (e) Inference rules for node deletion transformations}
\end{tabular}
\end{minipage}}
\caption{Transformation specific inference rules}
\label{fig:inf-rules}
\end{figure}

The transformation specific rules define sound inferences for introduction of 
temporal operators in the TTL formulae. The rules we present in Figure~\ref{fig:inf-rules}
allow introduction of the same temporal operators on both sides of a TTL implication.
The rules are different for different types of primitives transformations.
We differ the proofs of soundness of the rules until Section~\ref{sec:proofs}.
Here, we discuss them only informally. The boolean matrix algebraic semantics
of $\ctlbp$ operators is given in Appendix~\ref{sec:ctlbp}.

An NM or an EA transformation inserts new edges in the transformed Kripke structure
but preserves all the existing edges. Thus, we can correlate any existential temporal operator
between the input and the transformed Kripke structures.

An ED transformation, on the other hand, deletes edges from the input Kripke structure
but does not add any edges. Thus, we can correlate any universal temporal operator
between the input and the transformed Kripke structures.
However, $\text{TTL}_{ed}$ rules are applicable only for future operators. 
In $\ctlbp$, past operator $\AY$ (Definition~\ref{def:ctlbp-past-semantics}) is interpreted in a strong sense i.e. for an $\AY$
formula to hold at a node, it must have at least one predecessor. 
At an entry node i.e. a node without any predecessors, no $\AY$ formula holds.
For edge deletion, we cannot guarantee existence of a predecessor for every node in the transformed Kripke structure. 
The existence of successors however is guaranteed
because for a Kripke structure to be well-formed, every node has to have an outgoing edge.
Hence, we can correlate universal future $\ctlbp$ operators.

In Figure~\ref{fig:inf-rules}~(d), $\newstates$ is an atomic proposition of a transformed Kripke
structure that denotes the newly added nodes by an NA transformation. 
In Figure~\ref{fig:inf-rules}~(e), $\delstates$ is an atomic proposition of an input Kripke structure
that denotes the nodes being deleted by an ND transformation.

\subsection{Verification of compiler optimizations revisited}
\label{sec:cse-proof}

\newcommand{\seqOne}{$\vdash_1$}
\newcommand{\seqTwoOne}{$\vdash_{2.1}$}
\newcommand{\seqTwoTwo}{$\vdash_{2.2}$}
\newcommand{\seqTwoThree}{$\vdash_{2.3}$}
\newcommand{\seqTwoFour}{$\vdash_{2.4}$}

\renewcommand{\seqOne}{\ensuremath{\cdots \vdash_1 \pvsid{SoundRE(prog4, redund4, e, t)}}}
\renewcommand{\seqTwoOne}{\ensuremath{\cdots \vdash_{2.1} \pvsid{ASSIGNs?(prog4`L,redund4)}}}
\renewcommand{\seqTwoTwo}{\ensuremath{\cdots \vdash_{2.2} \neg \pvsid{member(t, VOperands(e))}}}
\renewcommand{\seqTwoThree}{\ensuremath{\cdots \vdash_{2.3} \pvsid{redund4} \leq \pvsid{Antloc(prog4, e)}}}
\renewcommand{\seqTwoFour}{\ensuremath{\cdots \vdash_{2.4} \pvsid{redund4} \leq \pvsid{EqValue\_IN(prog4, t, e)}}}

In Section~\ref{sec:motivation}, we discussed a verification scheme for specifications of compiler optimizations
that motivated the development of TTL.
We now give a sketch of the proof of a verification condition for 
the CSE specification (Figure~\ref{fig:cse-spec}) using TTL. 
We have used TTL to prove soundness of several optimizations viz. common subexpression elimination,
optimal code placement, loop invariant code motion, lazy code motion, and
full and partial dead code elimination, in the PVS theorem prover. 

Consider the last transformation \pvsid{RE(prog4, redund4, e, t)}.
Common subexpression elimination is performed for non trivial expressions i.e. expressions 
containing some operator. Thus we need to establish part (b) of the constraint (4) in
the soundness condition for the primitive \pvsid{RE}. Formally, we need to show that
\begin{equation}
\label{eqn:my-1}
\pvsid{redund4} \Imp \pvsid{AY}(\pvsid{prog4`cfg}, 
\pvsid{AS}(\pvsid{prog4`cfg}, \pvsid{Transp(prog4,e)} * \neg\pvsid{Def(prog4,t),AssignStmt(prog4,t,e)})) \equiv \varphi_4
\end{equation}

From the definition of \pvsid{Avail}, we know that
\[
\pvsid{redund} \Imp \pvsid{AY(prog1`cfg, AS(prog1`cfg, Transp(prog1,e), orgavails))}
\]
We can prove that $\pvsid{orgavails} \Imp \pvsid{Transp(prog1,e)}$.
Let \pvsid{redund2} be the set of program points in \pvsid{prog2} that correspond to \pvsid{redund}.
Since, \pvsid{IP} is a node addition transformation, using $\text{TTL}_{na}$; and the correlations
of local data flow properties under an \pvsid{IP} transformation (Section~\ref{sec:ip-as-na}),
\[
\pvsid{redund2} \Imp \pvsid{AY(prog2`cfg, AS(prog2`cfg, Transp(prog2,e), newpoints))}
\]
The variable \pvsid{t} is declared as a new variable for \pvsid{prog2}.
Thus, $\neg \pvsid{Def(prog2,t)}$ is an invariant property. Hence,
\[
\pvsid{redund2} \Imp \pvsid{AY}(\pvsid{prog2`cfg}, \pvsid{AS}(\pvsid{prog2`cfg}, \pvsid{Transp(prog2,e)} * \neg 
\pvsid{Def(prog2,t), newpoints}))
\]
Since, \pvsid{IA} is an isomorphic transformation, using $\text{TTL}_{im}$; and the correlations
of local data flow properties under an \pvsid{IA} transformation (Section~\ref{sec:ia-as-im}),
\begin{equation}
\label{eqn:my-2}
\pvsid{redund2} \Imp \pvsid{AY}(\pvsid{prog3`cfg}, \pvsid{AS}(\pvsid{prog3`cfg}, \pvsid{Transp(prog3,e)} * \neg 
\pvsid{Def(prog3,t), AssignStmt(prog3,t,e)})) \equiv \varphi_3
\end{equation}

Now, consider the following derivation:
\[
\setstretch{1.2}
\begin{array}{lclr}
\pvsid{Def(prog3,t)} & \Rightarrow & \pvsid{Def(prog4,t)} & \text{Section~\ref{sec:re-as-im}}\\
\neg \pvsid{Def(prog3,t)} & \timp & \neg \pvsid{Def(prog4,t)} & \text{Rule NI2 (Figure~\ref{fig:ax})}\\
\pvsid{Transp(prog3,e)} & \Rightarrow & \pvsid{Transp(prog4,e)} & \text{Section~\ref{sec:re-as-im}}\\
\pvsid{Transp(prog3,e)} * \neg \pvsid{Def(prog3,t)} & \timp & 
\pvsid{Transp(prog4,e)} * \neg \pvsid{Def(prog4,t)} & \text{Rules CI1, CI2 (Figure~\ref{fig:ax})}
\end{array}
\]

Recall that \pvsid{t} is a new variable with respect to \pvsid{prog2}.
Hence, there cannot be any assignment to \pvsid{t} in \pvsid{prog2}.
The transformation of \pvsid{prog2} to \pvsid{prog3} involves insertion of \pvsid{ASSIGN(t,e)}
at \pvsid{newpoints} which are distinct from \pvsid{orgavails3}.
Hence, $\pvsid{AssignStmt(prog3,t,e)} \timp \pvsid{AssignStmt(prog4,t,e)}$.
Since \pvsid{RE} is an isomorphic transformation, using $\text{TTL}_{im}$ for the transformation
of \pvsid{prog3} to \pvsid{prog4}, we have $\varphi_3 \timp \varphi_4$.

Also $\pvsid{redund2} \Rightarrow \pvsid{redund4}$. We thus have the following simple derivation:
(1)~$\pvsid{redund2} \leftarrow \pvsid{redund4}$: by IR2,
(2)~$\pvsid{redund2} \Imp \varphi_3$: from Equation~\ref{eqn:my-2},
(3)~$\varphi_3 \timp \varphi_4$: proved above,
(4)~$\pvsid{redund2} \timp \varphi_4$: by FC1, and
(5)~$\pvsid{redund4} \Imp \varphi_4$: from steps (1) and (4), MC.
This establishes the verification condition (Equation~\ref{eqn:my-1})
of the last transformation in the CSE specification.

\section{Soundness of the logic}
\label{sec:proofs}

In Section~\ref{sec:sim-rel}, we define simulation, bisimulation, and a special case of weak bisimulation 
between Kripke structures. 
For each transformation primitive, we show in Section~\ref{sec:trans-sim-rel} 
what kind of simulation relation the transformation constructs
between the input and the transformed Kripke structures.
In Section~\ref{sec:corr}, we prove that certain correlations between $\ctlbp$ formulae hold 
between a pair of Kripke structures if a particular kind of simulation relation exists between them.
The soundness of the transformation specific inference rules follows immediately
from the results in Section~\ref{sec:trans-sim-rel} and Section~\ref{sec:corr}.

\subsection{Algebraic formulations of simulation relations}
\label{sec:sim-rel}

\subsubsection{Binary relations}

Consider sets $P$ and $Q$ containing $n$ and $m$ elements respectively.
We represent an ordered relation $R \subseteq Q \times P$ by an $(m \Tm n)$ boolean matrix.
A subset $X \subseteq P$ (or $X \subseteq Q$) is denoted by a boolean vector of size $n$ (or size $m$).
Let $\bool_n$ and $\bool_m$
denote the sets of all boolean vectors of sizes $n$ and $m$ respectively.

A relation $R \subseteq Q \times P$ can also be considered as a function $R \colon \bool_n \rightarrow \bool_m$.
For simplicity of presentation, we consider boolean vectors and column matrices interchangeably and
do not distinguish between $\bool_n$ and $\bool_{(n,1)}$.
With the interpretation of a boolean matrix as a function over boolean vectors,
the matrix multiplication $R \Mm \vec{z}$ where $\vec{z} \in \bool_n$ can be seen as a function application.
A matrix multiplication $R_1 \Mm R_2$ can be considered as function composition $(R_1 \circ R_2)$.
The transpose $\Tr{R}$ of matrix $R$ represents the inverse of relation $R$
and is the function $\Tr{R} \colon \bool_m \rightarrow \bool_n$.

We now formalize various properties of a relation $R$ in boolean matrix algebraic setting.
Let $\vec{z} \in \bool_n$ and $\vec{z'} \in \bool_m$.
We abbreviate the one-to-one, one-to-many, many-to-one, and many-to-many properties of relations by 1-1, 1-m, m-1, m-m respectively.
\renewcommand{\Mm}{\mbox{\ensuremath{\cdot}}}
\begin{equation}
\label{eqn:first-R}
\setstretch{1.3}
\begin{array}{lclr}
R \Mm \vecone & = & \vecone & \text{ (total)}\\
R \Mm \vecone & \leq & \vecone & \text{ (may not be total)}\\
\vec{z'} & = & R \Mm \Tr{R} \Mm \vec{z'} & \text{ (total but neither m-1 nor m-m)}\\
\vec{z'} & \leq & R \Mm \Tr{R} \Mm \vec{z'} & \text{ (total and possibly m-1 or m-m)}\\
R \Mm \Tr{R} \Mm \vec{z'} & \leq & \vec{z'} & \text{ (possibly not total and neither m-1 nor m-m)}\\
\Ng{R \Mm \vec{z}} & = & R \Mm \Ng{\vec{z}} & \text{ (total but neither 1-m nor m-m)}\\
\Ng{R \Mm \vec{z}} & \leq & R \Mm \Ng{\vec{z}} & \text{ (total and possibly 1-m or m-m)}\\
R \Mm \Ng{\vec{z}} & \leq & \Ng{R \Mm \vec{z}} &
\text{ \scalebox{1}{(possibly not total and neither 1-m nor m-m)}}
\end{array}
\end{equation}

The above list gives properties of relation $R$ in its first argument.
Similarly, we define properties of relation $R$ in its second argument as follows.
\begin{equation}
\label{eqn:second-R}
\setstretch{1.3}
\begin{array}{lclr}
\Tr{R} \Mm \vecone & = & \vecone & \text{ (onto)}\\
\Tr{R} \Mm \vecone & \leq & \vecone & \text{ (may not be onto)}\\
\vec{z} & = & \Tr{R} \Mm R \Mm \vec{z} & \text{ (onto but neither 1-m nor m-m)}\\
\vec{z} & \leq & \Tr{R} \Mm R \Mm \vec{z} & \text{ (onto and possibly 1-m or m-m)}\\
\Tr{R} \Mm R \Mm \vec{z} & \leq & \vec{z} & \text{ (possibly not onto and neither 1-m nor m-m) }\\
\Ng{\Tr{R} \Mm \vec{z'}} & = & \Tr{R} \Mm \Ng{\vec{z'}} & \text{ (onto but neither m-1 nor m-m)}\\
\Ng{\Tr{R} \Mm \vec{z'}} & \leq & \Tr{R} \Mm \Ng{\vec{z'}} & \text{ (onto and possibly m-1 or m-m)}\\
\Tr{R} \Mm \Ng{\vec{z'}} & \leq & \Ng{\Tr{R} \Mm \vec{z'}} &
\text{ \scalebox{1}{(possibly not onto and neither m-1 nor m-m)} }
\end{array}
\end{equation}

\renewcommand{\Mm}{\ensuremath{\cdot}}

\subsubsection{Simulation relations between Kripke structures}

Consider Krike structures $\mykripke = (\graph,\myprops,\mylabeling)$ and 
$\mykripke' = (\graph',\myprops',\mylabeling')$ with $\mat$ and $\mat'$ as
respective adjacency matrices. Let there be $n$ nodes in $\mykripke$ and $n'$ nodes in $\mykripke'$.
Let an $(n' \Tm n)$ boolean matrix $R$ be a relation between the nodes of $\mykripke'$ and $\mykripke$.
We use unprimed symbols for $\mykripke$ and primed symbols for $\mykripke'$.

\begin{definition}
\label{def:kripke-sim}
A relation $R$ is a \emph{simulation relation} between Kripke structures $\mykripke$ and $\mykripke'$
{(denoted as $\mykripke \;\Kfsim_R\; \mykripke'$)} if 
{(1)}~$R \Mm \mat \Mm \Tr{R} \;\leq\; \mat'$ and
{(2)}~for all $\vec{z} \in \bool_n$, $\vec{z} \leq \Tr{R} \Mm R \Mm \vec{z}$.
\end{definition}

Condition (1) states that if $p'$ and $p$ are related by R ($[R]_{p'}^{p} = \one$) and
there exists an edge $\pair{p}{q}$ in $\mykripke$ ($[\mat]_{p}^{q} = \one$), and 
some $q'$ is related to $q$ ($[R]_{q'}^{q} = [\Tr{R}]_{q}^{q'} = \one$) then
there is an edge $\pair{p'}{q'}$ ($[\mat']_{p'}^{q'} = \one$).
In other words, if there is an edge $\pair{p}{q}$ in $\mykripke$ and $p'$ and $q'$
are related to $p$ and $q$ through $R$, then there exists an edge $\pair{p'}{q'}$ in $\mykripke'$.

As discussed in Section~\ref{sec:kripke-transformations}, we model programs as Kripke structures.
We consider local data flow properties as atomic propositions. 
The labeling of nodes by atomic propositions is determined by valuations of the properties.
The standard formulations of simulation 
relations~\cite{clarke:model-checking} require related nodes to be labeled with same atomic propositions.
Our aim is to model program transformations as Kripke transformations.
Since program transformations can change statements,
we shall not insist on equality of atomic labels. 

From Equation~(\ref{eqn:second-R}), we know that condition (2) states that
$R$ should be an \emph{onto} relation and it can possibly be one-to-many or many-to-many.
Thus, $\mykripke \Kfsim_R \mykripke'$
if for \emph{every} edge in $\mykripke$, there is a corresponding edge in $\mykripke'$.

\begin{definition}
\label{def:kripke-bisim}
A relation $R$ is a \emph{bisimulation relation} between Kripke structures $\mykripke$ and $\mykripke'$
{(denoted as $\mykripke \;\Kfbsim_R\; \mykripke'$)} if $\mykripke \;\Kfsim_R\; \mykripke'$
and $\mykripke' \;\Kfsim_{\Tr{R}}\; \mykripke$.
\end{definition}

A relation $R$ is a bisimulation between Kripke structures $\mykripke$ and $\mykripke'$
if $R$ is a simulation between $\mykripke$ and $\mykripke'$ and
its inverse i.e. $\Tr{R}$ is a simulation between $\mykripke'$ and $\mykripke$.
Thus, $R$ should be a total and onto relation such that
for every edge in $\mykripke$ there is a corresponding edge in $\mykripke'$ and vice versa.

We now consider a special case of weak bisimulation called \emph{one-step weak bisimulation}.
The formulation is \emph{asymmetric} in a sense that only one Kripke structure has internal nodes. 
Let $\mykripke$ and $\mykripke'$ be two Kripke structures such that only $\mykripke'$ has internal nodes
which do not correspond to any nodes in $\mykripke$.
Unlike the standard formulations of weak bisimulation~\cite{55083,201032},
we require that no two internal nodes can have an edge between them.
In other words, internal and visible nodes can only alternate.
Let us denote the set of nodes in $\mykripke$ by $\mystates$ and 
the set of nodes in $\mykripke'$ by $\mystates'$.
Let us denote the set of internal nodes in $\mykripke'$ by $\newstates$.
The set of visible nodes is $\mystates' - \newstates$.

We characterize a one-step weak bisimulation by three relations 
between the nodes of $\mykripke$ and $\mykripke'$. 
An $(n' \Tm n)$ relation $R \subseteq (\mystates' - \newstates) \times \mystates$ 
which relates visible nodes of the two Kripke structures.
An $(n' \Tm n)$ relation $S \subseteq \newstates \times \mystates$ 
which relates internal nodes of $\mykripke'$ to visible nodes of $\mykripke$ and
an $(n \Tm n')$ relation $P \subseteq \mystates \times \newstates$ 
which relates visible nodes of $\mykripke$ to internal nodes of $\mykripke'$.
Let $n$ be the number of nodes in $\mykripke$ and $n'$ be the number of nodes in $\mykripke'$.

\begin{definition}
\renewcommand{\Mm}{\mbox{\ensuremath{\cdot}}}
\label{def:weak-bisim}
A triple $\langle R, P, S \rangle$ is a \emph{one-step weak bisimulation} between $\mykripke$ and $\mykripke'$
{(denoted as $\mykripke \;\Kwbisim_{\langle R,P,S \rangle}\; \mykripke'$)} if 
{(1)}~$R \Mm \mat \Mm \Tr{R} = (\mat' \Mn R \Mm P \Mn S \Mm \Tr{R}) \Pl R \Mm P \Mm S \Mm \Tr{R}$,
{(2)}~$R \Mm P \leq \mat'$,
{(3)}~$S \Mm \Tr{R} \leq \mat'$,
{(4)}~$S \Mm \vecone = \Tr{P} \Mm \vecone = \vec{\newstates} = \vecone \Mn R \Mm \vecone$,
{(5)}~$P \Mm R = \zero$,
{(6)}~$\Tr{R} \Mm S = \zero$,
{(7)}~for all $\vec{z} \in \bool_n$, $\vec{z} = \Tr{R} \Mm R \Mm \vec{z}$, 
{(8)}~for all $\vec{z'} \in \bool_{n'}$, $\vec{z'} \St \Ng{\vec{\newstates}} = R \Mm \Tr{R} \Mm (\vec{z'} \St \Ng{\vec{\newstates}})$, 
{(9)}~for all $\vec{z'} \in \bool_{n'}$, $\Tr{R} \Mm (\vec{z'} \St \Ng{\newstates}) = \Tr{R} \Mm \vec{z'}$, and
{(10)}~$\vec{\newstates} \St (\mat' \Mm \vec{\newstates} \Pl \Tr{\mat'} \Mm \vec{\newstates}) = \veczero$.
\end{definition}

Condition (1) states that if $p'$ and $p$ are related by $R$ ($[R]_{p'}^{p} = \one$) and 
there exists an edge $\pair{p}{q}$ in $\mykripke$ ($[\mat]_{p}^{q} = \one$), and
some $q'$ is related to $q$ ($[R]_{q}^{q'} = [\Tr{R}]_{q'}^{q} = \one$) then one of the following holds:
\begin{itemize}
\item[(a)] There is 
an edge $\pair{p'}{q'}$ ($[\mat']_{p'}^{q'} = \one$) and there does not exist any node $r$ in $\mykripke$
such that (i)~$p'$ is related to $r$ by $R$ ($[R]_{p'}^{r} = \one$) and 
$r$ is related to $q'$ by $P$ ($[P]_{r}^{q'} = \one$) or
(ii)~$p'$ is related to $r$ by $S$ ($[S]_{p'}^{r} = \one$) and
$r$ is related to $q'$ by $\Tr{R}$ ($[\Tr{R}]_{r}^{q'} = \one$).
\item[(b)] There exists a node $r'$ in $\mykripke'$ such that $p$ is related to $r'$ through $P$ ($[P]_{p}^{r'} = \one$) and
$r'$ is related to $q$ through $S$ ($[S]_{r'}^{q} = \one$).
In other words, the condition states that the edge $\pair{p}{q}$ is split into two edges $\pair{p'}{r'}$
and $\pair{r'}{q'}$ (where $r'$ is an internal node).
\end{itemize}

Condition (2) states that if $p'$ and $p$ are related by $R$ ($[R]_{p'}^{p} = \one$)
and $p$ is related to $r'$ through $P$ ($[P]_{p}^{r'} = \one$) then 
there is an edge $\pair{p'}{r'}$ ($[\mat']_{p'}^{r'} = \one$). 
Similarly, condition (3) states that if $r'$ and $q$ are related by $S$ and
$q'$ and $q$ are related by $R$ then there is an edge $\pair{r'}{q'}$.

Condition (4) states that the domain of $S$ and the range of $P$ is the set of internal nodes $\newstates$ in $\mykripke'$.
Further, the set of internal nodes is disjoint with the domain of $R$.

Condition (5) states that if $p$ is related to $p'$ by $P$ ($[P]_{p}^{p'} = \one$) 
then there does not exist any $q$ related to $p'$ by $R$ ($[R]_{p'}^{r} = \zero$, for all $r \in \mystates$).
Similarly, condition (6) states that if $p'$ is related to $q$ by $S$ ($[S]_{p'}^{q} = \one$)
then there does not exist any $p$ such that $p$ is related to $p'$ by $R$ 
($[R]_{p'}^{p} = [\Tr{R}]_{p}^{p'} = \zero$, for all $p \in \mystates$).

From Equations~(\ref{eqn:second-R}), we know that
condition (7) states that $R$ is an onto relation but neither one-to-many nor many-to-many.
From Equation~(\ref{eqn:first-R}), we know that
condition (8) states that $R$ is a total relation on $\mystates' - \newstates$ but
is neither many-to-one nor many-to-many. Further, condition (9) states that 
$R$ does not relate any node in $\newstates$ to a node node in $\mystates$.
Together, conditions (7), (8), and (9) state that 
$R \subseteq (\mystates' - \newstates) \times \mystates$ is a bijection.

Condition (10) states that no two internal nodes can be adjacent.

Thus, $\mykripke \Kwbisim_{\langle R,P,S \rangle} \mykripke'$ if for every edge $\pair{p}{q}$ in $\mykripke$ there is a
corresponding edge $\pair{p'}{q'}$ in $\mykripke'$ or there are two adjacent edges
$\pair{p'}{r'}$ and $\pair{r'}{q'}$ such that $p'$ and $p$ as well as $q'$ and $q$ are related by $R$.
Further, $p$ and $r'$ are related by $P$ and $r'$ and $q$ are related by $S$
where $r'$ is an internal node in $\mykripke'$.

\subsection{Primitive Kripke transformations and simulation relations}
\label{sec:trans-sim-rel}

\begin{figure}
\centering
\begin{tabular}{|l|c|c|c|c|c|c|c|}
\hline
Relations $\setminus$ Trans. & NS & NM & EA & ED & NA & ND & IM\\\hline\hline
$\mykripke \;\Kfsim_{\statemap}\; \mykripke'$ & Yes & Yes & Yes & No & --- & --- & Yes \\\hline
$\mykripke' \;\Kfsim_{\Tr{\statemap}}\; \mykripke$ & Yes & No & No & Yes & --- & --- & Yes \\\hline
$\mykripke \;\Kfbsim_{\statemap}\; \mykripke'$ & Yes & No & No & No & --- & --- & Yes \\\hline
$\mykripke \;\Kwbisim_{\langle \statemap, N_P, N_S \rangle}\; \mykripke'$ & --- & --- & --- & --- & Yes & No & ---\\\hline
$\mykripke' \;\Kwbisim_{\langle \Tr{\statemap}, N_P, N_S \rangle}\; \mykripke$ & --- & --- & --- & --- & No & Yes & ---\\\hline
\end{tabular}
\caption{Simulation relations between Kripke structures under primitive transformations}
\label{fig:sim-rel-k-trans}
\end{figure}

A primitive Kripke transformation constructs some kind of simulation
relation between a transformed Kripke structure $\mykripke'$ and the corresponding input Kripke structure $\mykripke$.
We summarize these relations in Figure~\ref{fig:sim-rel-k-trans} and
use them to correlate temporal properties across transformations in Section~\ref{sec:corr}.

The relation $\statemap$ denotes the correspondence matrix of the primitive transformation 
(ref. Section~\ref{sec:graph-transformations}).
The relations $N_P$ and $N_S$ are relevant for node addition and node deletion transformations
and are defined in Sections~\ref{sec:trans.node-add} and~\ref{sec:trans.node-del} respectively.
The temporal logic $\ctlbp$ also consists of backward temporal operators which can be interpreted 
by inverting the edges of a Kripke structure.
The simulation relations for inverted Kripke structures can be identified similarly.

We now prove the construction of simulation relations noted in Figure~\ref{fig:sim-rel-k-trans}
for some representative transformations.
Let the adjacency matrices of Kripke structures $\mykripke$ and $\mykripke'$ be $\mat$ and $\mat'$ respectively. 

\begin{theorem}
\label{thm:ns-bisim}
If $\statemap$ is the correspondence matrix of a node splitting transformation from $\mykripke$ to $\mykripke'$ then
$\statemap$ is a bisimulation relation between $\mykripke'$ and $\mykripke$, that is,
$\mykripke \;\Kfbsim_{\;\statemap}\; \mykripke'$.
\end{theorem}
\begin{proof}
The proof is derived in two steps:
\begin{enumerate}
\item From Definition~\ref{def:ns}, $\statemap \Mm \mat \Mm \Tr{\statemap} = \mat'$.
From Lemma~\ref{lmm:ns-mapping}, for all $\vec{z} \in \bool_n$, $\vec{z} = \Tr{\statemap} \Mm \statemap \Mm \vec{z}$.
Hence, $\statemap$ is a simulation relation between $\mykripke'$ and $\mykripke$, that is,
$\mykripke \;\Kfsim_{\statemap}\; \mykripke'$.

\item We now show that $\Tr{\statemap}$ is a simulation relation between $\mykripke$ and $\mykripke'$, that is,
$\mykripke' \;\Kfsim_{\Tr{\statemap}}\; \mykripke$.

(1)~We prove that $\Tr{\statemap} \Mm \mat' \Mm \statemap = \mat$.
\[
\setstretch{1.4}
\begin{array}{lcl@{\hspace{5mm}}r}
\statemap \Mm \mat \Mm \Tr{\statemap} &=& \mat' & \comment{1}{Definition~\ref{def:ns}}\pnl
\Tr{\statemap} \Mm (\statemap \Mm \mat \Mm \Tr{\statemap}) &=& \Tr{\statemap} \Mm \mat' &
\mycomment{2}{From~(\ref{eqn:mono-3})}{monotonic}\pnl
\Tr{\statemap} \Mm \statemap \Mm (\mat \Mm \Tr{\statemap}) &=& \Tr{\statemap} \Mm \mat' &
\mycomment{2}{From~(\ref{eqn:assoc-3})}{associative}\pnl
\mat \Mm \Tr{\statemap} &=& \Tr{\statemap} \Mm \mat' &
\comment{1}{Lemma~\ref{lmm:ns-mapping}}\pnl
(\mat \Mm \Tr{\statemap}) \Mm \statemap &=& \Tr{\statemap} \Mm \mat' \Mm \statemap &
\mycomment{2}{From~(\ref{eqn:mono-11})}{monotonic}\pnl
\mat \Mm (\Tr{\statemap} \Mm \statemap) &=& \Tr{\statemap} \Mm \mat' \Mm \statemap &
\mycomment{2}{From~(\ref{eqn:assoc-3})}{associative}\pnl
\mat &=& \Tr{\statemap} \Mm \mat' \Mm \statemap &
\comment{1}{Lemma~\ref{lmm:ns-mapping}} 
\end{array}
\]

(2)~As a special case of Lemma~\ref{lmm:ns-mapping},
$\vec{z'} \leq \statemap \Mm \Tr{\statemap} \Mm \vec{z'}$.
\end{enumerate}
\hfill
\end{proof}

Similarly, we can prove that a node merging transformation constructs a simulation relation
between $\mykripke'$ and $\mykripke$.
However, it does not create a bisimulation relation between $\mykripke$ and $\mykripke'$. 
Suppose nodes $p$ and $q$ are merged into a node $p'$ and
there are edges $\pair{p}{r}$ and $\pair{q}{s}$. If $r'$ and $s'$ are the nodes
corresponding to $r$ and $s$ then in the transformed Kripke structure, we have edges
$\pair{p'}{r'}$ and $\pair{p'}{s'}$ which 
correspond to edges $\pair{p}{r}$, $\pair{q}{r}$, $\pair{p}{s}$, and $\pair{q}{s}$
of the input Kripke structure.
Thus, the transformed Kripke structure has edges which do not correspond to any edges in the input graph.

\begin{theorem}
\label{lmm:ea-fsim}
If $\statemap$ is the correspondence matrix of an edge addition transformation from $\mykripke$ to $\mykripke'$ then
$\statemap$ is a simulation relation between $\mykripke'$ and $\mykripke$ i.e.
$\mykripke \;\Kfsim_{\statemap}\; \mykripke'$.
\end{theorem}
\begin{proof}

(1)~We prove that $\statemap \Mm \mat \Mm \Tr{\statemap} < \mat'$. 
\[
\setstretch{1.2}
\begin{array}{lcl@{\hspace{5mm}}r}
\mat & < & \mat \Pl E & \comment{1}{Definition \ref{def:ea}}\pnl
\statemap \Mm \mat & < & \statemap \Mm (\mat \Pl E) &
\mycomment{2}{From~(\ref{eqn:mono-12})}{monotonic, $\statemap$ is a bijection}\pnl
\statemap \Mm \mat \Mm \Tr{\statemap} & < & \statemap \Mm (\mat \Pl E) \Mm \Tr{\statemap} &
\mycomment{2}{From~(\ref{eqn:mono-10})}{monotonic}\pnl
\statemap \Mm \mat \Mm \Tr{\statemap} & < & \statemap \Mm \mat \Mm \Tr{\statemap} \Pl \statemap \Mm E \Mm \Tr{\statemap}
= \mat' &
\mycomment{2}{From~(\ref{eqn:distr-5}) and (\ref{eqn:distr-9}); Definition~\ref{def:ea}}
{distributive, Definition~\ref{def:ea}}
\end{array}
\]

(2)~As a special case of Lemma~\ref{lmm:ea-mapping}, $\vec{z} = \Tr{\statemap} \Mm \statemap \Mm \vec{z}$.
\hfill
\end{proof}

An edge addition transformation adds at least one new edge to the transformed Kripke structure.
Therefore, $\statemap$ cannot be a bisimulation relation 
between $\mykripke'$ and $\mykripke$. 

An edge deletion transformation deletes at least one edge from the input Kripke structure.
Therefore $\Tr{\statemap}$ is a simulation relation relation $\mykripke$ and $\mykripke'$.
Clearly, $\statemap$ is not a bisimulation between $\mykripke$ and $\mykripke'$.

\begin{theorem}
\label{thm:na-weak-bisim}
If $\statemap$ is the correspondence matrix of a node addition transformation from $\mykripke$ to $\mykripke'$
and the relations $N_P$ and $N_S$ are as defined in Definition~\ref{def:na} then 
the triple $\langle \statemap, N_P, N_S \rangle$ is a one-step weak bisimulation relation between $\mykripke'$ and $\mykripke$ i.e.
$\mykripke \;\Kwbisim_{\langle \statemap, N_P, N_S \rangle}\; \mykripke'$.
\end{theorem}
\begin{proof}
\renewcommand{\Mm}{\mbox{\ensuremath{\cdot}}}
Let $E$ be the edges to be split and $\newstates$ be the set of new nodes in $\mykripke'$.
The proof is derived as follows:

\noindent
(1)~We prove that 
$\statemap \Mm \mat \Mm \Tr{\statemap} = 
(\mat' \Mn \statemap \Mm N_P \Mn N_S \Mm \Tr{\statemap}) \Pl \statemap \Mm N_P \Mm N_S \Mm \Tr{\statemap}$.
\[
\setstretch{1.3}
\renewcommand{\Mm}{\mbox{\ensuremath{\cdot}}}
\begin{array}{@{}lcl@{\hspace{5mm}}r@{}}
(\statemap \Mm \mat \Mm \Tr{\statemap} \Mn \statemap \Mm E \Mm \Tr{\statemap}) \Pl
\statemap \Mm N_P \Pl N_S \Mm \Tr{\statemap} & = & \mat' & \comment{1}{Definition~\ref{def:na}}\pnl
\statemap \Mm \mat \Mm \Tr{\statemap} \Mn \statemap \Mm E \Mm \Tr{\statemap} & = & 
\mat' \Mn \statemap \Mm N_P \Mn N_S \Mm \Tr{\statemap} & \comment{1}{}\pnl
(\statemap \Mm \mat \Mm \Tr{\statemap} \Mn \statemap \Mm E \Mm \Tr{\statemap}) \Pl \statemap \Mm E \Mm \Tr{\statemap} & = & 
(\mat' \Mn \statemap \Mm N_P \Mn N_S \Mm \Tr{\statemap}) \Pl \statemap \Mm E \Mm \Tr{\statemap} & \comment{1}{}\pnl
\statemap \Mm \mat \Mm \Tr{\statemap} & = & 
(\mat' \Mn \statemap \Mm N_P \Mn N_S \Mm \Tr{\statemap}) \Pl \statemap \Mm E \Mm \Tr{\statemap} & \comment{1}{$E \leq \mat$}\pnl
\statemap \Mm \mat \Mm \Tr{\statemap} & = & 
(\mat' \Mn \statemap \Mm N_P \Mn N_S \Mm \Tr{\statemap}) \Pl \statemap \Mm N_P \Mm N_S \Mm \Tr{\statemap} & 
\comment{1}{Definition~\ref{def:na}}
\end{array}
\]

\noindent
(2)~$\statemap \Mm N_P \leq \mat'$ and
(3)~$N_S \Mm \Tr{\statemap} \leq \mat'$.

\noindent
(4)~$N_S \Mm \vecone = \Tr{N_P} \Mm \vecone = \vec{\newstates} = \vecone \Mn \statemap \Mm \vecone$,
(5)~$N_P \Mm \statemap = \zero$, and
(6)~$\Tr{\statemap} \Mm N_S = \zero$.

\noindent
(7)~As a special case of Lemma~\ref{lmm:na-mapping}, $\vec{z} = \Tr{\statemap} \Mm \statemap \Mm \vec{z}$.

\noindent
(8)~From the nature of $\statemap$, 
$\vec{z'} \St \Ng{\vec{\newstates}} = \statemap \Mm \Tr{\statemap} \Mm (\vec{z'} \St \Ng{\vec{\newstates}})$.

\noindent
(9)~From the nature of $\statemap$, 
$\Tr{\statemap} \Mm (\vec{z'} \St \Ng{\vec{\newstates}}) = \Tr{\statemap} \Mm \vec{z'}$.

\noindent
(10)~We prove that $\vec{\newstates} \St (\mat' \Mm \vec{\newstates} \Pl \Tr{\mat'} \Mm \vec{\newstates}) = \veczero$
as follows:
\[
\setstretch{1.3}
\renewcommand{\Mm}{\mbox{\ensuremath{\cdot}}}
\begin{array}{@{}lcl@{\hspace{5mm}}r@{}}
\vec{\newstates} \St \mat' \Mm \vec{\newstates} & = & \vec{\newstates} \St 
((\statemap \Mm \mat \Mm \Tr{\statemap} \Mn \statemap \Mm E \Mm \Tr{\statemap}) \Pl
\statemap \Mm N_P \Pl N_S \Mm \Tr{\statemap}) \Mm \vec{\newstates} & \mycomment{1}{Definition~\ref{def:na}}{}\pnl
& = & \vec{\newstates} \St 
\statemap \Mm N_P \Mm \vec{\newstates} & \mycomment{1}{Definition~\ref{def:na}}{}\pnl
& = & (\vecone \Mn \statemap \Mm \vecone) \St 
\statemap \Mm N_P \Mm \vec{\newstates} & \mycomment{1}{Definition~\ref{def:na}}{}\pnl
& \leq & (\vecone \St \Ng{\statemap \Mm \vecone}) \St 
\statemap \Mm \vecone = \veczero & \mycomment{1}{}{}
\end{array}
\]
\[
\setstretch{1.3}
\renewcommand{\Mm}{\mbox{\ensuremath{\cdot}}}
\begin{array}{@{}lcl@{\hspace{5mm}}r@{}}
\vec{\newstates} \St \Tr{\mat'} \Mm \vec{\newstates} & = & \vec{\newstates} \St 
((\statemap \Mm \Tr{\mat} \Mm \Tr{\statemap} \Mn \statemap \Mm \Tr{E} \Mm \Tr{\statemap}) \Pl
\Tr{N_P} \Mm \Tr{\statemap} \Pl \statemap \Mm \Tr{N_S}) \Mm \vec{\newstates} & \mycomment{1}{Definition~\ref{def:na}}{}\pnl
& = & \vec{\newstates} \St 
\statemap \Mm \Tr{N_S} \Mm \vec{\newstates} & \mycomment{1}{Definition~\ref{def:na}}{}\pnl
& = & (\vecone \Mn \statemap \Mm \vecone) \St 
\statemap \Mm \Tr{N_S} \Mm \vec{\newstates} & \mycomment{1}{Definition~\ref{def:na}}{}\pnl
& \leq & (\vecone \St \Ng{\statemap \Mm \vecone}) \St 
\statemap \Mm \vecone = \veczero & \mycomment{1}{}{}
\end{array}
\]
Hence, $\vec{\newstates} \St (\mat' \Mm \vec{\newstates} \Pl \Tr{\mat'} \Mm \vec{\newstates}) = \veczero$.
\hfill
\end{proof}

Similarly, we can show that
if $\statemap$ is the correspondence matrix of a node deletion transformation from $\mykripke$ to $\mykripke'$
and the relations $N_P$ and $N_S$ are as defined in Definition~\ref{def:nd} then 
the triple $\langle \Tr{\statemap},N_P,N_S \rangle$ is a one-step weak bisimulation relation between $\mykripke$ and $\mykripke'$ i.e.
$\mykripke' \;\Kwbisim_{\langle \Tr{\statemap}, N_P, N_S \rangle}\; \mykripke$.

The case of an isomorphic transformation is straightforward.

\subsection{Correlating temporal properties}
\label{sec:corr}

In this section, we prove correlations of temporal properties between Kripke structures
that are related by various simulation relations.
In Appendix~\ref{sec:ctlbp}, we give boolean matrix algebraic semantics of $\ctlbp$ formulae.
Some temporal operators are given fixed point semantics.
Since Kripke structures can have different number of nodes,
in Appendix~\ref{sec:vector-space}, we define an ordering of functions defined over different vector spaces
and prove relations between the least and greatest fixed points of such functions.
For brevity, we consider correlations of future temporal properties under simulation and weak bisimulation relations.
We refer the reader to~\cite{kanade:thesis} for the complete set of correlations and their proofs.

\subsubsection{Existential future properties}
\label{sec:existential-future}

\begin{figure}
\fbox{
\begin{minipage}{\minipagewidth}
\centering
\begin{tabular}{cc}

\begin{tabular}{@{}c@{}}
\begin{tabular}{@{}c@{}}
$R \Mm \vec{\varphi} \leq \vec{\varphi'}$\\[0.6mm]
\end{tabular}\\\cline{1-1}
\begin{tabular}{c}
\\[-2.5mm]
$R \Mm \mat \Mm \vec{\varphi} \leq \mat' \Mm \vec{\varphi'}$
\end{tabular}\\[2mm](SR:EX)
\end{tabular}

&

\begin{tabular}{@{}c@{}}
\begin{tabular}{@{}cc@{}}
$R \Mm \vec{\varphi} \leq \vec{\varphi'}$
&
$R \Mm \vec{\psi} \leq \vec{\psi'}$\\[0.6mm]
\end{tabular}\\\cline{1-1}
\begin{tabular}{@{}c@{}}
\\[-2.5mm]
$R \Mm (\mu \vec{z}.\; \vec{\psi} \Pl (\vec{\varphi} \St \mat \Mm \vec{z})) \leq
\mu \vec{z'}.\; \vec{\psi'} \Pl (\vec{\varphi'} \St \mat' \Mm \vec{z'})$
\end{tabular}\\[2mm](SR:EU)
\end{tabular}

\\[8mm]


\begin{tabular}{@{}c@{}}
\begin{tabular}{@{}cc@{}}
$R \Mm \vec{\varphi} \leq \vec{\varphi'}$
&
$R \Mm \vec{\psi} \leq \vec{\psi'}$\\[0.6mm]
\end{tabular}\\\cline{1-1}
\begin{tabular}{c}
\\[-2.5mm]
$R \Mm (\nu \vec{z}.\; \vec{\psi} \Pl (\vec{\varphi} \St \mat \Mm \vec{z})) \leq
\nu \vec{z'}.\; \vec{\psi'} \Pl (\vec{\varphi'} \St \mat' \Mm \vec{z'})$
\end{tabular}\\[2mm](SR:EW)
\end{tabular}

&

\begin{tabular}{@{}c@{}}
\begin{tabular}{@{}c@{}}
$R \Mm \vec{\varphi} \leq \vec{\varphi'}$\\[0.6mm]
\end{tabular}\\\cline{1-1}
\begin{tabular}{@{}c}
\\[-2.5mm]
$R \Mm (\nu \vec{z}.\; \vec{\varphi} \St \mat \Mm \vec{z}) \leq
\nu \vec{z'}.\; \vec{\varphi'} \St \mat' \Mm \vec{z'}$
\end{tabular}\\[2mm](SR:EG)
\end{tabular}

\end{tabular}
\end{minipage}}
\caption{Correlations of existential future properties under a simulation relation $\mykripke \Kfsim_{R} \mykripke'$}
\label{fig:sim-rel-EX}
\end{figure}

\newcommand{\Rt}{\tau(R)}
The existential future $\ctlbp$ operators are $\EX$, $\EU$, $\EW$, $\EF$, and $\EG$
where $\EF$ is a special case of $\EU$.

Consider a simulation relation $R$ between Kripke structures $\mykripke$ and $\mykripke'$.
For every edge in $\mykripke$, there is an edge in $\mykripke'$. Thus, we can correlate
existential future properties between $\mykripke$ and $\mykripke'$.
The correlations are given in Figure~\ref{fig:sim-rel-EX} and their soundness is proved below. 

\begin{theorem}
\label{thm:ex}
If $R$ is a simulation relation between Kripke structures $\mykripke$ and $\mykripke'$
i.e. $\mykripke \Kfsim_R \mykripke'$ then the correlations between
the future temporal properties given in Figure~\ref{fig:sim-rel-EX} are sound.
\end{theorem}
\begin{proof}
We first prove soundness of the rule SR:EX.
\[
\setstretch{1.2}
\begin{array}{lcl@{\hspace{4cm}}r}
R \Mm \vec{\varphi} & \leq & \vec{\varphi'} & \comment{1}{given} \pnl
R \Mm \mat \Mm \Tr{R} \Mm (R \Mm \vec{\varphi}) & \leq & \mat' \Mm \vec{\varphi'} & 
\mycomment{2}{From (\ref{eqn:mono-4})}{monotonic, Definition~\ref{def:kripke-sim}} \pnl
R \Mm \mat \Mm (\Tr{R} \Mm R \Mm \vec{\varphi}) & \leq & \mat' \Mm \vec{\varphi'} &
\mycomment{2}{From (\ref{eqn:mono-4})}{associative} \pnl
R \Mm \mat \Mm \vec{\varphi} & \leq & \mat' \Mm \vec{\varphi'} &
\mycomment{2}{Def.~\ref{def:kripke-sim}(2)}{Definition~\ref{def:kripke-sim}, monotonic} 
\end{array}
\]

Note that the rules SR:EU, SR:EW, and SR:EG correlate fixed points of certain functions.
To prove soundness of these rules, we use the following strategy.
We first establish an ordering as per Definition~\ref{def:fun-space}
between the respective functions. We then use Lemmas~\ref{lmm:ns-lfp} and \ref{lmm:ns-gfp}
to correlate their fixed points. We demonstrate this strategy by proving
soundness of the rule SR:EU. 
Let $\vec{z}$ be a boolean vector defined on $\mykripke$.
\[
\setstretch{1.2}
\begin{array}{lcl@{\hspace{-3mm}}r}
R \Mm \vec{z} & \leq & R \Mm \vec{z} & \mycomment{2}{From~(\ref{eqn:misc-4})}{identity}\pnl
\nonumber
R \Mm \mat \Mm \vec{z} & \leq & \mat' \Mm (R \Mm \vec{z}) &
\comment{1}{substitute $[\vec{z}/\vec{\varphi},R \Mm \vec{z}/\vec{\varphi'}]$ in SR:EX}\pnl
\nonumber 
R \Mm \vec{\varphi} & \leq & \vec{\varphi'} & \mycomment{2}{assumption}{given}\pnl
\nonumber
(R \Mm \vec{\varphi}) \St (R \Mm \mat \Mm \vec{z}) & \leq &
\vec{\varphi'} \St (\mat' \Mm (R \Mm \vec{z})) &
\mycomment{2}{From~(\ref{eqn:mono-5})}{monotonic}\pnl
\nonumber
R \Mm (\vec{\varphi} \St \mat \Mm \vec{z}) & \leq &
\vec{\varphi'} \St (\mat' \Mm (R \Mm \vec{z})) &
\mycomment{2}{From~(\ref{eqn:distr-6})}{distributive}
\pnl
\nonumber
R \Mm \vec{\psi} & \leq & \vec{\psi'} & \mycomment{1}{given}{given}\pnl
\nonumber
(R \Mm \vec{\psi}) \Pl \left(R \Mm (\vec{\varphi} \St \mat \Mm \vec{z})\right) & \leq &
\vec{\psi'} \Pl (\vec{\varphi'} \St \mat' \Mm (R \Mm \vec{z})) &
\mycomment{2}{From~(\ref{eqn:mono-6})}{monotonic}
\pnl
\nonumber
R \Mm (\vec{\psi} \Pl (\vec{\varphi} \St \mat \Mm \vec{z})) & \leq &
\vec{\psi'} \Pl (\vec{\varphi'} \St \mat' \Mm (R \Mm \vec{z})) &
\mycomment{2}{From~(\ref{eqn:distr-5})}{distribute}
\pnl
\lambda \vec{z}.\; \vec{\psi} \Pl (\vec{\varphi} \St \mat \Mm \vec{z}) & \sqsubseteq_{R} &
\lambda \vec{z'}.\; \vec{\psi}' \Pl (\vec{\varphi'} \St \mat' \Mm \vec{z'}) &
\mycomment{2}{Def. \ref{def:fun-space}}{Definition~\ref{def:fun-space}}\pnl
\nonumber
R \Mm (\mu \vec{z}.\; \vec{\psi} \Pl (\vec{\varphi} \St \mat \Mm \vec{z})) & \leq &
\mu \vec{z'}.\; \vec{\psi'} \Pl (\vec{\varphi'} \St \mat' \Mm \vec{z'}) &
\comment{1}{Lemma \ref{lmm:ns-lfp}}
\end{array}
\]
\hfill
\end{proof}

\begin{figure}

\fbox{
\begin{minipage}{\minipagewidth}
\begin{tabular}{cc}

\begin{tabular}{c}
 \begin{tabular}{@{}cc@{}}
$R \Mm \vec{\varphi} \leq \vec{\varphi'}$
&
$\vec{\newstates} \leq \vec{\varphi'}$\\[0.8mm]
 \end{tabular}\\\cline{1-1}
\begin{tabular}{c}
\\[-2.5mm]
$R \Mm \mat \Mm \vec{\varphi} \leq \mat' \Mm \vec{\varphi'}$
 \end{tabular}\\[2mm](WB:EX)
\end{tabular}

&

\begin{tabular}{c}
 \begin{tabular}{@{}ccc@{}}
$R \Mm \vec{\varphi} \leq \vec{\varphi'}$
&
$\vec{\newstates} \leq \vec{\varphi'}$
& 
$R \Mm \vec{\psi} \leq \vec{\psi'}$\\[0.8mm]
 \end{tabular}\\\cline{1-1}
\begin{tabular}{c}
\\[-2.5mm]
$R \Mm (\mu \vec{z}.\; \vec{\psi} \Pl (\vec{\varphi} \St \mat \Mm \vec{z})) \leq
\mu \vec{z'}.\; \vec{\psi'} \Pl (\vec{\varphi'} \St \mat' \Mm \vec{z'})$
 \end{tabular}\\[2mm](WB:EU)
\end{tabular}

\\[8mm]


\begin{tabular}{c}
 \begin{tabular}{@{}ccc@{}}
$R \Mm \vec{\varphi} \leq \vec{\varphi'}$
&
$\vec{\newstates} \leq \vec{\varphi'}$
& 
$R \Mm \vec{\psi} \leq \vec{\psi'}$\\[0.8mm]
 \end{tabular}\\\cline{1-1}
\begin{tabular}{c}
\\[-2.5mm]
$R \Mm (\nu \vec{z}.\; \vec{\psi} \Pl (\vec{\varphi} \St \mat \Mm \vec{z})) \leq
\nu \vec{z'}.\; \vec{\psi'} \Pl (\vec{\varphi'} \St \mat' \Mm \vec{z'})$
 \end{tabular}\\[2mm](WB:EW)
\end{tabular}

&

\begin{tabular}{c}
 \begin{tabular}{@{}cc@{}}
 $R \Mm \vec{\varphi} \leq \vec{\varphi'}$ 
 & 
 $\vec{\newstates} \leq \vec{\varphi'}$\\[0.8mm]
 \end{tabular}\\\cline{1-1}
\begin{tabular}{c}
\\[-2.5mm]
$R \Mm (\nu \vec{z}.\; \vec{\varphi} \St \mat \Mm \vec{z}) \leq 
\nu \vec{z'}.\; \vec{\varphi'} \St \mat' \Mm \vec{z'}$
 \end{tabular}\\[2mm](WB:EG)
\end{tabular}

\end{tabular}
\end{minipage}}
\caption{Correlations of existential future properties under 
a weak bisimulation relation $\mykripke \Kwbisim_{\langle R, P, S \rangle} \mykripke'$}
\label{fig:weak-bisim-I-EX}
\end{figure}


\begin{theorem}
\label{thm:wb-ex}
If $\langle R,P,S \rangle$ is a weak bisimulation between Kripke structures $\mykripke$ and $\mykripke'$
i.e. $\mykripke \Kwbisim_{\langle R, P, S \rangle} \mykripke'$ 
and $\vec{\newstates} = \vecone \;\Mn\; R \Mm \vecone$ is the set of internal nodes of $\mykripke'$
then the correlations between future temporal properties
given in Figure~\ref{fig:weak-bisim-I-EX} are sound.
\end{theorem}
\begin{proof}
We first prove soundness of the rule WB:EX.
\[
\setstretch{1.2}
\begin{array}{rlcl@{\hspace{1.1cm}}r}
& S \Mm (\Tr{R} \Mm \vec{\varphi'}) & \leq & S \Mm \vecone = \vec{\newstates} & \mycomment{2}{given}{by definition}\pnl
(i) & R \Mm P \Mm (S \Mm \Tr{R} \Mm \vec{\varphi'}) & \leq & 
R \Mm P \Mm \vec{\newstates} \leq \mat' \Mm \vec{\newstates} \leq \mat' \Mm \vec{\varphi'}
& \mycomment{2}{given}{monotonic, Definition~\ref{def:weak-bisim}, given}\pnl
\end{array}
\]

The main derivation is as follows:
\[
\setstretch{1.2}
\begin{array}{rlclr}
& R \Mm \vec{\varphi} & \leq & \vec{\varphi'} & \mycomment{2}{given}{given}\pnl
& R \Mm \mat \Mm \Tr{R} \Mm (R \Mm \vec{\varphi}) & \leq & (\mat' \;\Pl\; R \Mm P \Mm S \Mm \Tr{R}) \Mm \vec{\varphi'} & 
\mycomment{2}{}{Definition~\ref{def:weak-bisim}, monotonic}\pnl
& R \Mm \mat \Mm (\Tr{R} \Mm R \Mm \vec{\varphi}) & \leq & 
\mat' \Mm \vec{\varphi'} \;\Pl\; R \Mm P \Mm S \Mm \Tr{R} \Mm \vec{\varphi'} & 
\mycomment{2}{}{associative, distributive}\pnl
(ii) & R \Mm \mat \Mm \vec{\varphi} & \leq & \mat' \Mm \vec{\varphi'} \;\Pl\; R \Mm P \Mm (S \Mm \Tr{R} \Mm \vec{\varphi'}) & 
\mycomment{2}{}{Definition~\ref{def:weak-bisim}, monotonic, associative}\pnl
& R \Mm \mat \Mm \vec{\varphi} & \leq & \mat' \Mm \vec{\varphi'} & \mycomment{2}{}{from (i) above}
\end{array}
\]

Note that the rules WB:EU, WB:EW, and WB:EG correlate fixed points of certain functions.
To prove soundness of these rules, we use the following strategy.
Suppose the respective functions are $f$ and $f'$. The functions are interpreted over different 
Kripke structures $\mykripke$ and $\mykripke'$ which
are related by a one-step weak bisimulation. Hence, $\mykripke'$
consists of internal nodes. Function $f'$ does not distinguish between visible and internal nodes of $\mykripke'$.
We therefore construct a function $g$ which is similar to $f'$ but \emph{distinguishes between visible and internal nodes}.
We show that $f$ and $g$ are related
by an ordering as per Definition~\ref{def:fun-space}. We then correlate fixed points of
$f$ and $g$ using Lemmas~\ref{lmm:ns-lfp} and \ref{lmm:ns-gfp}. 
Finally, we show that the required fixed point of $g$ is smaller than the required fixed point of $f'$.
We demonstrate this strategy by proving soundness of the rule WB:EU.
For the rule WB:EU, let $f = \lambda \vec{z}.\; \vec{\psi} \Pl (\vec{\varphi} \St \mat \Mm \vec{z})$ and
$f' = \lambda \vec{z'}.\; \vec{\psi'} \Pl (\vec{\varphi'} \St \mat' \Mm \vec{z'})$.
\[
\setstretch{1.2}
\begin{array}{rlcl@{\hspace{-3mm}}r}
& R \Mm \vec{z} & \leq & R \Mm \vec{z} & \mycomment{2}{From~(\ref{eqn:misc-3})}{identity}\pnl
& R \Mm \mat \Mm \vec{z} & \leq & \mat' \Mm (R \Mm \vec{z}) \;\Pl\; R \Mm P \Mm S \Mm \Tr{R} \Mm (R \Mm \vec{z})
& \mycomment{2}{From~(\ref{eqn:misc-3})}{from ($ii$) of Thm.~\ref{thm:wb-ex}}\pnl
& R \Mm \mat \Mm \vec{z} & \leq & \mat' \Mm (R \Mm \vec{z}) \;\Pl\; (R \Mm P) \Mm S \Mm \Tr{R} \Mm (R \Mm \vec{z})
& \mycomment{2}{From~(\ref{eqn:misc-3})}{associative}\pnl
& R \Mm \mat \Mm \vec{z} & \leq & \mat' \Mm (R \Mm \vec{z}) \;\Pl\; \mat' \Mm S \Mm \Tr{R} \Mm (R \Mm \vec{z})
& \mycomment{2}{From~(\ref{eqn:misc-3})}{Def.~\ref{def:weak-bisim}, monotonic}\pnl
& R \Mm \mat \Mm \vec{z} & \leq & \mat' \Mm (R \Mm \vec{z}) \;\Pl\; \mat' \Mm S \Mm \Tr{R} \Mm (\vecone \;\St\; R \Mm \vec{z})
& \mycomment{2}{From~(\ref{eqn:misc-3})}{identity}\pnl
& R \Mm \mat \Mm \vec{z} & \leq & \mat' \Mm (R \Mm \vec{z}) \;\Pl\; \mat' \Mm (S \Mm \Tr{R} \Mm \vecone \;\St\; S \Mm \Tr{R} \Mm R \Mm \vec{z})
& \mycomment{2}{From~(\ref{eqn:misc-3})}{distributive}\pnl
& R \Mm \mat \Mm \vec{z} & \leq & \mat' \Mm (R \Mm \vec{z}) \;\Pl\; \mat' \Mm (\vec{\newstates} \;\St\; (S \Mm \Tr{R}) \Mm R \Mm \vec{z})
& \mycomment{2}{From~(\ref{eqn:misc-3})}{Def.~\ref{def:weak-bisim}, associative}\pnl
(i) & R \Mm \mat \Mm \vec{z} & \leq & \mat' \Mm (R \Mm \vec{z}) \;\Pl\; \mat' \Mm (\vec{\newstates} \;\St\; \mat' \Mm R \Mm \vec{z})
& \mycomment{2}{From~(\ref{eqn:misc-3})}{Def.~\ref{def:weak-bisim}, monotonic}\pnl
& R \Mm (\vec{\psi} \Pl (\vec{\varphi} \St \mat \Mm \vec{z})) & \leq & 
\vec{\psi'} \Pl (\vec{\varphi'} \St (\mat' \Mm R \Mm \vec{z} \Pl \mat' \Mm (\vec{\newstates} \St \mat' \Mm R \Mm \vec{z})))
& \mycomment{2}{From~(\ref{eqn:misc-3})}{monotonic}\pnl
(ii) & \lambda \vec{z}. \vec{\psi} \Pl (\vec{\varphi} \St \mat \Mm \vec{z}) & \sqsubseteq_{R} & 
\underbrace{
\lambda \vec{z'}. \vec{\psi'} \Pl (\vec{\varphi'} \St (\mat' \Mm \vec{z'} \Pl \mat' \Mm (\vec{\newstates} \St \mat' \Mm \vec{z'})))
}_{g}
& \mycomment{2}{From~(\ref{eqn:misc-3})}{Def.~\ref{def:fun-space}}\pnl
(iii) & R \Mm (\mu \vec{z}. \vec{\psi} \Pl (\vec{\varphi} \St \mat \Mm \vec{z})) & \leq & 
\mu \vec{z'}. \vec{\psi'} \Pl (\vec{\varphi'} \St (\mat' \Mm \vec{z'} \Pl \mat' \Mm (\vec{\newstates} \St \mat' \Mm \vec{z'})))
& \mycomment{2}{From~(\ref{eqn:misc-3})}{Lemma~\ref{lmm:ns-lfp}}
\end{array}
\]

We now show that $\mu \vec{z'}. g(\vec{z'}) \leq \mu \vec{z'}. \vec{\psi'} \Pl (\vec{\varphi'} \St \mat' \Mm \vec{z'})$.
Let $\vec{q'} = \mu \vec{z'}. \vec{\psi'} \Pl (\vec{\varphi'} \St \mat' \Mm \vec{z'})$.
\[
\setstretch{1.25}
\begin{array}{rlcl@{\hspace{1.2cm}}r}
& \vec{\newstates} & \leq & \vec{\varphi'} & \mycomment{1}{given}{}\pnl
& \vec{\newstates} \St \mat' \Mm \vec{q'} & \leq & \vec{\varphi'} \St \mat' \Mm \vec{q'} \leq \vec{q'} &
\mycomment{2}{From~(\ref{eqn:mono-5}); Def. of $\vec{q'}$}{monotonic; Def. of $\vec{q'}$}\pnl
& \mat' \Mm (\vec{\newstates} \St \mat' \Mm \vec{q'}) & \leq & \mat' \Mm \vec{q'} & 
\mycomment{2}{From~(\ref{eqn:mono-4})}{monotonic}\pnl
& \mat' \Mm \vec{q'} \Pl \mat' \Mm (\vec{\newstates} \St \mat' \Mm \vec{q'}) & = &
\mat' \Mm \vec{q'} & \mycomment{2}{From~(\ref{eqn:mono-6})}{}\pnl
(iv) & \vec{\psi'} \Pl (\vec{\varphi'} \St (\mat' \Mm \vec{q'} \Pl \mat' \Mm (\vec{\newstates} \St \mat' \Mm \vec{q'}))) & = &
\vec{\psi'} \Pl (\vec{\varphi'} \St \mat' \Mm \vec{q'}) = \vec{q'} &
\mycomment{2}{monotonicity; Def. of $\vec{q'}$}{monotonic; Def. of $\vec{q'}$}
\end{array}
\]
Thus, $\vec{q'}$ is a fixed point of $g$. Hence, $\mu \vec{z'}. g(\vec{z'}) \leq \vec{q'}$.
From ($iii$) and ($iv$), by transitivity, we get the following result:
$
R \Mm (\mu \vec{z}.\; \vec{\psi} \Pl (\vec{\varphi} \St \mat \Mm \vec{z})) \;\leq\;
\mu \vec{z'}.\; \vec{\psi'} \Pl (\vec{\varphi'} \St \mat' \Mm \vec{z'})
$.
\hfill
\end{proof}

\subsubsection{Universal future properties}
\label{sec:universal-future}

\begin{figure}
\fbox{
\scalebox{1}{
\begin{minipage}{\minipagewidth}

\begin{tabular}{@{}c@{}c@{}}

\begin{tabular}{c}
\begin{tabular}{@{}ccc@{}}
$\Tr{R} \Mm R \Mm \vec{z} \leq \vec{z}$
&
{$\Ng{R \Mm \vec{z}} \!=\! R \Mm \Ng{\vec{z}}$}
&
$R \Mm \vec{\varphi} \leq \vec{\varphi'}$\\[0.6mm]
\end{tabular}\\\cline{1-1}
\begin{tabular}{c}
\\[-3mm]
$R \Mm \Ng{\mat \Mm \Ng{\vec{\varphi}}} \;\leq\; \Ng{\mat' \Mm \Ng{\vec{\varphi'}}}$
\end{tabular}\\[2.5mm](SR:AX)
\end{tabular}
\hspace{5mm}
&

\hspace{-1cm}
\begin{tabular}{c}
\begin{tabular}{@{}cccc@{}}
{$\Tr{R} \Mm R \Mm \vec{z} \leq \vec{z}$}
& 
$\Ng{R \Mm \vec{z}} \!=\! R \Mm \Ng{\vec{z}}$
&
$R \Mm \vec{\varphi} \!\leq\! \vec{\varphi'}$
&
$R \Mm \vec{\psi} \!\leq\! \vec{\psi'}$\\[0.6mm]
\end{tabular}\\\cline{1-1}
\begin{tabular}{c}
\\[-3mm]
$R \Mm (\mu \vec{z}.\; \vec{\psi} \Pl (\vec{\varphi} \St \Ng{\mat \Mm \Ng{\vec{z}}})) \leq
\mu \vec{z'}.\; \vec{\psi'} \Pl (\vec{\varphi'} \St \Ng{\mat' \Mm \Ng{\vec{z'}}})$
\end{tabular}\\[2.5mm](SR:AU)
\end{tabular}

\\[10mm]


\begin{tabular}{c}
\begin{tabular}{@{}cccc@{}}
{$\Tr{R} \Mm R \Mm \vec{z} \leq \vec{z}$}
&
$\Ng{R \Mm \vec{z}} \!=\! R \Mm \Ng{\vec{z}}$
&
$R \Mm \vec{\varphi} \leq \vec{\varphi'}$
&
$R \Mm \vec{\psi} \leq \vec{\psi'}$\\[0.6mm]
\end{tabular}\\\cline{1-1}
\begin{tabular}{c}
\\[-3mm]
$R \Mm (\nu \vec{z}.\; \vec{\psi} \Pl (\vec{\varphi} \St \Ng{\mat \Mm \Ng{\vec{z}}})) \leq
\nu \vec{z'}.\; \vec{\psi'} \Pl (\vec{\varphi'} \St \Ng{\mat' \Mm \Ng{\vec{z'}}})$
\end{tabular}\\[2.5mm](SR:AW)
\end{tabular}

& 

\begin{tabular}{c}
\begin{tabular}{@{}ccc@{}}
{$\Tr{R} \Mm R \Mm \vec{z} \leq \vec{z}$}
&
$\Ng{R \Mm \vec{z}} \!=\! R \Mm \Ng{\vec{z}}$
&
$R \Mm \vec{\varphi} \leq \vec{\varphi'}$\\[0.6mm]
\end{tabular}\\\cline{1-1}
\begin{tabular}{c}
\\[-3mm]
$R \Mm (\nu \vec{z}.\; \vec{\varphi} \St \Ng{\mat \Mm \Ng{\vec{z}}}) \leq
\nu \vec{z'}.\; \vec{\varphi'} \St \Ng{\mat' \Mm \Ng{\vec{z'}}}$
\end{tabular}\\[2.5mm](SR:AG)
\end{tabular}

\end{tabular}

\end{minipage}}
}
\caption{Correlations of universal future properties under a simulation relation $\mykripke' \Kfsim_{\Tr{R}} \mykripke$}
\label{fig:sim-rel-AX}
\end{figure}

The universal future $\ctlbp$ operators are $\AX$, $\AU$, $\AW$, $\AF$, and $\AG$
where $\AF$ is a special case of $\AU$. 

Consider a simulation relation $\Tr{R}$ between Kripke structures $\mykripke'$ and $\mykripke$
where $\mykripke' \Kfsim_{\Tr{R}} \mykripke$.
The set of outgoing edges of a node in $\mykripke$ is a superset of the set of outgoining
edges of a corresponding node in $\mykripke'$. 
Since $\mykripke'$ is a Kripke structure,
the set of outgoing edges for any node is non-empty. We can correlate
universal future properties between $\mykripke$ and $\mykripke'$.
The correlations are given in Figure~\ref{fig:sim-rel-AX} and their soundness is proved below.

\begin{theorem}
\label{thm:ax}
If $\Tr{R}$ is a simulation relation between Kripke structures $\mykripke'$ and $\mykripke$
i.e. $\mykripke' \Kfsim_{\Tr{R}} \mykripke$ then the correlations between future temporal
properties given in Figure~\ref{fig:sim-rel-AX} are sound.
\end{theorem}
\begin{proof}
We first determine the nature of the simulation relation from its definition and 
the premises common in the inference rules in Figure~\ref{fig:sim-rel-AX}.

From Definition~\ref{def:kripke-sim}, we know that $\vec{z'} \leq R \Mm \Tr{R} \Mm \vec{z'}$.
Thus, $R$ is total and possibly many-to-one or many-to-many.
From (\ref{eqn:first-R}), we know that $\Ng{R \Mm \vec{z}} = R \Mm \Ng{\vec{z}}$ states that 
$R$ is total but neither one-to-many nor many-to-many.
From (\ref{eqn:second-R}), we know that $\Tr{R} \Mm R \Mm \vec{z} \leq \vec{z}$ states that
$R$ is possibly not onto and neither one-to-many nor many-to-many.
Thus, $R$ is a total and either a one-to-one or a many-to-one relation.

Further, from Definition~\ref{def:kripke-sim}, $\Tr{R} \Mm \mat' \Mm R \leq \mat$.
Hence, the set of outgoing edges of a node in $\mykripke$ is a superset of the set of outgoing
edges of a corresponding node in $\mykripke'$. 

The proof of the rule SR:AX is as follows:
\[
\setstretch{1.2}
\begin{array}{lcl@{\hspace{38mm}}r}
R \Mm \vec{\varphi} & \leq & \vec{\varphi'} & \comment{1}{given}\pnl
\Ng{\vec{\varphi'}} & \leq & \Ng{R \Mm \vec{\varphi}} = R \Mm \Ng{\vec{\varphi}} & 
\mycomment{2}{From~(\ref{eqn:misc-1}); assumption}{negation, given}\pnl
\Tr{R} \Mm \Ng{\vec{\varphi'}} & \leq & \Tr{R} \Mm R \Mm \Ng{\vec{\varphi}} \leq \Ng{\vec{\varphi}} & 
\mycomment{2}{From~(\ref{eqn:misc-1}); assumption}{monotonicity, given}\pnl
\Tr{R} \Mm \mat' \Mm R \Mm (\Tr{R} \Mm \Ng{\vec{\varphi'}}) & \leq & \mat \Mm \Ng{\vec{\varphi}} & 
\mycomment{2}{From~(\ref{eqn:misc-1}); assumption}{Definition~\ref{def:kripke-sim}, monotonic}\pnl
\Tr{R} \Mm \mat' \Mm (R \Mm \Tr{R} \Mm \Ng{\vec{\varphi'}}) & \leq & \mat \Mm \Ng{\vec{\varphi}} & 
\mycomment{2}{From~(\ref{eqn:misc-1}); assumption}{associative}\pnl
R \Mm (\Tr{R} \Mm \mat' \Mm \Ng{\vec{\varphi'}}) & \leq & R \Mm (\mat \Mm \Ng{\vec{\varphi}}) & 
\mycomment{2}{From~(\ref{eqn:misc-1}); assumption}{Definition~\ref{def:kripke-sim}, monotonic}\pnl
R \Mm \Tr{R} \Mm (\mat' \Mm \Ng{\vec{\varphi'}}) & \leq & R \Mm \mat \Mm \Ng{\vec{\varphi}} & 
\mycomment{2}{From~(\ref{eqn:misc-1}); assumption}{associative}\pnl
\mat' \Mm \Ng{\vec{\varphi'}} & \leq & R \Mm \mat \Mm \Ng{\vec{\varphi}}  &
\mycomment{2}{From~(\ref{eqn:assoc-3})}{Definition~\ref{def:kripke-sim}}\pnl
\Ng{R \Mm \mat \Mm \Ng{\vec{\varphi}}} & \leq &
\Ng{\mat' \Mm \Ng{\vec{\varphi'}}} &
\mycomment{2}{From~(\ref{eqn:misc-1})}{negation}\pnl
R \Mm \Ng{\mat \Mm \Ng{\vec{\varphi}}} & \leq &
\Ng{\mat' \Mm \Ng{\vec{\varphi'}}} &
\comment{1}{given}
\end{array}
\]

The proof of the rule SR:AU is as follows:

Let $\vec{z}$ be a boolean vector defined on $\mykripke$.
\[
\setstretch{1.2}
\begin{array}{clclr}
& R \Mm \vec{z} & \leq & R \Mm \vec{z} & \mycomment{2}{From~(\ref{eqn:misc-4})}{identity}\pnl
& R \Mm \Ng{\mat \Mm \Ng{\vec{z}}} & \leq & \Ng{\mat' \Mm (\Ng{R \Mm \vec{z}})} &
\comment{1}{$[\vec{z}/\vec{\varphi},R \Mm \vec{z}/\vec{\varphi'}]$ in Theorem~\ref{thm:ax}}\pnl
& R \Mm \vec{\varphi} & \leq & \vec{\varphi'} & \comment{1}{given}\pnl
& (R \Mm \vec{\varphi}) \St (R \Mm \Ng{\mat \Mm \Ng{\vec{z}}}) & \leq &
\vec{\varphi'} \St (\Ng{\mat' \Mm (\Ng{R \Mm \vec{z}})}) &
\mycomment{2}{From~(\ref{eqn:mono-5})}{monotonic}\pnl
& R \Mm (\vec{\varphi} \St \Ng{\mat \Mm \Ng{\vec{z}}}) & \leq &
\vec{\varphi'} \St (\Ng{\mat' \Mm (\Ng{R \Mm \vec{z}})}) &
\mycomment{2}{From~(\ref{eqn:distr-6})}{distributive}\pnl
& R \Mm \vec{\psi} & \leq & \vec{\psi'} & \comment{1}{assumption}\pnl
\end{array}
\]
\[
\setstretch{1.2}
\begin{array}{clclr}
& (R \Mm \vec{\psi}) \Pl \left(R \Mm (\vec{\varphi} \St \Ng{\mat \Mm \Ng{\vec{z}}})\right) & \leq &
\vec{\psi'} \Pl (\vec{\varphi'} \St \Ng{\mat' \Mm (\Ng{R \Mm \vec{z}})}) &
\mycomment{2}{From~(\ref{eqn:mono-6})}{monotonic}\pnl
& R \Mm (\vec{\psi} \Pl (\vec{\varphi} \St \Ng{\mat \Mm \Ng{\vec{z}}})) & \leq &
\vec{\psi'} \Pl (\vec{\varphi'} \St \Ng{\mat' \Mm (\Ng{R \Mm \vec{z}})}) &
\mycomment{2}{From~(\ref{eqn:distr-5})}{distributive}\pnl
(i) & \lambda \vec{z}.\; \vec{\psi} \Pl (\vec{\varphi} \St \Ng{\mat \Mm \Ng{\vec{z}}}) & \sqsubseteq_{R} &
\lambda \vec{z'}.\; \vec{\psi}' \Pl (\vec{\varphi'} \St \Ng{\mat' \Mm \Ng{\vec{z'}}}) &
\mycomment{2}{Def. \ref{def:fun-space}}{Definition~\ref{def:fun-space}}\pnl
& R \Mm (\mu \vec{z}.\; \vec{\psi} \Pl (\vec{\varphi} \St \Ng{\mat \Mm \Ng{\vec{z}}})) & \leq &
\mu \vec{z'}.\; \vec{\psi'} \Pl (\vec{\varphi'} \St \Ng{\mat' \Mm \Ng{\vec{z'}}}) &
\comment{1}{Lemma \ref{lmm:ns-lfp}} 
\end{array}
\]
The proofs of SR:AW and SR:AG can be derived similarly.
\hfill
\end{proof}

\begin{figure}
\fbox{
\scalebox{1}{
\begin{minipage}{\minipagewidth}
\centering
\begin{tabular}{c}

\begin{tabular}{c}
\renewcommand{\Mm}{\mbox{\ensuremath{\cdot}}}
 \begin{tabular}{@{}cccc@{}}
$\Ng{\Tr{R} \Mm \vec{z'}} \leq \Tr{R} \Mm \Ng{\vec{z'}}$
&
$\Ng{R \Mm \vec{z}} \;\St\; \Ng{\vec{\newstates}} = R \Mm \Ng{\vec{z}}$
&
 $R \Mm \vec{\varphi} \leq \vec{\varphi'}$ 
 & 
 $\vec{\newstates} \leq \vec{\varphi'}$
 \end{tabular}\\\cline{1-1}
\begin{tabular}{c}
\\[-3mm]
$R \Mm \Ng{\mat \Mm \Ng{\vec{\varphi}}} \leq \Ng{\mat \Mm \Ng{\vec{\varphi'}}}$
 \end{tabular}\\[2mm](WB:AX)
\end{tabular}

\\[8mm]

\begin{tabular}{c}
\renewcommand{\Mm}{\mbox{\ensuremath{\cdot}}}
\begin{tabular}{@{\!\!\!}c@{\;\;\;}c@{\;\;\;}ccc@{}}
$\Ng{\Tr{R} \Mm \vec{z'}} \leq \Tr{R} \Mm \Ng{\vec{z'}}$
&
$\Ng{R \Mm \vec{z}} \;\St\; \Ng{\vec{\newstates}} \!=\! R \Mm \Ng{\vec{z}}$
&
$R \Mm \vec{\varphi} \!\leq\! \vec{\varphi'}$
&
$\vec{\newstates} \!\leq\! \vec{\varphi'}$
& 
$R \Mm \vec{\psi} \!\leq\! \vec{\psi'}$
 \end{tabular}\\\cline{1-1}
\begin{tabular}{c}
\\[-3mm]
$R \Mm (\mu \vec{z}.\; \vec{\psi} \Pl (\vec{\varphi} \St \Ng{\mat \Mm \Ng{\vec{z}}})) \leq
\mu \vec{z'}.\; \vec{\psi'} \Pl (\vec{\varphi'} \St \Ng{\mat' \Mm \Ng{\vec{z'}}})$
 \end{tabular}\\[2mm](WB:AU)
\end{tabular}

\\[8mm]

\begin{tabular}{c}
 \begin{tabular}{@{\!\!\!}c@{\;\;}c@{\;\;}c@{\;\;}c@{\;\;}c@{}}
$\Ng{\Tr{R} \Mm \vec{z'}} \leq \Tr{R} \Mm \Ng{\vec{z'}}$
&
$\Ng{R \Mm \vec{z}} \;\St\; \Ng{\vec{\newstates}} \!=\! R \Mm \Ng{\vec{z}}$
&
$R \Mm \vec{\varphi} \!\leq\! \vec{\varphi'}$
&
$\vec{\newstates} \!\leq\! \vec{\varphi'}$
& 
$R \Mm \vec{\psi} \!\leq\! \vec{\psi'}$
 \end{tabular}\\\cline{1-1}
\begin{tabular}{c}
\\[-3mm]
$R \Mm (\nu \vec{z}.\; \vec{\psi} \Pl (\vec{\varphi} \St \Ng{\mat \Mm \Ng{\vec{z}}})) \leq
\nu \vec{z'}.\; \vec{\psi'} \Pl (\vec{\varphi'} \St \Ng{\mat' \Mm \Ng{\vec{z'}}})$
 \end{tabular}\\[2mm](WB:AW)
\end{tabular}

\\[8mm]

\begin{tabular}{c}
\renewcommand{\Mm}{\mbox{\ensuremath{\cdot}}}
 \begin{tabular}{@{}cccc@{}}
$\Ng{\Tr{R} \Mm \vec{z'}} \leq \Tr{R} \Mm \Ng{\vec{z'}}$
&
$\Ng{R \Mm \vec{z}} \;\St\; \Ng{\vec{\newstates}} = R \Mm \Ng{\vec{z}}$
&
 $R \Mm \vec{\varphi} \leq \vec{\varphi'}$ 
 & 
 $\vec{\newstates} \leq \vec{\varphi'}$
 \end{tabular}\\\cline{1-1}
\begin{tabular}{c}
\\[-3mm]
$R \Mm (\nu \vec{z}.\; \vec{\varphi} \St \Ng{\mat \Mm \Ng{\vec{z}}}) \leq
\nu \vec{z'}.\; \vec{\varphi'} \St \Ng{\mat' \Mm \Ng{\vec{z'}}}$
 \end{tabular}\\[2mm](WB:AG)
\end{tabular}

\end{tabular}
\end{minipage}}
}
\caption{Correlations of universal future properties under a weak bisimulation relation
$\mykripke \Kwbisim_{\langle R, P, S \rangle} \mykripke'$}
\label{fig:weak-bisim-I:AX}
\end{figure}

Consider a one-step weak bisimulation $\mykripke \Kwbisim_{\langle R, P, S \rangle} \mykripke'$ where
$\newstates$ denotes the set of internal nodes in $\mykripke'$. 
The correlations between universal future properties of $\mykripke$ and $\mykripke'$
are given in Figure~\ref{fig:weak-bisim-I:AX}. 

We prove two supporting lemmas before giving the proof of soundness of the inference rules.
\begin{lemma}
\label{lmm:wb-back-forth-derived}
If $\mykripke \Kwbisim_{\langle R,P,S \rangle} \mykripke'$ then $R \Mm \Tr{R} \Mm \vec{z'} \leq \vec{z'}$.
\end{lemma}
\begin{proof}
From Definition~\ref{def:weak-bisim}, 
$\Tr{R} \Mm \vec{z'} = \Tr{R} \Mm (\vec{z'} \St \Ng{\newstates})$. 
Multiply both sides by $R$. Again, from Definition~\ref{def:weak-bisim},
$R \Mm \Tr{R} \Mm \vec{z'} = R \Mm \Tr{R} \Mm (\vec{z'} \St \Ng{\newstates}) = \vec{z'} \St \Ng{\newstates}$. 
Further, $\vec{z'} \St \Ng{\newstates} \leq \vec{z'}$. Hence, $R \Mm \Tr{R} \Mm \vec{z'} \leq \vec{z'}$.
\hfill
\end{proof}

\begin{lemma}
\label{lmm:wb-back-forth}
If $\mykripke \Kwbisim_{\langle R,P,S \rangle} \mykripke'$ then
$\Tr{R} \Mm \mat' \Mm R \Pl \Tr{R} \Mm R \Mm P \Mm S \Mm \Tr{R} \Mm R = \Tr{R} \Mm (R \Mm \mat \Mm \Tr{R}) \Mm R$.
\end{lemma}
\begin{proof}
\[
\renewcommand{\Mm}{\mbox{\ensuremath{\cdot}}}
\setstretch{1.2}
\begin{array}{clcl@{\hspace{5mm}}r}
& (\mat' \Mn R \Mm P \Mn S \Mm \Tr{R}) \Pl R \Mm P \Mm S \Mm \Tr{R} & = & R \Mm \mat \Mm \Tr{R} & 
\comment{1}{Definition~\ref{def:weak-bisim}}\pnl
& \Tr{R} \Mm ((\mat' \Mn R \Mm P \Mn S \Mm \Tr{R}) \Pl R \Mm P \Mm S \Mm \Tr{R}) \Mm R & = & 
\Tr{R} \Mm (R \Mm \mat \Mm \Tr{R}) \Mm R & \mycomment{1}{monotonic}{}\pnl
& \Tr{R} \Mm \mat' \Mm R \Pl \Tr{R} \Mm R \Mm P \Mm S \Mm \Tr{R} \Mm R & = & 
\Tr{R} \Mm (R \Mm \mat \Mm \Tr{R}) \Mm R & \mycomment{1}{distributive, Definition~\ref{def:weak-bisim}}{}
\end{array}
\]
\hfill
\end{proof}

\begin{theorem}
\label{thm:wb-ax}
If $\langle R,P,S \rangle$ is a weak bisimulation between Kripke structures $\mykripke$ and $\mykripke'$
i.e. $\mykripke \Kwbisim_{\langle R,P,S \rangle} \mykripke'$ and $\vec{\newstates} = \vecone \Mn R \Mm \vecone$ 
is the set of internal nodes of $\mykripke'$ then the correlations between future temporal properties
given in Figure~\ref{fig:weak-bisim-I:AX} are sound.
\end{theorem}
\begin{proof}
We first determine the nature of the relation $R$ from the definition of weak bisimulation relations
and the premises common in the rules of Figure~\ref{fig:weak-bisim-I:AX}.

From Definition~\ref{def:weak-bisim}, we know that $\vec{z} = \Tr{R} \Mm R \Mm \vec{z}$.
Thus, $R$ is onto but neither one-to-many nor many-to-many ref. (\ref{eqn:second-R}).
From~(\ref{eqn:second-R}), we know that $\Ng{\Tr{R} \Mm \vec{z'}} \leq \Tr{R} \Mm \Ng{\vec{z'}}$
states that $R$ is onto and possibly many-to-one or many-to-many.
Further, $\Ng{R \Mm \vec{z}} \St \Ng{\vec{\newstates}} = R \Mm \Ng{\vec{z}}$ implies that
$\Ng{R \Mm \vec{z}} \geq R \Mm \Ng{\vec{z}}$. From~(\ref{eqn:first-R}), we know that 
$\Ng{R \Mm \vec{z}} \geq R \Mm \Ng{\vec{z}}$ states that $R$ is not total and neither one-to-many nor many-to-many.
Thus, $R$ is partial, onto, and either one-to-one or many-to-one.

Further, from Definition~\ref{def:weak-bisim},
$R \Mm \mat \Mm \Tr{R} = (\mat' \Mn R \Mm P \Mn S \Mm \Tr{R}) \Pl R \Mm P \Mm S \Mm \Tr{R}$.
Hence, for every edge $\pair{p}{q}$ in $\mykripke$, in $\mykripke'$ either
(a)~there is an edge $\pair{p'}{q'}$
where $p'$ and $p$ as $q'$ and $q$ are related by R 
or (b)~a pair of adjacent edges $\pair{p'}{r'}$ and $\pair{r'}{q'}$
s.t. $p'$ and $p$ as well as $q'$ and $q$ are related by $R$ and
$r'$ is a internal node. Similarly, for every edge or a pair of adjacent edges in $\mykripke'$ there is an edge
in $\mykripke$.

The proof of the rule WB:AX is as follows:
\[
\setstretch{1.2}
\begin{array}{rlcl@{\hspace{15mm}}r}
& R \Mm \vec{\varphi} & \leq & \vec{\varphi'} & \mycomment{2}{From (\ref{eqn:mono-6})}{given}\pnl
& \vec{\newstates} & \leq & \vec{\varphi'} & \mycomment{2}{From (\ref{eqn:mono-6})}{given}\pnl
& R \Mm \vec{\varphi} \;\Pl\; \vec{\newstates} & \leq & \vec{\varphi'} & 
\mycomment{2}{From (\ref{eqn:mono-6})}{monotonic, identity}\pnl
& \Ng{\vec{\varphi'}} & \leq & \Ng{R \Mm \vec{\varphi}} \;\St\; \Ng{\vec{\newstates}} = R \Mm \Ng{\vec{\varphi}} &
\mycomment{2}{From (\ref{eqn:misc-1}); Lemma \ref{lmm:node-add-negation}}{monotonic, given}\pnl
& \Ng{\vec{\varphi'}} \St \Ng{\vec{\newstates}} & \leq & R \Mm \Ng{\vec{\varphi}} & 
\mycomment{1}{$\Ng{\vec{\varphi'}} \leq \Ng{\vec{\newstates}}$ 
hence $\Ng{\vec{\varphi'}} \St \Ng{\vec{\newstates}} = \Ng{\vec{\varphi'}}$}{}\pnl
& \Tr{R} \Mm (\Ng{\vec{\varphi'}} \St \Ng{\vec{\newstates}}) & \leq & 
\Tr{R} \Mm R \Mm \Ng{\vec{\varphi}} = \Ng{\vec{\varphi}} & \mycomment{2}{}{monotonic, Definition~\ref{def:weak-bisim}}\pnl
& \Tr{R} \Mm \mat' \Mm R \Mm (\Tr{R} \Mm (\Ng{\vec{\varphi'}} \St \Ng{\vec{\newstates}})) & \leq & 
\Tr{R} \Mm (R \Mm \mat \Mm \Tr{R}) \Mm R \Mm \Ng{\vec{\varphi}} & \mycomment{2}{}{Lemma~\ref{lmm:wb-back-forth}, monotonic}\pnl
& \Tr{R} \Mm \mat' \Mm (R \Mm \Tr{R} \Mm (\Ng{\vec{\varphi'}} \St \Ng{\vec{\newstates}})) & \leq & 
\Tr{R} \Mm R \Mm \mat \Mm (\Tr{R} \Mm R \Mm \Ng{\vec{\varphi}}) & \mycomment{2}{}{associative}\pnl
& \Tr{R} \Mm \mat' \Mm (\Ng{\vec{\varphi'}} \St \Ng{\vec{\newstates}}) & \leq & 
\Tr{R} \Mm R \Mm (\mat \Mm \Ng{\vec{\varphi}}) & \mycomment{2}{}{Definition~\ref{def:weak-bisim}, associative}\pnl
& \Tr{R} \Mm \mat' \Mm \Ng{\vec{\varphi'}} & \leq & 
\mat \Mm \Ng{\vec{\varphi}} & 
\mycomment{2}{}{Definition~\ref{def:weak-bisim}, $\Ng{\vec{\varphi'}} \St \Ng{\vec{\newstates}} = \Ng{\vec{\varphi'}}$}\pnl
& \Ng{\mat \Mm \Ng{\vec{\varphi}}} & \leq & \Ng{\Tr{R} \Mm \mat' \Mm \Ng{\vec{\varphi'}}} 
\leq \Tr{R} \Mm \Ng{\mat' \Mm \Ng{\vec{\varphi'}}} & \mycomment{1}{negation, given}{}\pnl
& R \Mm \Ng{\mat \Mm \Ng{\vec{\varphi}}} & \leq & R \Mm \Tr{R} \Mm \Ng{\mat' \Mm \Ng{\vec{\varphi'}}} & \mycomment{1}{monotonic}{}\pnl
& R \Mm \Ng{\mat \Mm \Ng{\vec{\varphi}}} & \leq & \Ng{\mat' \Mm \Ng{\vec{\varphi'}}} & 
\mycomment{2}{}{Lemma~\ref{lmm:wb-back-forth-derived}}
\end{array}
\]

The proof of the rule WB:AU is as follows:
\[
\renewcommand{\Mm}{\mbox{\ensuremath{\cdot}}}
\setstretch{1.2}
\begin{array}{@{}r@{\;}l@{}c@{}l@{\hspace{10mm}}r@{}}
& \Ng{R \Mm \vec{z}} \;\St\; \Ng{\vec{\newstates}} & = & R \Mm \Ng{\vec{z}} &
\mycomment{2}{Lemma \ref{lmm:node-add-negation}}{given}\pnl
& \Tr{R} \Mm (\Ng{R \Mm \vec{z}} \;\St\; \Ng{\vec{\newstates}}) & = & \Tr{R} \Mm R \Mm \Ng{\vec{z}} = \Ng{\vec{z}} &
\mycomment{2}{Lemma \ref{lmm:node-add-negation}}{Def.\ref{def:weak-bisim}}\pnl
& (\Tr{R} \Mm \mat' \Mm R \Pl \Tr{R} \Mm R \Mm P \Mm S \Mm \Tr{R} \Mm R) \Mm \Tr{R} \Mm (\Ng{R \Mm \vec{z}} \;\St\; \Ng{\vec{\newstates}}) & = & \Tr{R} \Mm R \Mm \mat \Mm \Tr{R} \Mm R \Mm \Ng{\vec{z}} &
\mycomment{2}{Lemma \ref{lmm:node-add-negation}}{Lemma~\ref{lmm:wb-back-forth}}\pnl
& \Tr{R} \Mm \mat' \Mm R \Mm \Tr{R} \Mm (\Ng{R \Mm \vec{z}} \;\St\; \Ng{\vec{\newstates}}) \Pl 
\Tr{R} \Mm R \Mm P \Mm S \Mm \Tr{R} \Mm R \Mm \Tr{R} \Mm (\Ng{R \Mm \vec{z}} \;\St\; \Ng{\vec{\newstates}}) & = & 
\mat \Mm \Ng{\vec{z}} &
\mycomment{2}{Lemma \ref{lmm:node-add-negation}}{distributive, Def.~\ref{def:weak-bisim}}\pnl
& \Tr{R} \Mm \mat' \Mm (\Ng{R \Mm \vec{z}} \;\St\; \Ng{\vec{\newstates}}) \Pl 
\Tr{R} \Mm (R \Mm P) \Mm S \Mm \Tr{R} \Mm (\Ng{R \Mm \vec{z}} \;\St\; \Ng{\vec{\newstates}}) & \leq & 
\mat \Mm \Ng{\vec{z}} &
\mycomment{2}{Lemma \ref{lmm:node-add-negation}}{assoc., Definition~\ref{def:weak-bisim}}\pnl
& \Tr{R} \Mm \mat' \Mm (\Ng{R \Mm \vec{z}} \;\St\; \Ng{\vec{\newstates}}) \Pl 
\Tr{R} \Mm (R \Mm P) \Mm S \Mm \Tr{R} \Mm (\Ng{R \Mm \vec{z}} \;\St\; \Ng{\vec{\newstates}}) & \leq & 
\mat \Mm \Ng{\vec{z}} &
\mycomment{2}{Lemma \ref{lmm:node-add-negation}}{Definition~\ref{def:weak-bisim}}\pnl
& \Tr{R} \Mm \mat' \Mm (\Ng{R \Mm \vec{z}} \;\St\; \Ng{\vec{\newstates}}) \Pl 
\Tr{R} \Mm (R \Mm P) \Mm S \Mm \Tr{R} \Mm (\vecone \St (\Ng{R \Mm \vec{z}} \;\St\; \Ng{\vec{\newstates}})) & \leq & 
\mat \Mm \Ng{\vec{z}} &
\mycomment{2}{Lemma \ref{lmm:node-add-negation}}{identity}\pnl
& \Tr{R} \Mm \mat' \Mm (\Ng{R \Mm \vec{z}} \;\St\; \Ng{\vec{\newstates}}) \Pl 
\Tr{R} \Mm (R \Mm P) \Mm (S \Mm \Tr{R} \Mm \vecone \St S \Mm \Tr{R} \Mm (\Ng{R \Mm \vec{z}} \;\St\; \Ng{\vec{\newstates}})) & \leq & 
\mat \Mm \Ng{\vec{z}} &
\mycomment{2}{Lemma \ref{lmm:node-add-negation}}{distributive}\pnl
& \Tr{R} \Mm \mat' \Mm (\Ng{R \Mm \vec{z}} \;\St\; \Ng{\vec{\newstates}}) \Pl 
\Tr{R} \Mm (R \Mm P) \Mm (\vec{\newstates} \St (S \Mm \Tr{R}) \Mm (\Ng{R \Mm \vec{z}} \;\St\; \Ng{\vec{\newstates}})) & \leq & 
\mat \Mm \Ng{\vec{z}} &
\mycomment{2}{Lemma \ref{lmm:node-add-negation}}{Definition~\ref{def:weak-bisim}}\pnl
& \Ng{\Tr{R} \Mm \mat' \Mm (\Ng{R \Mm \vec{z}} \;\St\; \Ng{\vec{\newstates}})} \St
\Ng{\Tr{R} \Mm \mat' \Mm (\vec{\newstates} \St \mat' \Mm (\Ng{R \Mm \vec{z}} \;\St\; \Ng{\vec{\newstates}}))} & \geq & 
\Ng{\mat \Mm \Ng{\vec{z}}} &
\mycomment{2}{}{negation}\pnl
& \Tr{R} \Mm \Ng{\mat' \Mm (\Ng{R \Mm \vec{z}} \;\St\; \Ng{\vec{\newstates}})} \St
\Tr{R} \Mm \Ng{\mat' \Mm (\vec{\newstates} \St \mat' \Mm (\Ng{R \Mm \vec{z}} \;\St\; \Ng{\vec{\newstates}}))} & \geq & 
\Ng{\mat \Mm \Ng{\vec{z}}} &
\mycomment{2}{}{given}\pnl
& R \Mm \Tr{R} \Mm \Ng{\mat' \Mm (\Ng{R \Mm \vec{z}} \;\St\; \Ng{\vec{\newstates}})} \St
R \Mm \Tr{R} \Mm \Ng{\mat' \Mm (\vec{\newstates} \St \mat' \Mm (\Ng{R \Mm \vec{z}} \;\St\; \Ng{\vec{\newstates}}))} & \geq & 
R \Mm \Ng{\mat \Mm \Ng{\vec{z}}} &
\mycomment{2}{}{mono., distributive}\pnl
& \Ng{\mat' \Mm (\Ng{R \Mm \vec{z}} \;\St\; \Ng{\vec{\newstates}})} \;\St\;
\Ng{\mat' \Mm (\vec{\newstates} \St \mat' \Mm (\Ng{R \Mm \vec{z}} \;\St\; \Ng{\vec{\newstates}}))} & \geq & 
\Ng{R \Mm \mat \Mm \Ng{\vec{z}}} &
\mycomment{2}{}{Lemma~\ref{lmm:wb-back-forth-derived}}\pnl
(i) & \Ng{\mat' \Mm (\Ng{R \Mm \vec{z}} \;\St\; \Ng{\vec{\newstates}})} \;\St\;
\Ng{\mat' \Mm (\vec{\newstates} \St \mat' \Mm (\Ng{R \Mm \vec{z}} \;\St\; \Ng{\vec{\newstates}}))} & \geq & 
R \Mm \Ng{\mat \Mm \Ng{\vec{z}}} &
\mycomment{2}{Lemma \ref{lmm:node-add-negation}}
{$\Ng{R \Mm \vec{z}} \geq R \Mm \Ng{\vec{z}}$}\pnl
& \vec{\psi'} \Pl (\vec{\varphi'} \St \Ng{\mat' \Mm (\Ng{R \Mm \vec{z}} \;\St\; \Ng{\vec{\newstates}})} \;\St\;
\Ng{\mat' \Mm (\vec{\newstates} \St \mat' \Mm (\Ng{R \Mm \vec{z}} \;\St\; \Ng{\vec{\newstates}}))}) & \geq & 
R \Mm (\vec{\psi} \Pl (\vec{\varphi} \St \Ng{\mat \Mm \Ng{\vec{z}}})) &
\mycomment{2}{Lemma \ref{lmm:node-add-negation}}{monotonic}\pnl
(ii) & 
\underbrace{\lambda \vec{z'}. \vec{\psi'} \Pl (\vec{\varphi'} \St \Ng{\mat' \Mm (\Ng{\vec{z'}} \;\St\; \Ng{\vec{\newstates}})} \;\St\;
\Ng{\mat' \Mm (\vec{\newstates} \St \mat' \Mm (\Ng{\vec{z'}} \;\St\; \Ng{\vec{\newstates}}))})}_{g} & \sqsupseteq_R & 
\lambda \vec{z}. \vec{\psi} \Pl (\vec{\varphi} \St \Ng{\mat \Mm \Ng{\vec{z}}}) &
\mycomment{2}{Lemma \ref{lmm:node-add-negation}}{Definition~\ref{def:fun-space}}\pnl
(iii) & \mu \vec{z'}. \vec{\psi'} \Pl (\vec{\varphi'} \St \Ng{\mat' \Mm (\Ng{\vec{z'}} \;\St\; \Ng{\vec{\newstates}})} \;\St\;
\Ng{\mat' \Mm (\vec{\newstates} \St \mat' \Mm (\Ng{\vec{z'}} \;\St\; \Ng{\vec{\newstates}}))}) & \geq & 
R \Mm (\mu \vec{z}. \vec{\psi} \Pl (\vec{\varphi} \St \Ng{\mat \Mm \Ng{\vec{z}}})) &
\mycomment{2}{Lemma \ref{lmm:node-add-negation}}{Lemma~\ref{lmm:ns-lfp}}
\end{array}
\]

We now show that $\mu \vec{z'}. g(\vec{z'}) \leq \mu \vec{z'}. \vec{\psi'} \Pl \vec{\varphi'} \St \Ng{\mat' \Mm \Ng{\vec{z'}}}$.
Let $\vec{q'} = \mu \vec{z'}. \vec{\psi'} \Pl \vec{\varphi'} \St \Ng{\mat' \Mm \Ng{\vec{z'}}}$.
\[
\setstretch{1.2}
\begin{array}{rlcl@{\hspace{3mm}}r}
& \vec{q'} & \geq & \vec{\varphi'} \St \Ng{\mat' \Mm \Ng{\vec{q'}}} & \mycomment{1}{Definition of $\vec{q}$}{}\pnl 
& \Ng{\vec{q'}} & \leq & \Ng{\vec{\varphi'}} \Pl \mat' \Mm \Ng{\vec{q'}} & 
\mycomment{2}{From~(\ref{eqn:misc-1})}{negation}\pnl 
& \Ng{\vec{q'}} \St \vec{\newstates} & \leq &
(\vec{\newstates} \St \Ng{\vec{\varphi'}}) \Pl (\vec{\newstates} \St \mat' \Mm \Ng{\vec{q'}}) 
& \mycomment{2}{From~(\ref{eqn:mono-5})}{monotonic}\pnl 
& \Ng{\vec{q'}} \St \vec{\newstates} & \leq &
\veczero \Pl (\vec{\newstates} \St \mat' \Mm \Ng{\vec{q'}}) 
& \mycomment{1}{$\vec{\newstates} \leq \vec{\varphi'}$}{}\pnl 
(iv) & \Ng{\vec{q'}} \St \vec{\newstates} & \leq & 
\vec{\newstates} \St \mat' \Mm (\Ng{\vec{q'}} \St \Ng{\vec{\newstates}})
& \mycomment{1}{Definition~\ref{def:weak-bisim}}{}
\end{array}
\]

Further,
\[
\setstretch{1.2}
\begin{array}{rlclr}
& (\Ng{\vec{q'}} \St \Ng{\vec{\newstates}}) \Pl (\Ng{\vec{q'}} \St \vec{\newstates}) & = & \Ng{\vec{q}} & \pnl
& (\Ng{\vec{q'}} \St \Ng{\vec{\newstates}}) \Pl (\vec{\newstates} \St \mat' \Mm (\Ng{\vec{q'}} \St \Ng{\vec{\newstates}}))
& \geq & \Ng{\vec{q'}} & \mycomment{1}{from~($iv$)}{}\pnl
& \mat' \Mm (\Ng{\vec{q'}} \St \Ng{\vec{\newstates}}) \Pl \mat' \Mm 
(\vec{\newstates} \St \mat' \Mm (\Ng{\vec{q'}} \St \Ng{\vec{\newstates}}))
& \geq & \mat' \Mm \Ng{\vec{q'}} & 
\mycomment{2}{From (\ref{eqn:mono-4}) and (\ref{eqn:distr-5})}{monotonic}\pnl
& \Ng{\mat' \Mm (\Ng{\vec{q'}} \St \Ng{\vec{\newstates}})} \;\St\; 
\Ng{\mat' \Mm (\vec{\newstates} \St \mat' \Mm (\Ng{\vec{q'}} \St \Ng{\vec{\newstates}}))}
& \leq & \Ng{\mat' \Mm \Ng{\vec{q'}}} & 
\mycomment{2}{From (\ref{eqn:misc-1}) and (\ref{eqn:de-morgan-2})}{negate; distribute}\pnl
& \vec{\psi'} \Pl
(\vec{\varphi'} \;\St\;
\Ng{\mat' \Mm (\Ng{\vec{q'}} \St \Ng{\vec{\newstates}})} \;\St\; 
\Ng{\mat' \Mm (\vec{\newstates} \St \mat' \Mm (\Ng{\vec{q'}} \St \Ng{\vec{\newstates}}))})
& \leq & \vec{\psi'} \Pl (\vec{\varphi'} \St \Ng{\mat' \Mm \Ng{\vec{q'}}}) & 
\mycomment{1}{monotonic}{}\pnl
& \vec{\psi'} \Pl
(\vec{\varphi'} \;\St\;
\Ng{\mat' \Mm (\Ng{\vec{q'}} \St \Ng{\vec{\newstates}})} \;\St\; 
\Ng{\mat' \Mm (\vec{\newstates} \St \mat' \Mm (\Ng{\vec{q'}} \St \Ng{\vec{\newstates}}))})
& \leq & \vec{q'} & 
\mycomment{1}{Def. of $\vec{q'}$}{}
\end{array}
\]

Thus, $\vec{q'}$ is a post-fixed point of $g$. Hence, $\mu \vec{z'}. g(\vec{z'}) \leq \vec{q'}$.
From ($iii$) above, by transitivity,
\[
R \Mm (\mu \vec{z}. \vec{\psi} \Pl (\vec{\varphi} \St \Ng{\mat \Mm \Ng{\vec{z}}})) \leq 
\mu \vec{z'}. \vec{\psi'} \Pl (\vec{\varphi'} \St \Ng{\mat' \Mm \Ng{\vec{z'}}})
\]
The proofs of WB:AW and WB:AG can be derived similarly.
\hfill
\end{proof}

\section{Conclusions}
\label{sec:conclusions}

Program transformations are used routinely in many software development and maintenance activities.
Temporal logics provide an expressive and powerful framework to specify program properties.
In this paper, we have considered the problem of correlating temporal properties across program transformations.
We have developed a logic called temporal transformation logic and presented inference rules for a
comprehensive set of primitive program transformations.
This formulation enables us to deductively verify soundness of program transformations
when the soundness conditions of individual transformations can be captured as temporal logic formulae.
In particular, we have considered its application to verification of compiler optimizations. 
We have made novel use of boolean matrix algebra in defining the transformation primitives and proving
soundness of the inference rules of the logic.

Future work involves applying the logic to verification
problems arising in software engineering due to the use of program transformations.
On the theoretical front, we plan to study completeness of the logic and 
its extension to CTL* and modal mu-calculus. 
The present formulations are applicable to intra-procedural temporal properties only.
In the inter-procedural framework,
we need to consider more expressive models of programs like 
recursive state machines~\cite{recursive-state-machines}
or pushdown systems~\cite{RepsHS95}. 
For inter-procedural analysis, temporal logics have been extended
with matching calls and returns~\cite{AlurEM04}.
Extending the logic to inter-procedural setting is a challenging research direction.

\bibliographystyle{alpha}
\bibliography{Bibliography}

\appendix

\section{Computational tree logic with branching past}
\label{sec:ctlbp}

The syntax of $\ctlbp$ formulae is as follows:
\begin{equation*}
\label{eqn:ctl-syntax}
\setstretch{1.2}
\begin{array}{ccl}
\varphi & \!\!:=\!\! & \top \sthat \bot \sthat p \sthat \neg \varphi \sthat \varphi \Imp \varphi \sthat 
\varphi \vee \varphi \sthat \varphi \wedge \varphi \sthat \\
& & 
\EX(\varphi) \sthat \EU(\varphi,\varphi) \sthat \EW(\varphi,\varphi) \sthat \EF(\varphi) \sthat 
\EG(\varphi) \sthat\\
&&
\EY(\varphi) \sthat \ES(\varphi,\varphi) \sthat \EP(\varphi) \sthat \EH(\varphi) \sthat\\
& & 
\AX(\varphi) \sthat \AU(\varphi,\varphi) \sthat \AW(\varphi,\varphi) \sthat \AF(\varphi) \sthat
\AG(\varphi) \sthat\\
&&
\AY(\varphi) \sthat \AS(\varphi,\varphi) \sthat \AP(\varphi) \sthat \AH(\varphi)
\end{array}
\end{equation*}
where $\top$ and $\bot$ are special propositions denoting true and false values and
$p$ is an atomic proposition.

$\ctlbp$ is interpreted over models with branching past and branching future. Future is infinite
whereas past is finite.
Models of $\ctlbp$ formulae are nodes of Kripke structures.
For a formula $\varphi$ and a node~$i$ of a Kripke structure $\mykripke = (\graph,\myprops,\mylabeling)$,
we write $\mykripke, i \models \varphi$ to denote that the formula $\varphi$ holds at the node~$i$.
With the boolean matrix algebraic semantics, 
$
\mykripke, i \models \varphi \text{ iff } \left[\vec{\varphi}\right]_i = \one
$.

\newcommand{\valuation}{V}

Consider a valuation function $\valuation$ from $\ctlbp$ formulae to boolean vectors and
denote $\valuation(\varphi)$ by $\vec{\varphi}$. The valuation of an atomic proposition
$p$ is $\valuation(p) = \mylabeling(p) = \vec{p}$. 

To define semantics of $\ctlbp$ operators, we require
the notion of fixed points over boolean vector domains.
Consider a complete lattice $D = \langle \bool, \leq \rangle$ where $\bool = \{\zero,\one\}$
is a set of boolean values with a partial order $\zero \leq \one$.
A complete lattice $D_n = \langle \bool_n, \leq_n \rangle$ is a cartesian product of 
$n$ $D$-lattices. Consider a monotonic function $f: D^n \rightarrow D^n$.
The least fixed point of $f$ is denoted by $\mu \vec{z} . f(\vec{z})$ and
the greatest fixed point of $f$ is denoted by $\nu \vec{z} . f(\vec{z})$.
As a convention, we do not write the subscript of the partial order $\leq$ when used for vectors.

\begin{definition}
\label{def:ctlbp-future-semantics}
Consider a Kripke structure~$\mykripke$.
Let $\valuation(\top) = \vecone$, $\valuation(\bot) = \veczero$, and
$\valuation(p) = \vec{p}$.
\begin{equation*}
\setstretch{1.2}
\begin{array}{lcl@{\hspace{8mm}}r}
\valuation\!\left(\neg \varphi\right) & \;\defas\; & \Ng{\vec{\varphi}} & \text{negation}\\
\valuation\!\left(\varphi \Imp \psi\right) & \defas & \Ng{\vec{\varphi}} \;\Pl\; \vec{\psi} &\text{implication}\\
\valuation\!\left(\varphi \vee \psi\right) & \defas & \vec{\varphi} \;\Pl\; \vec{\psi} &\text{disjunction}\\
\valuation\!\left(\varphi \wedge \psi\right) & \defas & \vec{\varphi} \;\St\; \vec{\psi} &\text{conjunction}\\
\valuation\!\left(\EX(\varphi)\right) & \defas & \mat \Mm \vec{\varphi} & \text{Exists neXt}\\
\valuation\!\left(\EU(\varphi,\psi)\right) & \defas & 
\mu \vec{z}.\; \vec{\psi} \;\Pl\; (\vec{\varphi} \;\St\; \mat \Mm \vec{z}) & \text{Exists Until}\\ 
\valuation\!\left(\EW(\varphi,\psi)\right) & \defas & 
\nu \vec{z}.\; \vec{\psi} \;\Pl\; (\vec{\varphi} \;\St\; \mat \Mm \vec{z}) & \text{Exists Weak until}\\
\valuation\!\left(\EF(\varphi)\right) & \defas & 
\mu \vec{z}.\; \vec{\varphi} \;\Pl\; (\vecone \;\St\; \mat \Mm \vec{z}) & \text{Exists in Future}\\ 
\valuation\!\left(\EG(\varphi)\right) & \defas & \nu \vec{z}.\; \vec{\varphi} \;\St\; \mat \Mm \vec{z} & \text{Exists Globally}\\
\valuation\!\left(\AX(\varphi)\right) & \defas & \Ng{\mat \Mm \Ng{\vec{\varphi}}} & \text{forAll neXt}\\
\valuation\!\left(\AU(\varphi,\psi)\right) & \defas & 
\mu \vec{z}.\; \vec{\psi} \;\Pl\; (\vec{\varphi} \;\St\; \Ng{\mat \Mm \Ng{\vec{z}}}) & \text{forAll Until}\\ 
\valuation\!\left(\AW(\varphi,\psi)\right) & \defas & 
\nu \vec{z}.\; \vec{\psi} \;\Pl\; (\vec{\varphi} \;\St\; \Ng{\mat \Mm \Ng{\vec{z}}}) & \text{forAll Weak until}\\
\valuation\!\left(\AF(\varphi)\right) & \defas & 
\mu \vec{z}.\; \vec{\varphi} \;\Pl\; (\vecone \;\St\; \Ng{\mat \Mm \Ng{\vec{z}}}) & \text{forAll in Future}\\ 
\valuation\!\left(\AG(\varphi)\right) & \defas & \nu \vec{z}.\; \vec{\varphi} \;\St\; \Ng{\mat \Mm \Ng{\vec{z}}}
& \text{forAll Globally}
\end{array}
\end{equation*}
\end{definition}

To define semantics of past operators, we consider
the transpose of the adjacency matrix $\mat$ which corresponds to the inverted edges of the Kripke structure.

\begin{definition}
\label{def:ctlbp-past-semantics}
Consider a Kripke structure $\mykripke$.
\begin{equation*}
\setstretch{1.2}
\begin{array}{lcl@{\hspace{8mm}}r}
\valuation\!\left(\EY(\varphi)\right) & \;\defas\; & \Tr{\mat} \Mm \vec{\varphi} & \text{Exists Yesterday}\\
\valuation\!\left(\ES(\varphi,\psi)\right) & \defas & 
\mu \vec{z}.\; \vec{\psi} \;\Pl\; (\vec{\varphi} \;\St\; \Tr{\mat} \Mm \vec{z}) & \text{Exists Since}\\
\valuation\!\left(\EP(\varphi)\right) & \defas & 
\mu \vec{z}.\; \vec{\varphi} \;\Pl\; (\vecone \;\St\; \Tr{\mat} \Mm \vec{z}) & \text{Exists in Past}\\
\valuation\!\left(\EH(\varphi)\right) & \defas & \nu \vec{z}.\ \vec{\varphi} \;\St\; \Tr{\mat} \Mm \vec{z} & 
\text{Exists in History}\\
\valuation\!\left(\AY(\varphi)\right) & \defas & \Ng{\Tr{\mat} \Mm \Ng{\vec{\varphi}}} \;\Mn\; \entry &
\text{forAll Yesterday}\\
\valuation\!\left(\AS(\varphi,\psi)\right) & \defas & 
\nu \vec{z}.\; \vec{\psi} \;\Pl\; \left(\vec{\varphi} \;\St\; \left(\Ng{\Tr{\mat} \Mm \Ng{\vec{z}}} \Mn \entry\right)\right) &
\text{forAll Since}\\
\valuation\!\left(\AP(\varphi)\right) & \defas & 
\nu \vec{z}.\; \vec{\varphi} \;\Pl\; \left(\vecone \;\St\; \left(\Ng{\Tr{\mat} \Mm \Ng{\vec{z}}} \Mn \entry\right)\right) &
\text{forAll in Past}\\
\valuation\!\left(\AH(\varphi)\right) & \defas & 
\nu \vec{z}.\ \vec{\varphi} \;\St\; \left(\Ng{\Tr{\mat} \Mm \Ng{\vec{z}}} \Mn \entry\right) & \text{forAll in History}\\
\end{array}
\end{equation*}
\end{definition}

Since the entry node does not have any predecessors, an $\AY$ formula should not hold at the entry node
denoted by vector $\entry$.
Recall that in $\ctlbp$, past is finite. Therefore, $\AS$ 
is defined as the greatest fixed point. The greatest fixed point subsumes
the infinite paths of loops in the Kripke structure. However, since
$\AY$ does not hold at the entry node, $\psi$ has to hold before or at the
entry node. This gives the required finitary semantics to $\AS$.

\section{Fixed points over different vector spaces}
\label{sec:vector-space}

Our aim is to define correlations between temporal formulae interpreted on different Kripke structures.
The Kripke structures can have different number of nodes. 
Some temporal operators are given fixed point semantics. Therefore, to correlate such formulae
we need a correlation between fixed points of functions over vector spaces of different sizes.

Consider a complete lattice $D_n = \langle \bool_n, \leq_n, \bot_n, \top_n \rangle$ 
of boolean vectors of size $n$ 
where $\bot_n = \veczero_n$ and $\top_n = \vecone_n$ are the top and bottom elements.
Let $D_m = \langle \bool_m, \leq_m, \bot_m, \top_m \rangle$ be another complete lattice
of boolean vectors of size $m$. 
Let $R$ be an $(m \Tm n)$ boolean matrix 
correlating elements of vectors in $D_m$ and $D_n$.
Consider monotonic functions $f: D_n \rightarrow D_n$ and 
$f': D_m \rightarrow D_m$. We first define an ordering between $f$ and $f'$ with respect to $R$.

\begin{definition}
\label{def:fun-space}
$
f \sqsubseteq_{R} f' \text{ iff \ }
\forall \vec{z} \in D_n: R \Mm f(\vec{z}) \leq_m f'(R \Mm \vec{z})
$
\end{definition}

Recall that an $(m \Tm n)$ relation $R$ can be seen as a function $R : D_n \rightarrow D_m$.
We now determine the correlations between the least and the greatest fixed
points of functions $f$ and $f'$ s.t. $f \sqsubseteq_{R} f'$.

\begin{lemma}
\label{lmm:ns-lfp}
If $f \sqsubseteq_{R} f'$ then
$
R \Mm (\mu \vec{z}. f(\vec{z})) \leq_m \mu \vec{z'}. f'(\vec{z'})
$.
\end{lemma}
\begin{proof}
By Kleene's fixed point theorem~\cite{kleene}, 
$\mu \vec{z}. f(\vec{z}) \defas f^a(\bot_n)$ for some $a \geq 1$
and $\mu \vec{z'}. f'(\vec{z'}) \defas f^b(\bot_m)$ for some $b \geq 1$.


We first prove by induction that $\forall k \in \mathcal{N}: R \Mm f^k(\bot_n) \leq_m f'^k(\bot_m)$.

\vspace*{1.5mm}\noindent{\bf Base case} $[k = 1]$:
From Definition~\ref{def:fun-space}, $R \Mm f(\bot_n) \leq_m f'(R \Mm \bot_n)$.
Since $\bot_n = \veczero_n$ and $\bot_m = \veczero_m$, 
$R \Mm \bot_n = \bot_m$. Hence, $R \Mm f(\bot_n) \leq_m f'(\bot_m)$.

\vspace*{1.5mm}\noindent{\bf Induction hypothesis} $[k \geq 1]$:
Let $R \Mm f^k(\bot_n) \leq_m f'^k(\bot_m)$.

\vspace*{1.5mm}\noindent{\bf Inductive case}:
From Definition~\ref{def:fun-space}, $R \Mm f(f^k(\bot_n)) \leq_m f'(R \Mm f^k(\bot_n))$.
From the induction hypothesis and monotonicity of $f'$,
$f'(R \Mm f^k(\bot_n)) \leq_m f'(f'^k(\bot_m)) = f'^{k+1}(\bot_m)$.
Hence, by transitivity, \mbox{$R \Mm f^{k+1}(\bot_n) \leq_m f'^{k+1}(\bot_m)$}.


Let $\mu \vec{z}. f(\vec{z}) = f^a(\bot_n)$ and 
$\mu \vec{z'}. f'(\vec{z'}) = f'^b(\bot_m)$ for $a \geq 1$ and $b \geq 1$.
We consider all the relations between $a$ and $b$
and show that $R \Mm f^a(\bot_n) \leq_m f'^b(\bot_m)$.

\vspace*{1.5mm}\noindent{\bf Case $[a < b]$}:
From {Part I}, $R \Mm f^a(\bot_n) \leq_m f'^a(\bot_m)$.
Since $f'$ is monotonic, we know that $\forall k \in \mathcal{N}: f'^k(\bot_m) \leq_m f'^{k+1}(\bot_m)$.
Since $a < b$, $f'^a(\bot_m) \leq_m f'^b(\bot_m)$. Hence, by transitivity,
$R \Mm f^a(\bot_n) \leq_m f'^b(\bot_m)$.


\vspace*{1.5mm}\noindent\underline{Case $[a \geq b]$}:
From {Part I}, $R \Mm f^a(\bot_n) \leq_m f'^a(\bot_m)$.
If $f'^k(\bot_m)$ is the least fixed point,
$\forall l \geq k: f'^l(\bot_m) = f'^k(\bot_m)$.
Since $f'^b(\bot_m)$ is the least fixed point and $a \geq b$,
$f'^a(\bot_m) = f'^b(\bot_m)$. Hence, by substituting, $R \Mm f^a(\bot_n) \leq_m f'^b(\bot_m)$.
\hfill
\end{proof}

Similarly, we prove the correlation between greatest fixed points stated below.
\begin{lemma}
\label{lmm:ns-gfp}
If $f \sqsubseteq_{R} f'$ then
$
R \Mm (\nu \vec{z}. f(\vec{z})) \leq_m \nu \vec{z'}. f'(\vec{z'})
$.
\end{lemma}

\end{document}